\documentclass[aps,preprintnumbers,nofootinbib,onecolumn, superscriptaddress]{revtex4}

\usepackage{graphicx}
\usepackage{pdfpages}
\usepackage[shortlabels]{enumitem}
\usepackage{mathtools}
\usepackage{fancyhdr}
\usepackage{nicefrac}
\usepackage{mathrsfs}
\usepackage{bbm}
\makeatletter

\usepackage{graphics}
\usepackage{amsfonts,amssymb,amsmath}            
\usepackage{ifthen}
\usepackage{amsthm, thm-restate}
\usepackage{dsfont}
\usepackage{float}                       
\usepackage{MnSymbol}

\usepackage[caption=false]{subfig}
\PassOptionsToPackage{hyphens}{url}
\usepackage[hyperfootnotes=false]{hyperref}
\usepackage{xcolor}
\definecolor{mylinkcolor}{rgb}{0,0,0.7} 
\hypersetup{unicode=true,%
  bookmarksnumbered=false,bookmarksopen=false,bookmarksopenlevel=1, %
  breaklinks=true,pdfborder={0 0 0},colorlinks=true}%
\hypersetup{%
  anchorcolor=mylinkcolor,citecolor=mylinkcolor, %
  filecolor=mylinkcolor,linkcolor=mylinkcolor, %
  menucolor=mylinkcolor,runcolor=mylinkcolor, %
  urlcolor=mylinkcolor}

\usepackage{cleveref}

\usepackage{tikz}
\usetikzlibrary{arrows}
\usetikzlibrary{arrows.meta}
\usepackage{rotating, xcolor}
\usetikzlibrary{shapes}
\usetikzlibrary{hobby}
\usetikzlibrary{decorations.markings}
\usetikzlibrary{arrows.meta}
\usetikzlibrary{snakes} 
\usepackage{tikz-cd}
\tikzset{degil/.style={
            decoration={markings,
            mark= at position 0.5 with {
                  \node[transform shape] (tempnode) {$\setminus$};
                  }
              },
              postaction={decorate}
}
}

\newcommand\longrsquigarrow{
\begin{tikzpicture}
\draw [decorate, decoration={zigzag, segment length=+6pt, amplitude=+.95pt,post length=+2pt}, arrows={-Classical TikZ Rightarrow}]  (0,0.1) -- (0.6,0.1); \draw[draw=none] (0,0)--(0.6,0);
\end{tikzpicture}
}

\newcommand\longlsquigarrow{
\begin{tikzpicture}
\draw [decorate, decoration={zigzag, segment length=+6pt, amplitude=+.95pt,post length=+2pt}, arrows={-Classical TikZ Rightarrow},rotate around={180:(0.3,0.1)}]  (0,0.1) -- (0.6,0.1); \draw[draw=none] (0,0)--(0.6,0);
\end{tikzpicture}
}

\makeatletter
\newcommand*{\da@rightarrow}{\mathchar"0\hexnumber@\symAMSa 4B }
\newcommand*{\da@leftarrow}{\mathchar"0\hexnumber@\symAMSa 4C }
\newcommand*{\xdashrightarrow}[2][]{%
  \mathrel{%
    \mathpalette{\da@xarrow{#1}{#2}{}\da@rightarrow{\,}{}}{}%
  }%
}
\newcommand{\xdashleftarrow}[2][]{%
  \mathrel{%
    \mathpalette{\da@xarrow{#1}{#2}\da@leftarrow{}{}{\,}}{}%
  }%
}
\newcommand*{\da@xarrow}[7]{%
  \sbox0{$\ifx#7\scriptstyle\scriptscriptstyle\else\scriptstyle\fi#5#1#6\m@th$}%
  \sbox2{$\ifx#7\scriptstyle\scriptscriptstyle\else\scriptstyle\fi#5#2#6\m@th$}%
  \sbox4{$#7\dabar@\m@th$}%
  \dimen@=\wd0 %
  \ifdim\wd2 >\dimen@
    \dimen@=\wd2 %
  \fi
  \count@=2 %
  \def\da@bars{\dabar@\dabar@}%
  \@whiledim\count@\wd4<\dimen@\do{%
    \advance\count@\@ne
    \expandafter\def\expandafter\da@bars\expandafter{%
      \da@bars
      \dabar@ 
    }%
  }%
  \mathrel{#3}%
  \mathrel{%
    \mathop{\da@bars}\limits
    \ifx\\#1\\%
    \else
      _{\copy0}%
    \fi
    \ifx\\#2\\%
    \else
      ^{\copy2}%
    \fi
  }%
  \mathrel{#4}%
}
\makeatother

\theoremstyle{plain}

\newtheorem{lemma}{Lemma}[section]
\newtheorem{corollary}{Corollary}[section]

\newtheorem{claim}{Claim}[section]

\theoremstyle{definition}
\newtheorem{definition}{Definition}[section]
\newtheorem{remark}{Remark}[section]
\newtheorem{example}{Example}[section]

\usepackage[nodayofweek]{datetime}
\newcommand{\comment}[1]{}
\newcommand{\ket}[1]{| #1 \rangle}

\newcommand{\di}{d}

\newcommand{\cA}{\mathcal{A}}
\newcommand{\cB}{\mathcal{B}}
\newcommand{\cC}{\mathcal{C}}

\newcommand{\cG}{\mathcal{G}}

\newcommand{\cS}{\mathcal{S}}

\newcommand{\cW}{\mathcal{W}}
\newcommand{\cX}{\mathcal{X}}
\newcommand{\cY}{\mathcal{Y}}
\newcommand{\cZ}{\mathcal{Z}}
\newcommand{\indep}{\upmodels}
\newcommand{\nindep}{\not\upmodels}

\begin{document}

\title{The standard no-signalling constraints in Bell scenarios are neither sufficient nor necessary for preventing superluminal signalling with general interventions}
\author{V. Vilasini}
\email{vilasini@inria.fr}

\affiliation{Université Grenoble Alpes, Inria, Grenoble 38000, France}
\affiliation{Institute for Theoretical Physics, ETH Zurich, 8093 Z\"{u}rich, Switzerland}

\author{Roger Colbeck}
\email{roger.colbeck@york.ac.uk}
\affiliation{Department of Mathematics, University of York, Heslington, York YO10 5DD.}

\date{\today}

\begin{abstract}
The non-classical nature of correlations resulting from entangled quantum systems has led to debates about the compatibility of quantum theory and relativity, and about the right way to think about causation in light of non-classical correlations.
Key to a causal theory is that superluminal signalling is forbidden, which holds in quantum theory even with entangled systems. In Bell scenarios, relativistic principles like no superluminal signalling are usually taken to follow from the standard no-signalling constraints on correlations. 
In this work we restrict to multi-party Bell scenarios and explore the connections between a range of relativistic causality principles, including no superluminal signalling and no causal loops, and constraints on correlations that can arise with arbitrary interventions. These constraints include both the standard no-signalling conditions and proposed relaxations that allow theories that incorporate a phenomenon called jamming. It has been suggested that superluminal signalling and causal loops remain impossible in such theories.
Using a recent formalism that combines relativistic principles with a causal modelling framework, we show that any theory (classical or non-classical) that generates jamming correlations must necessarily do so through causal fine-tuning and by means of superluminal causal influences. Moreover, we show that jamming theories can lead to superluminal signalling unless certain systems are fundamentally inaccessible to agents. This highlights limitations for the physicality of jamming theories. However, we identify situations where jamming correlations do not lead to superluminal signalling under general interventions, therefore showing that the conventional no-signalling constraints are not necessary for this purpose. Next, we show through examples, that no-signalling constraints on correlations are generally insufficient for ruling out superluminal signalling, and are neither necessary nor sufficient for ensuring no causal loops, when general interventions are allowed. Finally, we derive necessary and sufficient conditions for ruling out superluminal signalling and operationally detectable causal loops. These results solidify our understanding of relativistic causality principles in information processing tasks in space-time.

\end{abstract}

\maketitle
\tableofcontents

\section{Introduction}
The nature of correlations arising from measurements on space-like separated quantum systems has fuelled historical debates concerning the compatibility between quantum theory and special relativity~\cite{EPR, Bohr, Bohm1957, BellEPR, Bell}. The developments in quantum information theory in recent decades, have revealed the power of quantum correlations for information processing tasks such as computing, communication and cryptography. Notably, these advances have also affirmed that, despite its non-relativistic as well as non-classical nature (as highlighted by the existence of quantum correlations that violate Bell inequalities~\cite{Bell, CHSH}), quantum mechanics nevertheless aligns harmoniously with relativistic principles, such as the impossibility of superluminal signalling. This has placed our understanding of quantum entanglement, quantum correlations, and their information-processing possibilities on solid theoretical and conceptual ground.\footnote{An overview of these advances and results can be found in any textbook on quantum information theory, e.g.,~\cite{NielsenChuang}.}

More generally, this raises the question: when can a theory of information processing be deemed compatible with relativistic causality principles in space-time? In Bell scenarios, where space-like separated parties locally measure subsystems of a shared quantum system, obtaining measurement statistics adhering to the non-signalling constraints (which demand that the measurement setting of one party remains uncorrelated with the measurement outcomes of the remaining parties) is usually taken as necessary and sufficient to prevent superluminal signalling. Quantum correlations satisfy these standard non-signalling constraints, as do correlations arising in more general post-quantum theories~\cite{PR1994}, and for this reason such theories are also known as generalized non-signalling theories (GNST)~\cite{Barrett07}. These non-signalling conditions considered widely in the previous literature are constraints on the observable correlations. However, signalling is a concept that indicates underlying causation, which cannot be fully captured by correlations alone and requires consideration of general interventions. Therefore, to characterize the necessary and sufficient conditions for satisfying relativistic principles such as no superluminal signalling (NSS) and no causal loops (NCL), we must also account for interventions. 

A related question is whether the standard non-signalling constraints (on correlations) are necessary and sufficient to ensure NSS in Bell scenarios, under general interventions. A similar question was previously considered in \cite{Horodecki2019} but without accounting for interventions, i.e., considering correlations alone. In this context, in Bell scenarios with three or more parties, relaxations of the standard non-signalling constraints have been proposed~\cite{Grunhaus1996, Horodecki2019}. These relaxations are permissive to a broader range of theories than GNST, with the so-called \emph{jamming non-local theories}~\cite{Grunhaus1996} being examples. Existing literature suggests that jamming theories are consistent with relativistic causality principles in Minkowski space-time in that they do not give rise to superluminal signalling or causal loops, and therefore go on to dub them as \emph{relativistically causal theories}~\cite{Horodecki2019}. Moreover, \cite{Horodecki2019} claimed that these for bipartite Bell scenarios the regular non-signalling constraints and for tri-partite scenarios, the relaxed non-signalling constraints are necessary and sufficient for NSS and NCL.

However, these claims are not made within a mathematical framework that rigorously defines concepts such as \emph{causality}, \emph{causal loops} and \emph{superluminal signalling} that feature in the arguments. Such a framework is needed to put these concepts on a solid footing and clearly define the conditions under which a (possibly non-classical) theory adheres to or violates them. There are at least two notions of causality at play in information-processing protocols within space-time: the \emph{information-theoretic causal structure of the protocol} and the \emph{causal structure of space-time}. While the two are distinct, in physical experiments, they interact in a compatible manner. Therefore, in order to formalise relativistic principles of causation, a framework needs to distinguish and relate these concepts, while accounting for general interventions that agents can perform and by means of which they may signal to other agents. Such a formalism, addressing both these aspects in a theory-independent manner, has been developed in our recent works~\cite{VilasiniColbeckPRA, VilasiniColbeckPRL}.

The formalism of~\cite{VilasiniColbeckPRA, VilasiniColbeckPRL} adopts an approach based on \emph{causal modelling} to formalize information-theoretic causal structures. Causal modelling, originally developed in classical statistics and data sciences~\cite{Pearl2009, Spirtes2001}, is a powerful tool for linking causation, correlations, and interventions with applications in various fields, such as machine learning, medical testing, biology, and economics. In light of the challenges posed by quantum correlations to classical causal modelling frameworks~\cite{Wood2015}, there has been significant progress in developing quantum and non-classical causal models~\cite{Leifer2013, Henson2014, Costa2016, Allen2017, Barrett2020}. The formalism of~\cite{VilasiniColbeckPRA, VilasiniColbeckPRL} builds on this progress and allows cyclic and non-classical causal models to be described in quantum and post-quantum theories. Specifically, the concept of \emph{higher-order affects relations} was introduced to model general interventions in these settings. Relativistic causality principles, such as the impossibility of superluminal signalling, can be precisely formalized by embedding these causal models in a background space-time and imposing graph-theoretic \emph{compatibility conditions} between the higher-order affects relations coming from the causal model and the light-cone structure coming from the space-time. Using this approach, it was demonstrated that the absence of superluminal signalling in Minkowski space-time does not preclude causal loops~\cite{VilasiniColbeckPRL}, highlighting that no superluminal signalling and no causal loops are distinct principles that do not imply one another.

 In the present work, we apply this formalism to analyze these relativistic causality principles in Bell scenarios, in a theory-independent manner. Our main results are two-fold: we first focus on jamming theories (those admitting jamming non-local correlations), establishing a number of necessary properties of such theories and showing that they will inevitably violate certain principles of information-theoretic causation (such as the principle of no fine-tuning, see below) as well as relativistic causation (such as no superluminal causation). Secondly, we derive a set of necessary and sufficient conditions for ensuring no superluminal signalling, and for ruling out operationally relevant causal loops in Bell scenarios. These results and analysis shed light on the (un)physicality of jamming theories, and also on the fact captured in the title, that the standard non-signalling conditions are neither necessary nor sufficient for ensuring the (distinct) relativistic principles of NSS and NCL in Bell scenarios, when allowing general interventions. Additionally, this implies that the relaxed non-signalling conditions of \cite{Horodecki2019} are also insufficient for ensuring either of these relativistic principles. 
 
\subsection{Outline of paper and main contributions}
The initial sections of the paper provide an overview of the relevant background for our results. In Section~\ref{sec: standard_NS}, we review the standard bi and tripartite non-signalling constraints in the respective Bell scenarios. Section~\ref{sec:jam_review} offers an overview of jamming non-local correlations and the related claims about relativistic causality in previous literature. Section~\ref{sec: framework} provides a self-contained review of our recent causality framework~\cite{VilasiniColbeckPRA}, which was motivated in the introduction. The main contributions are summarized below. These are two-fold: insights on information-theoretic and relativistic causality principles in jamming theories, and necessary and sufficient conditions for ensuring relativistic principles in Bell scenarios in general theories and while allowing for interventions.

\begin{itemize}
    \item {\bf Jamming requires fine-tuning.} One reason why classical causal explanations of correlations violating a Bell inequality are rejected is because they must be fine-tuned~\cite{Wood2015}, that is, they must be carefully chosen in order to conceal causal influences from showing up in the observed correlations or signalling relations.  The usual quantum explanation need not be though. On the other hand, in Section~\ref{sec: jam_finetune_superlumcaus} we demonstrate that any setup exhibiting jamming correlations must involve \emph{causal fine-tuning} whether or not the underlying theory is classical. Fine-tuning is often considered an undesirable property of a causal model because it challenges the ability to infer causal influences from observed data and for this reason, is often used as a principle to reject particular theories that are unable to explain observations without alluding to fine-tuning (see e.g.,~\cite{Wood2015, Cavalcanti2018, Ying2023}). Hence the information theoretic causality principle no fine-tuning allows us to rule out jamming correlations regardless of the physical theory describing the underlying causes.
    
    \item {\bf Jamming and violation of relativistic causality principles.} Embedding the tripartite Bell scenario in its natural space-time configuration proposed in previous works, it follows that any theory that exhibits jamming correlations in this scenario must involve superluminal causal influences -- see Section~\ref{sec: jam_finetune_superlumcaus}. Furthermore, in Section~\ref{sec: jammingproof}, we show that jamming theories can also lead to superluminal signalling in this space-time configuration, contrary to the original motivation for such theories \cite{Grunhaus1996, Horodecki2019}. We find that to prevent superluminal signalling in jamming theories, it is necessary that the common shared system of the Bell scenario remains fundamentally inaccessible to the agents preventing it from being subject to interventions (even if classical), and there are jamming scenarios where superluminal signalling can be avoided in this manner. This result highlights a necessary tension between jamming theories and \emph{relativistic causality principles}.
    
    \item {\bf Space-time configurations that permit jamming.} Because jamming is key to this work we have also considered the set of space-time configurations in which it can occur. To support jamming, such configurations require the intersection of the future lightcones of two space-time events to be contained in the future lightcone of another. In Appendix~\ref{app:ST} we give a complete set of conditions for when this occurs in Minkowski spacetime of arbitrary dimension. The main point to note is that there is a significant difference between the set of possible configurations allowing jamming with one spatial dimension as opposed to more.

    \item {\bf Necessary and sufficient conditions for relativistic causality in Bell scenarios.} 
    In Section~\ref{sec: NS_affects}, we provide an equivalent formalization of the standard and relaxed non-signalling constraints in terms of higher-order affects relations. We then demonstrate through explicit examples that in bi and tripartite Bell scenarios, neither the standard nor the relaxed non-signalling constraints are sufficient to rule out superluminal signalling (Section~\ref{sec: Bell_relcaus}), in contrast to a previous claim~\cite{Horodecki2019}. The analysis reveals the precise reason for the failure of the previous claim, in essence: non-signalling constraints only pertain to correlations, and when the settings are freely chosen, they tell us about signalling through interventions on setting variables. However, this does not fully encompass causation and does not tell us whether more general interventions (on both settings and outcomes) could be used for signalling. We employ causal modelling and affect relations to establish a set of necessary and sufficient conditions for ruling out superluminal signalling in Bell scenarios, when considering arbitrary interventions. Finally, we discuss conditions for ruling out causal loops, showing with counter-examples that the necessary and sufficient conditions for no superluminal signalling (mentioned above) may not generally be sufficient or necessary for ruling out causal loops in Bell scenarios. 

\end{itemize}
These results show that previous claims~\cite{Grunhaus1996, Horodecki2019} regarding no superluminal signalling and no causal loops in jamming theories do that hold once the claims are formalised within our general framework. We are currently unaware of any other framework within which our counter-examples to the claims could be evaded and where the claims could be proven to be true, but this remains a possibility. With these results we lay a foundation for characterising relativistic causality principles (such as no superluminal signalling and no causal loops) in Bell-type scenarios in arbitrary theories, while accounting for correlations as well as interventions. More generally, the framework of~\cite{VilasiniColbeckPRL, VilasiniColbeckPRA} can also be applied to characterise relativistic causality principles in a theory-independent manner in arbitrary information-processing scenarios, in a similar manner as we have done for Bell-type scenarios in this work. Section~\ref{sec: conclusions} presents our conclusions and offers an outlook for future research. 

\section{The standard non-signalling conditions in Bell-type scenarios}
\label{sec: standard_NS}

Consider a bipartite Bell experiment involving two parties Alice and Bob who share a common system $\Lambda$ (which may be a classical random variable, a quantum system or more generally, a system described by a post-quantum probabilistic theory). Alice and Bob perform measurements chosen freely  using the measurements settings $A$ and $B$ on their subsystems of the shared system, and obtain corresponding measurement outcomes $X$ and $Y$ respectively (where $A$, $B$, $X$ and $Y$ are classical random variables). Here, a variable (such as $A$) is said to be freely chosen if it has no prior causes relevant to the experiment (see Definition~\ref{def: Bell_scenario} for further details on how we model free choice in Bell experiments).
Now, suppose that Alice and Bob's measurements are performed at space-like separated locations i.e., $A$, $B$, $X$, $Y$ and $\Lambda$ are embedded in space-time as shown in Figure~\ref{fig:Bell_bipart_sptime}. Then, in order to ensure that the parties cannot signal to each other superluminally, it is necessary to require that each party's outcome is independent of the measurement setting of the other party.

Formally, taking $P(XY|AB)$ to be the joint conditional probability of the outcomes given the settings, we have the following non-signalling conditions on the correlations arising in such a bipartite Bell scenario. In the rest of the paper, we use the small letter $x$ to denote a particular value of a random variable denoted by the corresponding capital letter $X$.

\begin{definition}[Standard bipartite non-signalling conditions \texorpdfstring{\hyperref[cond: NS2]{(NS2)}}{NS2}]
\label{cond: NS2}
\begin{align}
\label{eq: bipartiteNS}
    \begin{split}
         P(x|a)&:=\sum_{y} P(x,y|a,b)=\sum_{y} P(x,y|a,b')\quad \forall x,a,b,b'\\
         P(y|b)&:=\sum_{x} P(x,y|a,b)=\sum_{x} P(x,y|a',b)\quad \forall y,a,a',b\\
    \end{split}
\end{align}
\end{definition}

We refer to these as the standard bipartite non-signalling conditions or \hyperref[cond: NS2]{(NS2)}. In short, we will denote the above conditions as $P(X|AB)=P(X|A)$ and $P(Y|AB)=P(Y|B)$. If either one of these conditions is violated, then the measurement outcome of one party will be correlated with the choice of measurement of the other, space-like separated party allowing the latter party to communicate information to the former by choosing different measurements on their half of the joint state.

Similarly, consider a tripartite Bell experiment between the three parties Alice, Bob and Charlie who share a tripartite system $\Lambda$ (which may again be described by a classical, quantum or post-quantum theory), measure their subsystem with freely chosen measurement settings $A$, $B$ and $C$ and obtain corresponding measurement outcomes $X$, $Y$ and $Z$ respectively. 
 If the parties' measurements are conducted within pairwise space-like separated space-time regions, such as in the space-time configuration of Figure~\ref{fig: Bell_tripart1}, then the standard tri-partite non-signalling conditions (which are also motivated by the impossibility of superluminal signalling) require that the measurement setting of each party cannot be correlated with the outcomes of any set of complementary parties. Explicitly, this translates to the following constraints on the joint conditional probability distribution $P(X,Y,Z|A,B,C)$.

\begin{definition}[Standard tripartite non-signalling conditions \texorpdfstring{\hyperref[cond: NS3]{(NS3)}}{NS3}]
\label{cond: NS3}
\begin{align}
\label{eq: tripartiteNS}
    \begin{split}
         P(x,y|a,b)&:=\sum_zP(x,y,z|a,b,c)=\sum_z P(x,y,z|a,b,c')\quad \forall x,y,a,b,c,c'\\
          P(y,z|b,c)&:=\sum_xP(x,y,z|a,b,c)=\sum_x P(x,y,z|a',b,c)\quad \forall y,z,a,a',b,c\\
            P(x,z|a,c)&:=\sum_yP(x,y,z|a,b,c)=\sum_y P(x,y,z|a,b',c)\quad \forall x,z,a,b,b',c\\
    \end{split}
\end{align}
\end{definition}

In short, these conditions are expressed as $P(XY|ABC)=P(XY|AB)$, $P(YZ|ABC)=P(YZ|BC)$ and $P(XZ|ABC)=P(XZ|AC)$. Notice that the above conditions of Equation~\eqref{eq: tripartiteNS} imply the following conditions which are equivalent to writing $P(X|ABC)=P(X|A)$, $P(Y|ABC)=P(Y|B)$ and $P(Z|ABC)=P(Z|C)$.

\begin{align}
    \begin{split}
         P(x|a)&:=\sum_{y,z} P(x,y,z|a,b,c)=\sum_{y,z} P(x,y,z|a,b',c')\quad \forall x,a,b,b',c,c'\\
         P(y|b)&:=\sum_{x,z} P(x,y,z|a,b,c)=\sum_{x,z} P(x,y,z|a',b,c')\quad \forall y,a,a',b,c,c'\\
        P(z|c)&:=\sum_{x,y} P(x,y,z|a,b,c)=\sum_{x,y} P(x,y,z|a',b',c)\quad \forall z,a,a',b,b',c\\
 \end{split}
\end{align}

      \begin{figure}
          \centering
          \subfloat[\label{fig:Bell_bipart_CS}]{\includegraphics[]{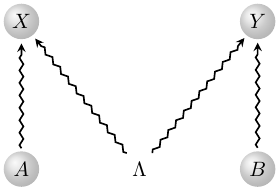}}\qquad\qquad\qquad\qquad \subfloat[]{    

         \includegraphics[]{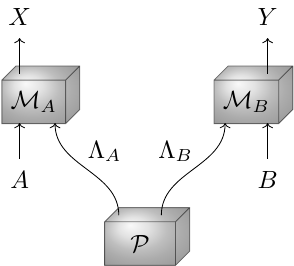}
}
          \caption{{\bf Typical causal structure associated with a bipartite Bell experiment} (a)  This figure represents the typical information-theoretic causal structure, which captures the flow of information in a protocol where Alice and Bob share a bipartite system $\Lambda$, and locally perform freely chosen measurements on their halves of this bipartite resource to generate outcomes. The settings $A$, $B$, and outcomes $X$ and $Y$ are observed classical variables (circled) and the resource corresponding to the common cause $\Lambda$ may be classical, quantum or post-quantum. (b) The described protocol can be represented by a circuit diagram as shown here, and one can see that the causal structure of (a) captures the connectivity of this circuit diagram. Such diagrams involve the causal mechanisms (systems, transformations, measurements) of a theory. The meaning of the information-theoretic causal structure and causal mechanisms are explained further in Section~\ref{sec: framework}.}
      \end{figure}

\begin{figure}[t!]
    \centering
 \includegraphics[]{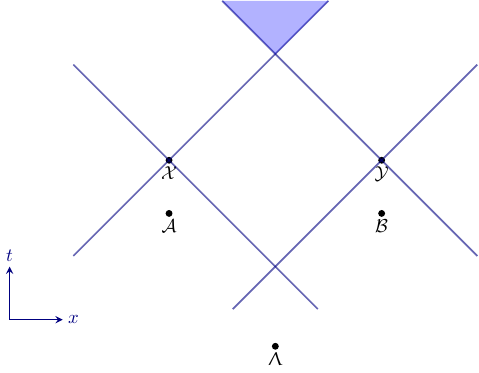}
    \caption{{\bf Space-time configuration associated with a bipartite Bell experiment} Space-time points are labelled by the variables of the information-theoretic protocol which are embedded at that location, this can be understood as the space-time location at which the corresponding measurement setting/outcome is generated. We distinguish space-time embedded random variables from the random variable itself using the caligraphic font for the former.  Here, the measurement events of the two parties are space-like separated. The diagonal lines represent light-like surfaces. The blue region is the joint future of both measurement outcomes, where the joint probability distribution $P(XY|AB)$ can be evaluated.
    }
\label{fig:Bell_bipart_sptime}
\end{figure}

{\bf Information-theoretic causal structure vs space-time configuration} In the above discussion of standard Bell scenarios, there are two distinct notions of causality at play: the information-theoretic causal structure of the experiment (Figures~\ref{fig:Bell_bipart_CS} and~\ref{fig: Bell_tripart2}) and the spatio-temporal causal structure in which we embed the physical systems of the experiment (Figures~\ref{fig:Bell_bipart_sptime} and~\ref{fig: Bell_tripart1}). The former captures the flow of information between physical systems in the experiment, for instance, if Alice does not communicate to Bob, then there would not be any causal arrows from systems in Alice's control to systems in Bob's control. This would be independent of whether or not Bob acts in the spatio-temporal future of Alice i.e., we can naturally have scenarios where Bob is in the future of Alice relative to the causal structure of the space-time but not so relative to the information-theoretic causal structure (due to absence of communication channels). However,  through principles of relativistic causality, the spatio-temporal structure constrains the possible information-theoretic structures. For instance, if Alice and Bob are space-like separated, we would expect, from relativistic principle of no superluminal causation, that the information-theoretic causal structure also does not have causal arrows between the variables describing Alice and Bob's measurements. Indeed, for the information-theoretic causal structures of the Bell experiment (illustrated in Figures~\ref{fig:Bell_bipart_CS} and~\ref{fig: Bell_tripart2}) and the corresponding space-time configurations (illustrated in Figures~\ref{fig:Bell_bipart_sptime} and~\ref{fig: Bell_tripart1}), all (information-theoretic) causal arrows flow from past to future in the space-time. 

Moreover, through a causal modelling approach, one can show that the information-theoretic causal structures of Figures~\ref{fig:Bell_bipart_CS} and~\ref{fig: Bell_tripart2} imply the non-signalling constraints \hyperref[cond: NS2]{(NS2)} and \hyperref[cond: NS3]{(NS3)} as a consequence of their topology, which highlights that these constraints are purely a property of the information-theoretic structure.\footnote{More formally, this would be a consequence of the d-separation property (cf.\ Definition~\ref{definition: compatdist}) applied to these causal structures, which has been proven to hold for acyclic causal models in classical, quantum and post-quantum theories~\cite{Pearl2009, Henson2014, Barrett2020A}.} However, when this structure is embedded in a space-time such that Alice and Bob's measurement variables are embedded are space-like separated locations, then these constraints become relevant for avoiding signalling outside the future light cone of the space-time. In Section~\ref{sec: framework} we review a framework that we have introduced in~\cite{VilasiniColbeckPRA, VilasiniColbeckPRL} within which such statements linking information-theoretic and spatio-temporal causal structures can be made mathematically precise in rather general scenarios, while maintaining a clear conceptual separation between these distinct notions of causation. 

\bigskip

{\bf Free choice:} Notice that in order to interpret the conditions \hyperref[cond: NS2]{(NS2)} and \hyperref[cond: NS3]{(NS3)} as no-signalling conditions, the free choice of the settings $A$ and $B$ is crucial. For instance, if the setting, $A$, was not freely chosen but rather, determined by some event in the past, then it could naturally be correlated with the outcome $Y$ of Bob. This is because an event in the common past lightcone of $A$ and $Y$ could influence and correlate these variables without having any causal influences from $A$ to $Y$ and without enabling Alice to signal information to Bob using these correlations. However, $A$ being freely chosen precludes any causes of $A$ that are relevant to the rest of the experiment, in which case correlations between $A$ and $Y$ would indicate signalling from $A$ to $Y$. Whenever $A$ and $Y$ are embedded at space-like separated locations in space-time, this signal would be superluminal. The condition that $A$, $B$ and $C$ are freely chosen is naturally captured through the information-theoretic causal structure of the experiment, by modelling them as parentless nodes in the causal graph. This is the case in the typical information-theoretic causal structures used for the bipartite and tripartite Bell experiment (Figures~\ref{fig:Bell_bipart_CS} and~\ref{fig: Bell_tripart2}).

\begin{figure}[t!]
\centering
\subfloat[\label{fig: Bell_tripart1}]{	 \includegraphics[]{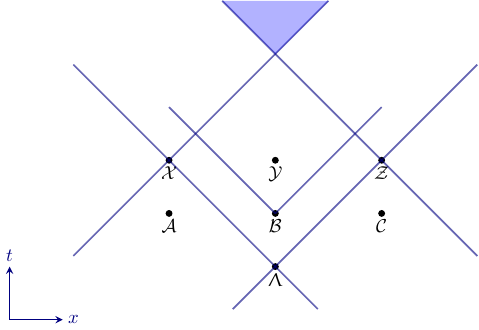}}\quad\quad
\subfloat[\label{fig: Bell_tripart2}]{ \includegraphics[]{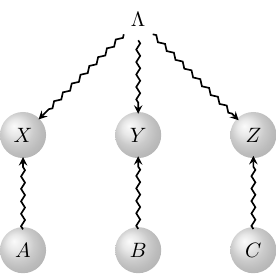}}	
	\caption{\textbf{A tripartite Bell experiment:}  \textbf{(a)} Space-time configuration for the jamming scenario~\cite{Grunhaus1996, Horodecki2019}: The measurement events of the three parties are space-like separated such that the future of Bob's input $B$ contains the entire joint future (in blue) of Alice and Charlie's outputs $X$ and $Z$. Here, $B$ is allowed to signal to $X$ and $Z$ jointly (which can only be verified in the blue region) but not individually. \textbf{(b)} Information-theoretic causal structure for the usual tripartite Bell experiment: 
 This causal structure is typically used for modelling a tripartite Bell experiment where correlations are constrained by the standard non-signalling conditions \hyperref[cond: NS3]{(NS3)}. In order to model the more general conditions \hyperref[cond: NS3p]{(NS3$'$)} that permit jamming, some modification is required. Figure~\ref{fig: tripart_jamming} is a natural modification of this causal structure which can admit jamming correlations. }
	\label{fig: Bell_tripart}
\end{figure}

\section{Jamming non-local correlations and claims about relativistic causality}
\label{sec:jam_review}

The post-quantum phenomenon known as ``jamming'' was originally proposed in~\cite{Grunhaus1996} to argue that in a tripartite Bell scenario with pairwise space-like separated measurements, the standard non-signalling conditions \hyperref[cond: NS3]{(NS3)} are not strictly necessary for ensuring no superluminal signalling. This defines a more general set of non-local correlations in tripartite Bell scenarios that go beyond correlations achievable in quantum theory or standard no-signalling probabilistic theories, but are nevertheless expected to not lead to any superluminal signalling. 
We now review the argument presented in the previous literature on jamming~\cite{Grunhaus1996, Horodecki2019}, while focussing our attention on the tripartite Bell scenario, which is the simplest Bell scenario that exhibits jamming.

The conditions \hyperref[cond: NS3]{(NS3)} require, in particular, that the outcomes $X$ and $Z$ of Alice and Charlie are independent of the setting $B$ of Bob, that is $P(XZ|ABC)=P(XZ|AC)$. However, it was noted in~\cite{Grunhaus1996} that in a space-time configuration where the joint future of the space-time locations of $X$ and $Z$ is entirely contained in the future of the location of $B$ (as shown in Figure~\ref{fig: Bell_tripart1}),
a violation of $P(XZ|ABC)= P(XZ|AC)$ would not lead to superluminal signalling, as long as $P(X|AB)= P(X|A)$ and $P(Z|BC)= P(Z|C)$ still hold. This is because whether or not $P(XZ|ABC)=P(XZ|AC)$ holds can only be checked within the joint future of $X$ and $Z$. The logic is that a failure of this condition only leads to signalling within $B$'s future in this space-time configuration, as long as the marginals of the space-like separated events, $P(X|A)$ and $P(Z|C)$ are independent of $B$, and therefore involves no superluminal signalling. Nevertheless, this would allow Bob to non-locally ``jam'' the correlations between Alice and Charlie, even though each of them are space-like separated to him. Using this observation,~\cite{Horodecki2019} propose the set of relaxed tripartite non-signalling conditions \hyperref[cond: NS3p]{(NS3$'$)} which allows for a larger set of correlations as compared to \hyperref[cond: NS3]{(NS3)} by virtue of being less constrained. The difference between \hyperref[cond: NS3]{(NS3)} and \hyperref[cond: NS3p]{(NS3$'$)} lies in whether or not $P(XZ|ABC)=P(XZ|AC)$ necessarily holds.

\begin{definition}[Relaxed tripartite non-signalling conditions \texorpdfstring{\hyperref[cond: NS3p]{(NS3$'$)}}{NS3p}]
\label{cond: NS3p}
\begin{align}
\label{eq: rel_tripartiteNS}
    \begin{split}
         P(x,y|a,b)&:=\sum_zP(x,y,z|a,b,c)=\sum_z P(x,y,z|a,b,c')\quad \forall x,y,a,b,c,c'\\
          P(y,z|b,c)&:=\sum_xP(x,y,z|a,b,c)=\sum_x P(x,y,z|a',b,c)\quad \forall y,z,a,a',b,c\\
      P(x|a)&:=\sum_{y,z} P(x,y,z|a,b,c)=\sum_{y,z} P(x,y,z|a,b',c')\quad \forall x,a,b,b',c,c'\\
       P(z|c)&:=\sum_{x,y} P(x,y,z|a,b,c)=\sum_{x,y} P(x,y,z|a',b',c)\quad \forall z,a,a',b,b',c\\
    \end{split}
\end{align}
\end{definition}

Notice that the first two of the above conditions imply the following independence. 

\begin{equation}
   P(y|b):=\sum_{x,z} P(x,y,z|a,b,c)=\sum_{x,z} P(x,y,z|a',b,c')\quad \forall y,a,a',b,c,c' 
\end{equation}

\begin{definition}[Tripartite jamming]
\label{def: jamming}
    We say that a tripartite Bell scenario admits jamming if and only if the corresponding joint distribution $P(XYZ|ABC)$ satisfies \hyperref[cond: NS3p]{(NS3$'$)} but not \hyperref[cond: NS3]{(NS3)} i.e., satisfies \hyperref[cond: NS3p]{(NS3$'$)} with $P(XZ|ABC)\neq P(XZ|AC)$.
\end{definition}


Based on the above arguments,~\cite{Grunhaus1996, Horodecki2019} claim that in the space-time configuration of Figure~\ref{fig: Bell_tripart1}, \hyperref[cond: NS3p]{(NS3$'$)} and not \hyperref[cond: NS3]{(NS3)} are the necessary conditions on tripartite Bell correlations required for avoiding superluminal signalling. 
Inspired by the special relativistic observation that the possibility of superluminal signalling in Minkowski space-time enables the possibility of signalling to the past,~\cite{Horodecki2019} propose the following definition for relativistic causality that forbids causal loops. Here, measurement settings and outcomes are viewed as space-time random variables, that is, random variables associated with a space-time location. The standard notation $(\vec{r},t)$ with $\vec{r}\in \mathbb{R}^3$ and $t\in \mathbb{R}$ is used to denote the space-time co-ordinates of a point in Minkowski space-time in some inertial frame of reference $\mathcal{I}$. A measurement event corresponds to a pair of random variables, a setting and an outcome variable.~\cite{Horodecki2019} associate the same location with both variables in a measurement event, as an idealisation, and use the following quoted text as their definition of relativistic causality.

\begin{center}
\label{definition: horodecki}
    Quotation from~\cite{Horodecki2019} (Relativistic causality condition): ``No causal loops occur, where a causal loop is a sequence of events, in which one event is among the causes of another event, which in turn is among the causes of the first event.''
\end{center}

They then go a step beyond the previous claim and make the following stronger claims for the bipartite and tripartite Bell scenarios.

\begin{claim}[Conditions for relativistic causality in bipartite Bell scenario~\cite{Horodecki2019}]
\label{claim: bipart}
In a bipartite Bell experiment the standard bipartite non-signalling conditions \hyperref[cond: NS2]{(NS2)} are necessary and sufficient for ensuring that no violation of relativistic causality occurs. 
\end{claim}

\begin{claim}[Conditions for relativistic causality in tripartite Bell scenario~\cite{Horodecki2019}]
\label{claim: tripart}
Consider a tripartite Bell experiment performed in the space-time configuration of Figure~\ref{fig: Bell_tripart1}, where in an inertial frame of reference $\mathcal{I}$ the space-time locations $({\bf x}_X,t_X)$, $({\bf x}_Z,t_Z)$ and $({\bf x}_B,t_B)$ associated with the random variables $X$, $Z$ and $B$ are such that the joint future of $({\bf x}_X,t_X)$ and $({\bf x}_Z,t_Z)$ is entirely contained in the future of $({\bf x}_B,t_B)$. Then the relaxed tripartite non-signalling conditions \hyperref[cond: NS3p]{(NS3$'$)} are necessary and sufficient for ensuring that no violation of relativistic causality occurs. 
\end{claim}

One of the aims of this work is to assess the validity of these claims; possible issues with them are outlined in the next section.

\section{Challenges in formulating jamming in a causal framework accounting for information-theoretic and relativistic causality}

Formulating the concept of jamming within a rigorous causal framework and analyzing whether theories that permit jamming non-local correlations adhere to relativistic principles—such as no superluminal signalling and no causal loops—requires addressing several broader conceptual and technical challenges. As mentioned above, previous works have claimed that jamming theories comply with these principles, but this has been insufficient to resolve the afore-mentioned challenges. In particular, this is because the claims are made in an intuitive manner rather than within a framework with rigorous definitions of concepts such as ``cause'', ``events'', ``signalling'', ``causal loops'' which are used in the informal claims. These ambiguities hinder a systematic evaluation of their validity. Below, we outline the key challenges that highlight the need for a clear and robust formalism, while also identifying the specific obstacles to verifying the previous claims in the absence of such a formal framework.

\bigskip

{\bf 1. How is causation defined?}\\
The relativistic causality condition proposed in~\cite{Horodecki2019} hinges on the concepts of ``cause'' and ``events'' but the precise meaning of these concepts is left ambiguous. There are at least two distinct notions of causation at play here: information-theoretic causation which can be defined as a directed graph capturing the flow of information in a network of physical systems, and spatio-temporal causation given by the light-cone structure of the space-time. These works focus on Minkowski space-time whose light-cone structure forms a partial order. If we interpret ``events'' here as space-time events and ``causation'' as given by the light-cone structure, then a claim of no causal loops holds by construction, irrespective of the correlations being considered and thus \hyperref[cond: NS2]{(NS2)} or \hyperref[cond: NS3p]{(NS3$'$)} would not be necessary for relativistic causality in this case, which would be inconsistent with the claims. 

On the other hand, if we interpret ``events'' and ``causation'' relative to the nodes (physical systems) and directed edges (given by the structure of information-theoretic channels) of the information-theoretic causal structure, then the standard causal structure of a Bell scenario is a directed acyclic graph and hence also has no causal loops. This is because the measurement settings being free choices have no incoming arrows and the measurement outcomes are typically assumed to not have any outgoing causal arrows (this assumption also appears to be made in~\cite{Horodecki2019}, in the proof of the claims). Then again, the validity of the claims comes under question.  Therefore it is unclear if there exists a mathematically rigorous definition of these concepts with respect to which Claims~\ref{claim: bipart} and~\ref{claim: tripart} hold true.

\bigskip

{\bf 2. How to incorporate general interventions into the analysis of relativistic causality?}\\
The claims suggest that a condition on the correlations is necessary and sufficient for ruling out causal loops. However, because correlation does not imply causation, one would expect that further assumptions relating to the causal structure are needed in order for these claims to be true. As noted in Point~1 above, the notion of causal structure featuring in the statements is ambiguous and consequently so are the assumptions made on it. Generally, when inferring causation in the real-world (e.g., when inferring whether a drug causes recovery from a disease), one must consider interventions~\cite{Spirtes2001, Pearl2009} in addition to correlations, which form the basis of controlled trials and of the general methods by which we infer causation from observed data in practice. 

Interventions give the general method by which parties can signal between each other, and allow us to infer the presence of causal influences. However, signalling and causation are generally not equivalent, as it is possible that underlying causal parameters are fine-tuned which can lead to certain causal influences being undetectable at the level of observable signalling relations. Previous works on relativistic causality have focused mainly on correlations, and on causal models without such fine-tuned influences. However, as suggested in~\cite{VilasiniColbeckPRA, VilasiniColbeckPRL}, and shown in Section~\ref{sec: jam_finetune_superlumcaus}, jamming correlations necessarily require fine-tuned causal explanations and if we assume no fine-tuning (which would allow signalling and causation to be equated), we can no longer study jamming correlations. It therefore becomes important, when analysing these questions, to account for correlations and interventions, while noting that signalling and causation are generally different (but related) concepts.

\bigskip

{\bf 3. Superluminal signalling and causal loops do not necessarily imply each other.}\\
The connection between the non-signalling conditions \hyperref[cond: NS2]{(NS2)} and \hyperref[cond: NS3p]{(NS3$'$)} and superluminal signalling in bi/tripartite Bell scenarios is more intuitively apparent from the original arguments of~\cite{Grunhaus1996}, which were also revisited in the previous sections. In~\cite{Horodecki2019} this is taken a step further by trying to link these non-signalling conditions to causal loops, arguing that a violation of \hyperref[cond: NS2]{(NS2)} or \hyperref[cond: NS3p]{(NS3$'$)} in the respective scenario would enable superluminal signalling which can be used to signal to the past and create causal loops. 
However, the possibility of superluminal signalling does not by itself imply causal loops: if superluminal signalling was possible only in a single preferred reference frame, this could not be used to create a causal loop as discussed in the proof of the claims in~\cite{Horodecki2019}. 

Conversely, one may consider whether the existence of a causal loop always enables parties to signal superluminally. In a recent work~\cite{VilasiniColbeckPRL}, we have shown that this is also not the case using an explicit example of an information-theoretic causal loop embedded in Minkowski space-time without leading to superluminal signalling. This highlights that these two principles (no causal loops and no superluminal signalling) which are both associated with relativistic causality are distinct principles that do not imply one another and one must therefore be careful to distinguish them when analysing relativistic causality.


\section{A framework for relating information-processing and relativistic causality in general theories}
\label{sec: framework}

We use the framework we developed in~\cite{VilasiniColbeckPRA} which has the following features, which together allow the problems identified in the previous section to be addressed:\footnote{A formalism incorporating similar desiderata was also developed in \cite{VilasiniRenner2022}, for the purpose of addressing open problems relating to so-called indefinite causal order quantum processes \cite{Oreshkov2012, Chiribella2013}. In this work we apply the framework of \cite{VilasiniColbeckPRA, VilasiniColbeckPRL} as we are interested in proving theory-independent statements in scenarios with a well-defined (but possibly cyclic) causal structure.}

\begin{itemize}
    \item Clearly disentangles and formalises spatio-temporal and information-theoretic causality notions
    \item Distinguishes between correlation, causation and signalling in a theory-independent manner
    \item Allows for non-classical, fine-tuned and cyclic causal influences in the information-theoretic causal structure
\item Connects the two causality notions in order to formalise different relativistic causality principles
\end{itemize}

We proceed by outlining this framework starting with causal models and then how they are embedded in space-time.

\subsection{Causal models, affects relations and causal loops}

To formalise an information-theoretic notion of causality, our framework from~\cite{VilasiniColbeckPRA} adopts a causal modelling approach formulated under minimal and theory-independent assumptions so as to be applicable to scenarios with cyclic and non-classical causal influences. Here we present an overview of the main ingredients of the formalism, illustrating the core physical intuition through examples, and refer the reader to Appendix~\ref{appendix: framework} and~\cite{VilasiniColbeckPRA, VilasiniColbeckPRL} for further details.

\bigskip
{\bf Causal structures and observed correlations}
A \emph{causal structure} is modelled as a directed graph $\mathcal{G}$ over a set of nodes $N$ with directed edges $\longrsquigarrow$ that denote direct causal influence from one node to another. This directed graph may either be cyclic or acyclic and is a priori independent of space and time.  Each node of the causal structure is labelled either as ``observed'' or ``unobserved''. The former are classical random variables (e.g., measurement settings and outcomes) while the latter may be a classical, quantum or post-quantum system, such as a system described by a generalised probabilistic theory (GPT). 
A causal structure is said to be classical, quantum or post-quantum depending on the nature of its unobserved nodes and the joint probability distribution $P_{\mathcal{G}}(N_{\mathrm{obs}})$ over the set $N_{\mathrm{obs}}\subseteq N$ of all observed nodes of a causal structure $\mathcal{G}$ is known as the \emph{observed distribution} of $\mathcal{G}$. Which observed distributions $P_{\mathcal{G}}$ are allowed by a causal structure $\mathcal{G}$ will in general depend on the theory describing its unobserved nodes. Whenever $\mathcal{G}$ is evident from context, we will simply denote $P_{\mathcal{G}}$ as $P$.

We use minimal, theory-independent constraints to define when we can associate a probability distribution over the variables $N_{\mathrm{obs}}$ with a causal structure $\mathcal{G}$. The idea behind this condition is that the topology of the causal graph $\mathcal{G}$, in particular whether or not there are graph separation relations between some sets of observed nodes, will imply corresponding conditional independences in any associated observed distribution over those nodes. The graph separation criterion used for this purpose is called \emph{d-separation} (where d stands for ``directed''), which was first introduced in the classical causal modelling literature~\cite{Pearl2009, Spirtes2001} and more recently considered in the context of quantum and GPT causal models~\cite{Henson2014, Barrett2020A}. We use $(X\perp^d Y|Z)_{\cG}$ to denote that $X$ is d-separated from $Y$ given $Z$ in $\cG$ (where $X$, $Y$, $Z$ are disjoint sets of observed nodes $N_{\mathrm{obs}}$ of $\cG$), and refer the reader to Appendix~\ref{appendix: framework} for the full technical definition. Then for a causal structure $\cG$ with a set of observed nodes $N_{\mathrm{obs}}$ we say that a distribution $P$ on $N_{\mathrm{obs}}$ satisfies the \emph{d-separation property} with respect to $\cG$ iff for all disjoint subsets $X$, $Y$, $Z$ of $N_{\mathrm{obs}}$,
\begin{equation}
    \label{eq: dsep_prop}
    (X\perp^d Y|Z)_{\cG} \quad \Rightarrow \quad (X\indep Y|Z)_{P},
\end{equation}

where $(X\indep Y|Z)_{P}$ denotes the conditional independence $P(XY|Z)=P(X|Z)P(Y|Z)$. This provides a theory-independent condition for associating a set of possible observed distributions with a graph $\cG$. The d-separation property implies certain intuitive features. For instance, in the typical bipartite Bell causal structure of Figure~\ref{fig:Bell_bipart_CS}, the non-signalling constraints \hyperref[cond: NS2]{(NS2)} on the observed correlations $P(XY|AB)$, along with the fact that $P(AB)=P(A)P(B)$ are implied by corresponding d-separations in the graph.

\bigskip

{\bf Causal model} The directed edges of $\mathcal{G}$ capture the flow of information through the network of observed and unobserved systems associated with its nodes. How this flow of information is mathematically modelled will depend on how states, information-processing channels and measurements (which we will broadly refer to as the \emph{causal mechanisms}) are modelled in the theory under consideration. For example, in classical deterministic theories, where all states are described by classical variables and all evolutions in terms of deterministic functions, then all the nodes of the causal structure would be classical random variables and we would require there to exist, for every node $N$, a deterministic function $f_N:\mathrm{par}(N)\to N$ from the set par$(N)$ of all the parents of $N$, to $N$, when $N$ has at least one parent, and the specification of a distribution $P(N)$ when $N$ is parentless. In classical probabilistic theories, one would require conditional distributions $P(N|\mathrm{par}(N))$ (or equivalently, stochastic maps) for each node $N$ that describe classical informational channels from the parents to the node. [Probabilistic models can be equivalently described as deterministic models by including additional (possibly hidden) variables associated with additional nodes.] 

More generally, these causal mechanisms can be quantum channels or dynamics described by a post-quantum theory. Previous approaches to causal modelling (in the classical and quantum literature) adopt a bottom-up approach where they start from assumptions on the causal mechanisms (possibly non-classical) associated with a graph and derive conditions (such as the d-separation property) on the observed correlations that they generate and these results are often derived for the case where the causal graph is acyclic. There are several such approaches to quantum and non-classical causal models (e.g.,~\cite{Leifer2013, Henson2014, Costa2016, Allen2017, Barrett2020}). In our framework~\cite{VilasiniColbeckPRA}, we adopt a top-down approach where we make no assumptions on the causal mechanisms or about the acyclicity of the causal structure. Instead we impose the d-separation property (Definition~\ref{definition: compatdist}) at the level of the observed correlations and study what this implies for inferring properties of the underlying causal structure (which may be fully or partially unknown). This gives the following minimal working definition of a causal model that does not make any assumptions about the underlying causal mechanisms of the theory (not even that the theory must have some circuit-like structure).

\begin{definition}[Causal model]
\label{def: causalmodel}
A causal model is a pair $(\cG, P_{\cG})$ where $\cG$ is a directed graph over a set $N_{\mathrm{obs}}$ of observed nodes (or equivalently random variables) and possibly additional unobserved nodes corresponding to classical/quantum/GPT systems, while $P_{\cG}$ is a joint probability distribution on $N_{\mathrm{obs}}$ that respects the d-separation property (Equation~\eqref{eq: dsep_prop}) relative to $\cG$.
\end{definition}

Throughout this paper we will often use $S_1S_2$ as shorthand for the union $S_1\cup S_2$ of two sets of variables.

\begin{figure}
    \centering
\subfloat[]{ \includegraphics[]{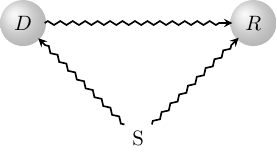}}\qquad\qquad\subfloat[]{ \includegraphics[]{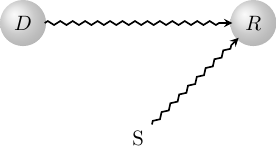}}
    \caption{If the original causal structure $\mathcal{G}$ between the observed nodes $D$, $R$ and unobserved node $S$ is that of (a), then the causal structure of (b) is the post-intervention causal structure $\mathcal{G}_{\mathrm{do}(D)}$ associated with an intervention on the node $D$.}
    \label{fig:intervention}
\end{figure}

\bigskip

{\bf Interventions and affects relations}
So far, we have discussed correlations, which are symmetric (unlike causation which is an asymmetric relation) and are not sufficient for inferring a causal structure. For instance, the existence of correlation between taking a drug (modelled as a variable $D$) and recovery from a disease (modelled as a variable $R$) in a population does not tell us whether or not the drug causes recovery, because it is possible to have a common factor $S$ (such as socio-economic background) that influences both an individual's access to a drug and their immunity to the disease, which could explain the same correlations. In addition to correlations, considering free interventions is
crucial for deducing causal relations from observable data. Such interventions play an important role in controlled trials that allow us to infer whether a drug causally influences recovery from a disease. 

Figure~\ref{fig:intervention} illustrates the concept of interventions in causal models. To deduce whether $D$ is a cause of $R$, we intervene on $D$ and this transforms the original causal structure $\cG$ (which may be unknown) to a post-intervention causal structure $\mathcal{G}_{\mathrm{do}(D)}$ that is a sub-graph of $\cG$ where the incoming arrows to $D$ are removed (making it parentless). An intervention then corresponds to fixing $D$ to take a certain value, $D=d$ in the post-intervention causal structure. Intuitively, this models the situation in a control trial where the values $d$ of $D$ would determine whether or not the drug (or a placebo) is advocated in a sub-population. Moreover, $D$ being parentless in the post-intervention causal structure models the natural assumption that the members of control group for the drug test and the placebo group are chosen freely and independently of prior causes of $D$ (such as $S$). Then, to infer whether $D$ causally influences another observed variable $R$, we check whether intervention on $D$ changes the distribution on $R$. In other words, whether there exists a value $d$ of $D$ such that the distribution on $R$ in the post-intervention scenario with $D$ fixed to $d$ differs from the original, pre-intervention distribution on $R$. If so, we will say that $D$ affects $R$. The following definitions formalise this intuition.

In previous causal modelling frameworks (prior to that of~\cite{VilasiniColbeckPRA}), the post-intervention causal model and causal inference are typically defined in terms of the causal mechanisms of the underlying theory. The framework of~\cite{VilasiniColbeckPRA} enables a more minimal definition of these concepts based on the minimal Definition~\ref{definition: compatdist} of a causal model, which is independent of the theory or causal mechanisms. 

\begin{definition}[Post-intervention causal model]
\label{def:post_intervention}
Consider a causal model on a graph $\cG$ specified by the graph together with a distribution $P_{\cG}$  satisfying Definition~\ref{definition: compatdist}. A post-intervention causal model associated with interventions on a subset $X$ of the observed nodes of $\cG$ is specified by a graph $\cG_{\mathrm{do}(X)}$ obtained from $\cG$ by removing all incoming directed edges to the set $X$, together with a distribution $P_{\cG_{\mathrm{do}(X)}}$ satisfying Definition~\ref{definition: compatdist} relative to $\cG_{\mathrm{do}(X)}$. Further, suppose $Y$ is a (possibly empty) subset of the observed nodes. Whenever another subset of the observed nodes, $X$ (disjoint from $Y$), contains only parentless nodes in $\cG_{\mathrm{do}(Y)}$\footnote{If $Y$ is empty, then  $\mathcal{G}_{\mathrm{do}(Y)}=\mathcal{G}$.}, we require that
\begin{equation}
    P_{\cG_{\mathrm{do}(XY)}}(S|XY)=P_{\cG_{\mathrm{do}(Y)}}(S|XY),
\end{equation}
where $S$ is the set of all other observed nodes.
\end{definition}

\begin{definition}[Affects relation]
\label{definition: affects}
Consider a causal model over a set of random variables $S$ with causal graph $\mathcal{G}$ and a joint distribution $P_{\cG}$ associated with it. For any disjoint and non-empty subsets $X$, $Y\subseteq S$, we say that $X$ \emph{affects} $Y$ if there exists a value $x$ of $X$ such that
\begin{equation*}
   P_{\cG_{\mathrm{do}(X)}}(Y|X=x)\neq P_{\cG}(Y).
\end{equation*}
\end{definition}

Operationally, $X$ \emph{affects} $Y$ tells us that $X$ \emph{signals to} $Y$, as it means that there is an operation (given by the choice of intervention) that can be performed locally on $X$ which results in an observable change at $Y$. In general, it has been shown in~\cite{VilasiniColbeckPRA} that affects relations are non-transitive ($X$ affects $Y$ and $Y$ affects $Z$ does not imply $X$ affects $Z$) and also that, correlations between $X$ and $Y$ in the post-intervention (on $X$) model implies that $X$ affects $Y$ i.e.,
\begin{equation}
\label{eq: correl_affects}
P_{\cG_{\mathrm{do}(X)}}(Y|X)\neq P_{\cG_{\mathrm{do}(X)}}(Y) 
\quad \Rightarrow \quad X \text { affects } Y.
\end{equation}

Interestingly, the converse of the above implication is not true, as shown by a counter-example in~\cite{VilasiniColbeckPRA}. Importantly, $X$ affects $Y$ allows us to infer that $X$ causally influences $Y$, i.e., if $X$ and $Y$ are sets of RVs, $X$ affects $Y$ implies that there exists $e_X\in X$ and $e_Y\in Y$ such that there is a directed path from $e_X$ to $e_Y$ in the underlying causal structure~\cite{VilasiniColbeckPRA}. However, the converse is not true, we can have $X\longrsquigarrow Y$ in the causal structure but have causal models on such a structure where $X$ does not affect $Y$. This is possible by fine-tuning the causal mechanisms of the model so as the hide the causal influence from being observable at the level of the affects relations. More formally, fine-tuning relates to the converse statement of the d-separation property (Equation~\eqref{eq: dsep_prop}), it tells that there is an absence of observed correlation (captured by conditional independence) despite a corresponding causal connection (captured by d-connection). 

\begin{definition}[Fine-tuned causal model]
A causal model $(\cG, P_{\cG})$  is said to be fine-tuned whenever for some disjoint sets $X$, $Y$ and $Z$ of observed nodes of $\cG$, we have $(X\indep Y|Z)_{P_{\cG}}$ but $(X\not\perp^d Y|Z)_{\cG}$. 
\end{definition}

In particular, in any causal structure where we have $X\longrsquigarrow Y$, $X$ and $Y$ are d-connected, so any causal model on such a causal structure where $P(XY)=P(X)P(Y)$ would be fine-tuned. For example, if we had a causal structure $X\longrsquigarrow Y$ and $Z\longrsquigarrow Y$ with a causal model where $X$, $Y$, $Z$ are binary, $Z$ is uniformly distributed and $Y=X\oplus Z$, then $P(XY)=P(X)P(Y)$ would hold for any distribution over $X$, thus making the causal model fine-tuned. Since $X$ is parentless, this also implies that $X$ does not affect $Y$ (even though $X$ is a cause of $Y$). Due to fine-tuning, a causal relationship $X\longrsquigarrow Y$ does not in general imply a corresponding affects relation $X$ affects $Y$, which motivates the following classification of the arrows $\longrsquigarrow$ into solid $\longrightarrow$ and dashed $\xdashrightarrow{}$ ones.

\begin{definition}[Solid and dashed arrows]
\label{definition: solidasharrows}
Given a causal graph $\mathcal{G}$, if two observed nodes $X$ and $Y$ in $\mathcal{G}$ sharing a directed edge $X\longrsquigarrow Y$ are such that $X$ affects $Y$, then the causal arrow $\longrsquigarrow$ between those nodes is called a \emph{solid arrow}, denoted $X\longrightarrow Y$. Further, all arrows $\longrsquigarrow$ between observed nodes in $\mathcal{G}$ that are \emph{not} solid arrows are called \emph{dashed arrows}, denoted $\xdashrightarrow{}$ i.e.,  $X\xdashrightarrow{} Y$ in $\cG$ implies that $X$ does not affect $Y$, for any observed nodes $X$ and $Y$.
\end{definition}

More generally, when considering interventions on multiple nodes, we can envisage protocols by which agents are able to signal to each other once information about other interventions are also given to them. The concept of \emph{higher-order affects relations} that we introduced in~\cite{VilasiniColbeckPRA} captures such signalling relations. 

\begin{definition}[Higher-order affects relation]
\label{definition:HOaffects}
Consider a causal model associated with a causal graph $\cG$ over a set $N_{\mathrm{obs}}$ of observed nodes and an observed distribution $P$. Let $X$ , $Y$, and $Z$ be three pairwise disjoint subsets of $N_{\mathrm{obs}}$, with $X$ and $Y$ non-empty. We say that $X$ affects $Y$ given do$(Z)$ if there exist values $x$ of $X$ and $z$ of $Z$ such that
\begin{equation}
\label{eq: HOaffects}
   P_{\cG_{\mathrm{do}(X Z)}}(Y|X=x,Z=z)\neq P_{\cG_{\mathrm{do}(Z)}}(Y|Z=z),
\end{equation}
which we denote in short as $ P_{\cG_{\mathrm{do}(X Z)}}(Y|X,Z)\neq P_{\cG_{\mathrm{do}(Z)}}(Y|Z)$.

When $Z= \emptyset$, we refer to this as a \emph{zeroth-order affects relation} and it coincides with Definition~\ref{definition: affects}. 
\end{definition}

As with (zeroth order) affects relations, higher-order affects relations $X$ affects $Y$ given do$(Z)$ also allow us to deduce that $X$ causally influences $Y$, or equivalently that \emph{at least one} variable $e_X$ in $X$ influences a variable $e_Y$ in $Y$, as shown below.
\begin{equation}
\label{eq: HO_cause}
X \text{ affects } Y \text{ given do} (Z)  \quad \Rightarrow\quad  \exists e_X\in X, e_Y\in Y, \quad e_X \longrsquigarrow...\longrsquigarrow e_Y \text{ in } \cG,
\end{equation}
where $e_X \longrsquigarrow...\longrsquigarrow e_Y$ denotes a directed path from $e_X$ to $e_Y$.
Notice that in a causal structure with $X_1 \longrsquigarrow Y$ and another disconnected node $X_2$ which has no in or outgoing edges, if we have $X_1$ affects $Y$, we will in general also have $X_1X_2$ affects $Y$. However we know that the latter affects relation does not add any new information beyond the former one since $X_2$ is completely irrelevant in this affects relation. The concept of \emph{reducibility} of (higher-order) affects relations formalises this idea just at the level of the observed affects relation (without assumptions on the causal structure). A formal definition can be found in Appendix~\ref{appendix: framework}. Applied to this simple example, we can say that $X_1X_2$ affects $Y$ is a reducible affects relation while $X_1$ is an irreducible affects relation, and $X_1X_2$ affects $Y$ can be reduced to $X_1$ affects $Y$. In the case that $X$ affects $Y$ given do$(Z)$ is an irreducible affects relation, we have a stronger causal inference statement that allows us to infer causation from \emph{every} element $e_X$ of $X$.
\begin{equation}
\label{eq: HOirred_cause}
X \text{ affects } Y \text{ given do} (Z) \text{ irreducible}   \quad \Rightarrow\quad  \forall e_X\in X, \exists e_Y\in Y, \quad e_X \longrsquigarrow...\longrsquigarrow e_Y \text{ in } \cG,
\end{equation}

Higher-order affects relations have several useful implications and results for causal inference in the presence of cyclic and fine-tuned causal influences as well as latent non-classical causes. In~\cite{VilasiniColbeckPRA} we show that \emph{Pearl's rules of do-calculus}~\cite{Pearl2009} originally proven for classical causal models hold in our general framework and follow from the d-separation property alone. The results mentioned here, along with other properties proven in~\cite{VilasiniColbeckPRA}, reveal the rich and highly non-trivial structure of (higher-order) affects relations and their precise relations to the concepts of causation and correlations. More generally, we have \emph{conditional higher-order affects relations} which also account for post-selection on non-interventional nodes, an interested reader can find further details on this in Appendix~\ref{appendix: framework}.

\bigskip

{\bf Causal loops}
We can now provide a definition of causal loops that is a priori independent of space-time. A causal loop simply corresponds to a directed cycle in a causal structure $\mathcal{G}$ and we only consider causal loops involving observed nodes i.e., where for two observed nodes $X$ and $Y$ in $\cG$, there exist directed paths from $X$ to $Y$ and from $Y$ to $X$. This itself is not a very useful definition because, as we have seen, due to fine-tuning, causation does not imply the existence of signalling (i.e., affects relations).  The latter provided motivation to classify causal arrows (Definition~\ref{definition: solidasharrows}) into solid and dashed based on the operational detectability of these causal influences through corresponding affects relations. Similarly, we can distinguish between different types of causal loops in our framework depending on whether they can be operationally detected through their affects relations. 

\begin{definition}[Affects causal loops (ACL)]
\label{def: ACL}
Any set of affects relations $\mathscr{A}$ that can only arise in a causal model associated with a cyclic causal structure $\mathcal{G}$ are said to contain an affects causal loop. In other words, affects causal loops certify the cyclicity of the underlying causal structure through the observed affects relations.
\end{definition}

The simplest example of an ACL is $\mathscr{A}=\{X \text{ affects } Y, Y \text{ affects } X\}$, where $X$ and $Y$ are individual variables (see also Example~\ref{example: affects_loop} later). The former relations implies that $X$ is a cause of $Y$ and the latter that $Y$ is a cause of $X$, implying that any causal model capable of producing these affects relations must contain a directed cycle between the nodes associated with $X$ and $Y$. More generally we can have affects causal loops characterised by affects relations between sets of observed nodes, e.g., $\mathscr{A}=\{AC \text{ affects } B, B \text{ affects } AC\}$ where both are irreducible (if they are not irreducible, this affects relations may not necessarily form an affects causal loop\footnote{Irreducibility in this case implies that $A$ is a cause of $B$ and $C$ is a cause of $B$, and $B$ is a cause of at least one of $A$ and $C$ (cf.\ Equation~\eqref{eq: HOirred_cause}), which implies a directed cycle between $A$ and $B$ or $C$ and $B$. However, if $AC$ affects $B$ was reducible, then it can be that only one of $A$ or $C$ is a cause of $B$ while $B$ is a cause of the other variable which does not cause $B$ (cf.\ Equation~\eqref{eq: HO_cause}), which would not necessarily be cyclic.}). In~\cite{VilasiniColbeckPRL} we have constructed an explicit causal model with these affects relations and shown that there exists a way to ``embed'' this affects causal loop in $(1+1)$-dimensional Minkowski space-time such that it does not lead to superluminal signalling, which establishes the (mathematical) possibility of operationally detectable causal loops that can be compatible with a partially ordered space-time. This example is illustrated and further explained in Figure~\ref{fig: loop_PRL}. Definitions and characterisations of several further classes of affects causal loops that arise in this framework can be found in \cite{VilasiniColbeckPRA}.

\begin{definition}[Hidden causal loop (HCL)]
\label{def: HCL}
Given a causal model whose causal structure contains a directed cycle, we say that this causal model contains a \emph{hidden causal loop} if the same set of affects relations and same correlations as this model are also realisable in a causal model on an acyclic causal structure.
\end{definition}

HCLs can never be operationally detected, and, as they can be explained also in an acyclic model, one would naturally prefer the simpler acyclic explanation of the associated operational predictions (see Example~\ref{example: hidden_loop} and~\cite{VilasiniColbeckPRA} for examples of HCLs). 

\subsection{Space-time embeddings and compatibility}
Having defined causality operationally, without reference to a space-time structure, we now bring space-time into the picture to define when a causal model can be said to be compatible with a space-time structure. To define this compatibility, we need a few concepts at hand. The first is that of an \emph{ordered random variable} (or ORV) which can be seen as an abstract version of space-time random variables. Here, we model space-time simply as a partially ordered set where elements appearing earlier in the order denote space-time points that are in the past of elements appearing later in the order, and unordered elements correspond to space-like separated points. More formally, we have the following.

\begin{definition}[Ordered random variable (ORV)]
\label{def: ORV}
Each ORV $\mathcal{X}$ is defined by a pair $\mathcal{X}:=(R(\mathcal{X}),O(\mathcal{X}))$ where $R(\mathcal{X})$ (equivalently denoted by the corresponding, non-caligraphic letter $X$) is a random variable and $O(\mathcal{X})\in \mathcal{T}$ specifies the location of $X$ with respect to a partially ordered set $\mathcal{T}$.
\end{definition}

We then have the following definition of what it means to embed a set of RVs in a space-time.\footnote{The definition of embedding provided here is a simplified version of the original definition proposed in~\cite{VilasiniColbeckPRA}. There, the embedding also involves an additional concept of an accessible region of an ORV, which corresponds to the subset of space-time locations at which it possible to find copies of that RV. Within standard relativity theories, the accessible region corresponds to the future light cone of the ORV, but in our general framework, this could be any subset of locations. For the results of present paper, we set the accessible region of an ORV to be the future light cone of the associated space-time location (while including that location).}

\begin{definition}[Embedding]
\label{definition: embedding}
Given a set of RVs $S$, an \emph{embedding of $S$} in a partially ordered set $\mathcal{T}$ produces a corresponding set of ORVs $\mathcal{S}$ by assigning a location $O(X)\in\mathcal{T}$ to each $X\in S$. An embedding of a set of RVs is called \emph{non-trivial} if no two RVs $X$ and $Y$ such that $X$ affects $Y$ are assigned the same location in $\mathcal{T}$. 
\end{definition}

\emph{Notation:} We use $\prec$, $\succ$ and $\nprec\nsucc$ to denote the order relations for a given partially ordered set $\mathcal{T}$, where for $\alpha$, $\beta\in \mathcal{T}$, $\alpha \nprec\nsucc \beta$ corresponds to $\alpha$ and $\beta$ being unordered with respect to $\mathcal{T}$. This is not to be confused with $\alpha=\beta$ which corresponds to the two elements being equal. These relations carry forth in an obvious way to ORVs and we say for example that 2 ORVs $\mathcal{X}$ and $\mathcal{Y}$ are ordered as $\mathcal{X}\prec \mathcal{Y}$ iff $O(\mathcal{X})\prec O(\mathcal{Y})$. In addition, $\mathcal{X}=\mathcal{Y}$ is equivalent to $X=Y$ and $O(\mathcal{X})=O(\mathcal{Y})$.

 \begin{definition}
 \label{definition: incfuture}
 The \emph{inclusive future} of an ORV is the set 
 \begin{equation*}
    \overline{\mathcal{F}}(\mathcal{X}):=\{\alpha \in \mathcal{T}: \alpha\succeq O(\mathcal{X})\}.
\end{equation*}
The inclusive future of a set $\cS$ of ORVs is given as the intersection of the inclusive futures of the ORVs in the set, and captures the region where they can be jointly accessed\footnote{In the original framework~\cite{VilasiniColbeckPRA} a separate concept called the \emph{accessible region} of an ORV is defined which is a priori distinct from the future of the ORV. Here, we have set the accessible region to be equal to inclusive future, such that the compatibility condition captures the impossibility of signalling outside the space-time future.}
 \begin{equation*}
    \overline{\mathcal{F}}(\mathcal{S}):=\bigcap_{\cX\in \cS}  \overline{\mathcal{F}}(\mathcal{X}).
\end{equation*}
 \end{definition}

 Once we assign space-time locations to observed RVs in a causal model to promote them to ORVs, a natural question that arises is whether the affects relations of the model can be used to signal outside the space-time's future.
 For this, notice that a higher-order affects relation $X$ affects $Y$ given do$(Z)$ (representing that $X$ can signal to $Y$ given information about interventions performed on $Z$) between disjoint sets of RVs essentially captures the most general way of signalling in such a causal modelling framework\footnote{In our original paper~\cite{VilasiniColbeckPRA} we also define conditional higher-order affects relations of the form $X$ affects $Y$ given $\{\mathrm{do}(Z),W\}$ where we have an additional conditioning on values of some set $W$ of nodes which are not intervened upon, this captures signalling from $X$ to $Y$ given an intervention on $Z$ and a post-selection on $W$. Our compatibility condition also takes this into account. In this paper, these conditional relations are not relevant and we therefore don't consider them, however they are defined in Appendix~\ref{appendix: framework} for completeness.}. When embedded in space-time, this becomes an affects relation $\cX$ affects $\cY$ given do$(\cZ)$ between corresponding sets of ORVs. Now, if Alice, who has access to $X$ and can intervene on it, wishes to signal to Bob through this affects relation,
 Bob needs to jointly access $Y$ and $Z$ to receive the signal from Alice. Once these variables are embedded in space-time, in order to ensure that the affects relation does not enable signalling outside the space-time future, we would require that the joint future of $\cY$ and $\cZ$ (which is where Bob can jointly access $Y$ and $Z$ to receive the signal) must be contained in the future of $\cX$ (where Alice performs her interventions). Further, we should only demand this for irreducible affects relations, as we would otherwise be imposing unnecessarily strong constraints. For instance in our example with $X_1\longrsquigarrow Y$ and $X_2$ with no in or outgoing edges, if we impose compatibility for the reducible relation $X_1X_2$ affects $Y$ then we would be requiring $Y$ to be embedded in the future of $X_2$ even though $X_2$ is completely causally disconnected from $Y$. However, if we reduce all our reducible relations to their minimal irreducible form, which in this case is $X_1$ affects $Y$, then compatibility only requires $Y$ to be embedded in the future of the location where we embed $X_1$, as we would expect. Therefore, the concept of irreducibility is key for ensuring that our compatibility condition is both necessary and sufficient for no superluminal signalling.

\begin{definition}[Compatibility of a set of affects relations with an embedding in a partial order ($\mathbf{compat}$)]
\label{definition: compatposet}
Let $\mathcal{S}$ be a set of ORVs formed by embedding a set of RVs $S$ in a space-time $\mathcal{T}$ with embedding $\mathscr{E}$. Then a set of affects relations $\mathscr{A}$ is said to be \emph{compatible} with the embedding $\mathscr{E}$ if whenever $X$ affects $Y$ given do$(Z)$ belongs to $\mathscr{A}$, and is irreducible with respect to the affects relations in $\mathscr{A}$, then $\overline{\mathcal{F}}(\cY)\bigcap \overline{\mathcal{F}}(\cZ)\subseteq \overline{\mathcal{F}}(\cX) $ with respect to $\mathscr{E}$. 
\end{definition}

\begin{definition}[Compatibility of a causal model with an embedding in a partial order]
We say that a causal model over a set of RVs $S$ is compatible with an embedding in a partial order if the set of affects relations $\mathscr{A}$ implied by the causal model are compatible with the embedding (cf.\ Definition~\ref{definition: compatposet}).
\end{definition}

\section{Formal definitions of Bell scenarios and corresponding space-time embeddings}

Based on the framework reviewed in the previous section, we now provide a formal definition of Bell scenarios, as well as the associated space-time embeddings, that will be used in the rest of the paper. 

\begin{definition}[Bi and tripartite Bell scenarios]
\label{def: Bell_scenario}
A bi (tri) partite Bell scenario refers to any experiment involving two (three) parties Alice and Bob (and Charlie) who share a common system $\Lambda$ and measure their respective sub-system as given by their measurement settings $A$, $B$ (, $C$) obtaining the corresponding measurement outcomes $X$, $Y$ (, $Z$) respectively. We assume that the settings are freely chosen i.e., have no causes which are relevant to the experiment. Any causal model associated with a bi (tri) partite Bell scenario will have the setting variables $A$, $B$ (, $C$) as parentless nodes (capturing the free choice assumption) and $\Lambda$ as a common cause of the outcomes $X$, $Y$ (, $Z$) and itself being parentless. The settings and outcomes are always observed classical nodes, while $\Lambda$ may be observed or unobserved and can thus be described using a general, non-classical theory. 
\end{definition}

The above is a minimal definition of a Bell scenario which does not fix the (information-theoretic) causal structure. In particular, the typical causal structures associated with Bell scenarios (Figures~\ref{fig:Bell_bipart_CS} and~\ref{fig: Bell_tripart2}), satisfy this definition. More generally, the definition also allows for causal structures where the setting of one party (e.g., $B$) is the cause of the outcome of another party (e.g., $X$), as well as cyclic causal structures where two different outcome variables can mutually influence each other. Note that although it would be natural to require that the settings are causes of the respective outcomes, our minimal definition of Bell scenarios does not enforce this, and allows situations where a party's outcome is not caused by their setting.

\begin{definition}[Space-time configuration for standard Bell scenarios]
    \label{def: spacetime_config_bipart}
    A bipartite (or tripartite) Bell scenario is said to take place in a standard space-time configuration if the random variables corresponding to the measurement settings $A$, $B$ (, $C$) and outcomes $X$, $Y$ (,$Z$) are embedded in space-time such that: $\cA\prec \cX$, $\cB\prec \cY$ (and $\cC\prec \cZ$) and the remaining pairs of space-time embedded variables are space-like separated (related by $\not\prec \not \succ$), where the order relations $\prec, \succ$ are defined by the light cone structure of the space-time. Such a space-time configuration is illustrated in Figure~\ref{fig:Bell_bipart_sptime} (and Figure~\ref{fig: Bell_tripart1}).
\end{definition}

Notice that our definition of a Bell scenario does not a priori forbid superluminal causal influences when the Bell scenario takes place in the standard space-time configurations. This will allow us to derive non-trivial statements about relativistic causality principles in Bell scenarios without assuming causal structures (such as those of Figures~\ref{fig:Bell_bipart_CS} and~\ref{fig: Bell_tripart2}) that already enforce these principles once embedded in space-time. Further details on how we define and model information-theoretic causal structures and their space-time embeddings can be found in Section~\ref{sec: framework}. 

The jamming space-time configuration in a tri-partite Bell scenario, is slightly different from the standard configuration as it involves as additional condition, as defined below.

\begin{definition}[Space-time configuration for jamming tripartite Bell scenario]
\label{def: spacetime_config_jamming}
   A tripartite Bell scenario is said to take place in a jamming space-time configuration if the ORVs corresponding to the measurement settings $A$, $B$, $C$ and outcomes $X$, $Y$, $Z$ are embedded in space-time such that:  $\mathcal{A}\prec \mathcal{X}$, $\mathcal{B}\prec \mathcal{Y}$, $\mathcal{C}\prec \mathcal{Z}$,  $\overline{\mathcal{F}}(\mathcal{X})\bigcap\overline{\mathcal{F}}(\mathcal{Z})\subseteq \overline{\mathcal{F}}(\mathcal{B})$, all other pairs of ORVs are space-like separated. Such a space-time configuration is illustrated in Figure~\ref{fig: Bell_tripart1}.
\end{definition}

Note that $X$, $B$ and $Z$ being spacelike separated is not enough to ensure that $\overline{\mathcal{F}}(\mathcal{X})\bigcap\overline{\mathcal{F}}(\mathcal{Z})\subseteq \overline{\mathcal{F}}(\mathcal{B})$, even if the spatial coordinates of $X$, $B$ and $Z$ lie on a straight line except in $(1+1)$ dimensional Minkowski space-time --- see Appendix~\ref{app:ST} for a discussion of the possible space-time locations in $(1+1)$ and $(2+1)$D Minkowski space-time.

\section{Causal modelling analysis of jamming and implications for relativistic causality}
\label{sec: jamming_sig}

In this section, we analyse jamming correlations in a Bell scenario using a causal modelling approach. We then show that in a Bell scenario where $\Lambda$ is an observed node and which takes place in the space-time configuration of Figure~\ref{fig: Bell_tripart1}, jamming correlations will lead to superluminal signalling. 

\subsection{Jamming requires fine-tuning and superluminal causal influence}
\label{sec: jam_finetune_superlumcaus}

Using our causal modelling framework formulated under minimal assumptions, the following observations about correlation and causation in scenarios with jamming follow immediately. Firstly, jamming correlations can only be explained by a fine-tuned causal model; secondly, jamming in the space-time configuration of~\cite{Grunhaus1996, Horodecki2019} necessarily involves superluminal causal influences (whether or not it involves superluminal signalling).

\begin{restatable}{lemma}{JamCause}
\label{lemma: jamming_causation}
In any tripartite Bell scenario (cf.\ Definition~\ref{def: Bell_scenario}) where $B$ is neither a cause of $X$ nor of $Z$, it is impossible to generate jamming correlations in any theory i.e., irrespective of the nature (classical, quantum, post-quantum) of the common cause $\Lambda$. 
\end{restatable}

\begin{corollary}
The presence of jamming in a tripartite Bell scenario taking place in a jamming space-time configuration (Definition~\ref{def: spacetime_config_jamming}) implies the existence of superluminal causal influence between at least one pair of space-like separated parties. 
\end{corollary}

\begin{restatable}{lemma}{JamFinetune}
\label{lemma: jamming_finetune}
Any causal model reproducing the jamming correlations must be fine-tuned, irrespective of the theory describing the common cause $\Lambda$. In other words, there exists no faithful causal model (classical or non-classical) in any tripartite Bell scenario that generates jamming correlations.
\end{restatable}

Proofs of these results can be found in Appendix~\ref{appendix: proof_jamming}.

\begin{figure}[t!]
    \centering
 \includegraphics[]{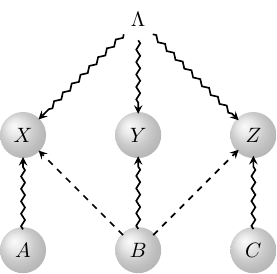}
    \caption{{\bf Causal structure of the tripartite Bell scenario with jamming} As shown in Lemma~\ref{lemma: jamming_causation}, in order for a causal model over the settings and outcomes of a tripartite Bell scenario to admit jamming correlations, the original tripartite Bell causal structure of Figure~\ref{fig: Bell_tripart2} must be modified to allow for causal influence from $B$ to $X$ and/or $Z$. Further, as shown in Lemma~\ref{lemma: jamming_finetune}, assuming that the settings $A$, $B$, $C$ are parentless nodes both these causal influences must be fine-tuned, hence represented by dashed arrows.}
    \label{fig: tripart_jamming}
\end{figure}

\subsection{Jamming can lead to superluminal signalling}
\label{sec: jammingproof}

Here we show that jamming can indeed lead to superluminal signalling i.e., that \hyperref[cond: NS3p]{(NS3$'$)} being satisfied in a tripartite Bell scenario is not sufficient for ruling out superluminal signals (contrary to the claims reviewed in Section~\ref{sec:jam_review}).  We first illustrate the possibility of signalling by constructing a causal model for the simplest jamming scenario, where Alice and Charlie have no settings and Bob has no outcome. Here the only relevant variables are Alice and Charlie's outcomes $X$ and $Z$, and Bob's setting $B$. We show this simplest jamming scenario can already lead to superluminal signalling in the space-time embedding considered in~\cite{Grunhaus1996, Horodecki2019} i.e., that of Definition~\ref{def: spacetime_config_jamming}. This protocol captures the main intuition behind our general Theorem~\ref{theorem: jamming} that will follow.

\bigskip
{\bf A causal model for jamming.}
Consider the causal structure shown in Figure~\ref{fig: jammingproof} where $\Lambda$ and $B$ are direct causes of $X$ and $\Lambda$ is a direct cause of $Z$. 
Suppose that $\Lambda$ is observed (and hence a classical RV) and all four variables, $\Lambda$, $X$, $B$ an $Z$ are binary and uniformly distributed. Then a causal model over this causal structure is given by specifying functions from the parents of $X$ ($B$ and $\Lambda$) to $X$ and from the parents of $Z$ (just $\Lambda$) to $Z$. We take $X=\Lambda\oplus B$ and $Z=\Lambda$. This immediately gives the correlation $B=X\oplus Z$ between Bob's setting $B$ and Alice and Charlie's outcomes $X$ and $Z$. Moreover, we can verify that these correlations do satisfy \hyperref[cond: NS3p]{(NS3$'$)} i.e., $P(X|B)=P(X)$ and $P(Z|B)=P(Z)$ and also admit jamming since $P(XZ|B)\neq P(XZ)$. Since $B$ is parentless, this is equivalent to: $B$ does not affect $X$ (hence the causal arrow from $B$ to $X$ is dashed in Figure~\ref{fig: jammingproof}), $B$ does not affect $Z$ but $B$ affects $XZ$. These affects relations are compatible with the space-time configuration of the figure.

However, since $\Lambda$ is observed, we additionally have the affects relation $B\Lambda$ affects $X$ and $\Lambda$ affects $Z$ (hence the solid causal arrow), along with $\Lambda$ does not affect $X$ (hence the dashed causal arrow). The affects relation $B\Lambda$ affects $X$ is irreducible in this model (since $B$ affects $X$ given do$(\Lambda)$ and $\Lambda$ affects $X$ given do$(B)$) and would lead to superluminal signalling unless $X$ is embedded in future of $B$ and of $\Lambda$. However, in the given space-time configuration (which is the one proposed in~\cite{Grunhaus1996, Horodecki2019}), the ORVs $\cX$ and $\cB$ are space-like separated and therefore the aforementioned affects relation leads to superluminal signalling in this space-time configuration, irrespective of the space-time location of $\Lambda$. The following theorem (proven in Appendix~\ref{appendix: proof_jamming}) generalises this idea and establishes the possibility of superluminal signalling using an observed $\Lambda$ in any tripartite Bell scenario with jamming. 

\begin{figure}[t!]
 \centering
  \subfloat[\label{fig:jammingsptime}]{ \includegraphics[]{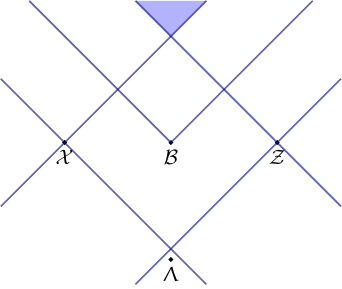}} \qquad\subfloat[\label{fig: jammingobs}]{ \includegraphics[]{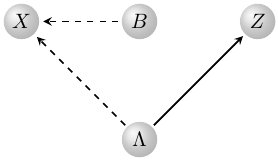}}\qquad\subfloat[\label{fig: jammingtable}]{ \includegraphics[]{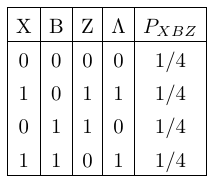}}
    \caption{\textbf{A simple example of how jamming as described in~\cite{Grunhaus1996, Horodecki2019} can lead to superluminal signalling:} \textbf{(a)} Space-time embedding of the variables $X$, $B$, $Z$ and $\Lambda$ suggested in~\cite{Grunhaus1996, Horodecki2019}, where the joint future of the space-time variables $\cX$ and $\cZ$ (blue region) is contained in the future of $\cB$. Here, $B$ corresponds to the input of a party, Bob while $X$ and $Z$ are the outputs of Alice and Charlie respectively. In this configuration,~\cite{Grunhaus1996, Horodecki2019} claim that there will be no superluminal signalling whenever the correlations $P(XZ|B)$ satisfy \hyperref[cond: NS3p]{(NS3$'$)} (where $A$, $C$ and $Y$ are taken to be trivial). \textbf{(b)} Causal structure in which the classical causal model described in the text satisfies \hyperref[cond: NS3p]{(NS3$'$)} and admits jamming, but leads to superluminal signalling in the space-time configuration if $\Lambda$ is observed. This shows that the claim does not hold in general, and it is necessary to assume that $\Lambda$ must be unobserved and fundamentally inaccessible in any theory that allows for jamming in this space-time configuration. \textbf{(c)} The correlations obtained from the causal model, which satisfy \hyperref[cond: NS3p]{(NS3$'$)}.}
    \label{fig: jammingproof}
\end{figure}

\begin{restatable}{theorem}{JammingLambda}
\label{theorem: jamming}
Consider a tripartite Bell scenario (Definition~\ref{def: Bell_scenario}) where the outcomes $X$, $Y$ and $Z$ do not causally influence any systems relevant to the experiment and thus are modelled as childless nodes in the causal structure. Suppose that the scenario admits jamming (Definition~\ref{def: jamming}), and its causal model is embedded in Minkowski space-time as per the jamming space-time configuration (Definition~\ref{def: spacetime_config_jamming}). Then, irrespective of the space-time embedding of $\Lambda$, the causal model is incompatible with the space-time and leads to superluminal signalling whenever $\Lambda$ is an observed system.
\end{restatable}

{\bf Non-classical causal models exhibiting jamming.} The causal model for jamming described above is entirely classical. Quantum and post-quantum extensions of the same protocol which generate jamming correlations satisfying \hyperref[cond: NS3p]{(NS3$'$)} are also possible. For example, consider Alice and Charlie who share a bipartite state, measure it using the inputs $A$ and $C$ and obtain the outcomes $X$ and $Z$, and Bob making an input $B$ that can causally influence at least one of $X$ or $Z$. Suppose, depending on $B$, we would like Alice and Charlie to share either correlations arising from measurements on the Bell state $\ket{\psi_0}=\frac{1}{\sqrt{2}}(\ket{00}+\ket{11})$ or those arising from the same measurements on a different Bell state,  $\ket{\psi_1}=\frac{1}{\sqrt{2}}(\ket{01}+\ket{10})$. One causal structure that could be used for this is an extension of the one in Figure~\ref{fig: jammingobs} (where $\Lambda\longrsquigarrow X$, $B\longrsquigarrow X$ and $\Lambda\longrsquigarrow Z$) with additional causal influences $A\longrsquigarrow X$ and $C\longrsquigarrow Z$ and where $\Lambda$ now represents an unobserved, quantum node. In the causal model, $\Lambda$ could correspond to a bipartite quantum system in the state $\frac{1}{\sqrt{2}}(\ket{00}+\ket{11})$, $X$ is generated from the first subsystem of $\Lambda$, $A$ and $B$ by first applying a controlled {\sc not} on the subsystem of $\Lambda$ with $B$ as control, and then measuring the resulting state of the subsystem in the basis choice given by $A$, while $Z$ is obtained from the second subsystem of $\Lambda$ by directly measuring in the basis choice given by $C$. Then the parties would measure $\ket{\psi_{0}}$ whenever $B=0$ and $\ket{\psi_{1}}$ whenever $B=1$ as required and generate the required correlations. Similarly, $B$ can be used to decide whether Alice and Charlie share one relabelling of a PR box or another and we can analogously construct a post-quantum causal model that achieves this.\footnote{Similarly, one can devise protocols to realise other types of jamming correlations by allowing the underlying mechanisms (like the choice of local operation prior to measurement) to depend on $B$ in a way that this does not reflect in the local statistics. Further one can also consider jamming correlations where Bob's output $Y$ is non-trivially involved in the process, which we do not analyse here as it is not directly relevant to the results of this paper.}

\bigskip

{\bf The causal model as a physical protocol vs dynamics of a post-quantum theory.} The causal models described above are purely information-theoretic, and have been classified as classical, quantum or post-quantum depending on the nature of the node $\Lambda$. On the other hand, whether such a causal model is physically implementable in space-time, according to classical, quantum or some other theory, depends not only on the causal model but also on the space-time embedding. For instance, if every causal influence in the model flows from past to future in the space-time (in particular, this would mean that $\cB$ is embedded in the past light-cone of $\cX$), then the model would not lead to any superluminal causal influences relative this embedding and describes a physical protocol within classical, quantum or a post-quantum theory (according to the nature of the model). Notice that jamming correlations can be realised whether the model is classical, quantum or post-quantum. 

On the other hand, consider the classical causal model for jamming relating to the causal structure in Figure~\ref{fig: jammingobs} embedded in space-time according to to the jamming space-time embedding illustrated in Figure~\ref{fig:jammingsptime}. In order to avoid superluminal signalling in this case, it is necessary that $\Lambda$ is unobserved or inaccessible to physical interventions (as implied by Theorem~\ref{theorem: jamming}). Withing our framework, the proposed causal model is compatible with the space-time embedding of Figure~\ref{fig: jammingproof} if and only if $\Lambda$ is unobserved: the only if part is implied by Theorem~\ref{theorem: jamming} and the if part is immediate from applying the compatibility condition to the affects relations of the example (also shown explicitly in~\cite{VilasiniColbeckPRA, VilasiniColbeckPRL}). However, the causal influence from $B$ to $X$ is then superluminal and goes outside the future lightcone, which is not something we can realise within standard quantum theory. More explicitly, the correlations between Bob's setting $B$ and Alice and Charlie's outcomes $X$ and $Z$ arising in the above mentioned protocols cannot be realised in any tripartite quantum mechanical experiment where the three parties are space-like separated\footnote{This is due to the non-signalling property of quantum theory which will ensure that measurements on different subsystems in any tripartite Bell scenario always lead to correlations that satisfy the stronger non-signalling condition \hyperref[cond: NS3]{(NS3)} which forbids jamming.}. Therefore a theory that generates jamming correlations in a jamming space-time configuration is truly post-quantum even though it may admit a classical or quantum causal model.

\bigskip

{\bf Comparison to non-local hidden variable theories.} Theorem~\ref{theorem: jamming} demonstrates that post-quantum theories admitting jamming correlations in the space-time configuration of~\cite{Grunhaus1996, Horodecki2019} are similar to non-local hidden variable explanations of Bell non-local correlations in that they involve superluminal signalling at the level of certain underlying variables without allowing superluminal signalling at the observed level, if these variables remain inaccessible. Explicitly, in non-local HV models for bipartite Bell correlations, we have the no-signalling conditions $P(X|AB)=P(X|A)$ and $P(Y|AB)=P(Y|B)$ by construction even though one of $P(X|AB\Lambda)=P(X|A\Lambda)$ and $P(Y|AB\Lambda)=P(Y|B\Lambda)$ necessarily fails. If $\Lambda$ were entirely accessible, then Alice and Bob could exploit this to signal to each other in a space-like separated Bell experiment. Similarly, jamming correlations satisfy $P(X|ABC)=P(X|A)$ and $P(Z|ABC)=P(Z|C)$ along with $P(XZ|ABC)\neq P(XZ|AC)$. However we have shown that we must necessarily have at least one of $P(X|ABC\Lambda)\neq P(X|A\Lambda)$ or $P(Z|ABC\Lambda)\neq P(Z|C\Lambda)$ whenever $\Lambda$ is classical in order to have jamming. This implies that $\Lambda$ must be unobserved even when it is classical, in order to avoid superluminal signalling. Another point of similarity is that both non-local HV explanations of Bell inequality violating correlations as well as causal explanations of jamming correlations must necessarily be fine-tuned causal explanations. This is established for non-local HV explanations in~\cite{Wood2015} and for the jamming in case in our Lemma~\ref{lemma: jamming_finetune}. A crucial point of difference is however that in the former case, there exists a faithful explanation of the correlations through a quantum (or post-quantum) causal model while our results establish that any explanation of jamming correlations, through classical quantum or post-quantum causal model must necessarily be fine-tuned.

\section{Relativistic causality in Bell scenarios: accounting for arbitrary interventions}

In this section, we apply our framework for causal inference and compatibility with space-time to the bi- and tripartite Bell scenarios. The usual treatment of these scenarios, including discussions on relativistic causality violations, focuses on correlation constraints such as \hyperref[cond: NS2]{(NS2)}, \hyperref[cond: NS3]{(NS3)} and \hyperref[cond: NS3p]{(NS3$'$)}. More generally however, we would like to rule out relativistic causality violations even when arbitrary interventions are considered. Our formalism allows us to systematically treat interventions and derive stronger results for relativistic causality that can rule out superluminal signalling under all possible interventions on observed variables. Proofs of all results can be found in Appendix~\ref{appendix: proofs}.

\subsection{Non-signalling constraints and higher-order affects relations}
\label{sec: NS_affects}
We first start by casting the non-signalling constraints \hyperref[cond: NS2]{(NS2)} and \hyperref[cond: NS3p]{(NS3$'$)} (which are expressed in terms of correlations) in terms of higher-order affects relations which we have introduced in our framework for capturing the general interventions on observed nodes. Proofs of these theorems can be found in Appendix~\ref{appendix: proof_ns_affects}.

\begin{restatable}{theorem}{NSAffectsBipart}
\label{theorem: ns2_affects}
    In any bipartite Bell scenario (Definition~\ref{def: Bell_scenario}), the following non-affects relations are equivalent to the bipartite no-signalling conditions \hyperref[cond: NS2]{(NS2)}: 
    \begin{align}
    \label{eq: ns2_affects}
        \begin{split}
            B &\text{ does not affect } X \text{ given do}(A),\\
            A &\text{ does not affect } Y \text{ given do}(B).
        \end{split}
    \end{align}

\end{restatable}

\begin{restatable}{theorem}{NSAffectsTripart}
\label{theorem: ns3p_affects}
    In any tripartite Bell scenario (Definition~\ref{def: Bell_scenario}) the following non-affects relations are equivalent to the relaxed tripartite no-signalling conditions \hyperref[cond: NS3p]{(NS3$'$)}: 

\begin{align}
 \label{eq: ns3p_affects}
     \begin{split}
         &C \text{ does not affect } XY \text{ given do}(AB),\\
         &A \text{ does not affect } YZ \text{ given do}(BC),\\   
         &BC \text{ does not affect } X \text{ given do}(A),\\
          & AB \text{ does not affect } Z \text{ given do}(C).\\
     \end{split}
 \end{align}

\end{restatable}

Correlation does not imply causation or signalling, however, correlation constraints such as \hyperref[cond: NS2]{(NS2)}, \hyperref[cond: NS3]{(NS3)} and \hyperref[cond: NS3p]{(NS3$'$)} capture the idea of no-signalling between parties through interventions on freely chosen settings (Definition~\ref{def: Bell_scenario}). Higher-order affects relations that we have introduced in~\cite{VilasiniColbeckPRA} capture signalling/no-signalling constraints in more general scenarios beyond Bell-type situations, where we have no freely chosen or parentless variables and where signalling cannot be always captured by the correlations alone. The above theorems show that in the special case of Bell scenarios, some of the higher-order affects relations are equivalent to the correlation constraints \hyperref[cond: NS2]{(NS2)} and \hyperref[cond: NS3p]{(NS3$'$)}. In essence, this is because in any causal model (according to Definition~\ref{def: causalmodel}), correlation between a parentless variable $A$ and another variable $X$ implies causal influence from $A$ to $X$ as well as an affects relation $A$ affects $X$.

\subsection{Conditions for preserving relativistic causality}
\label{sec: Bell_relcaus}

Theorem~\ref{theorem: jamming} highlights that jamming correlations arising in tri-partite Bell scenarios in the relevant space-time configuration, can lead to superluminal signalling if the common cause is accessible to agents. This raises the question of whether additionally taking $\Lambda$ to be unobserved is sufficient for preserving relativistic causality. We analyse this question within the causality framework reviewed in Section~\ref{sec: framework} and find that in general, superluminal signalling is possible in the standard space-time configurations of a tripartite (/bipartite) Bell scenario even when we impose \hyperref[cond: NS3p]{(NS3$'$)} (/\hyperref[cond: NS2]{(NS2)}) along with the requirement that $\Lambda$ is unobserved. Intuitively, this is because \hyperref[cond: NS2]{(NS2)} and \hyperref[cond: NS3p]{(NS3$'$)} are constraints only on the correlations but signalling and causation are concepts that are not fully determined by correlations alone, but also require us to consider interventions. Some of the previous analyses of jamming~\cite{Horodecki2019} (cf.\ Section~\ref{sec:jam_review}) focus on the relativistic causality principle of ``no causal loops'', here we consider both this and the related relativistic principle ``no superluminal signalling''.  These principles do not imply one another in general, as demonstrated by our recent work~\cite{VilasiniColbeckPRL}, which highlights the importance of disentangling these different causality principles.

\subsubsection{Conditions for ruling out superluminal signalling}

The compatibility of a causal model with a space-time embedding (cf.\ Definition~\ref{definition: compatposet}) formalises the requirement of no superluminal signalling while taking into account general interventions. Here, we apply this to characterise conditions for no superluminal signalling in the bipartite and tripartite Bell scenarios, and the relevant configurations in Minkowski space-time. Proofs of all results of this subsection can be found in Appendix~\ref{appendix: proof_superlum}.

\bigskip
{\bf Accounting for interventions on outcomes.} In Bell scenarios where measurement settings are freely chosen, the usual non-signalling constraints on correlations capture the inability of different parties to communicate to each other through interventions on the setting variables. More generally, one can consider situations where it is impossible to signal between parties through interventions on settings, but still possible to signal once we also consider interventions on outcomes, which is not captured by the usual non-signalling constraints. 

By applying the causal inference formalism, we account for interventions on all observed variables (settings and outcomes) of the scenario. If the causal structure is one where the settings are free (parentless) then interventions on settings need not be explicitly considered (see Definition~\ref{def:post_intervention}). Typically, interventions on outcomes are not considered because the causal structure is taken to be of the form of Figures~\ref{fig:Bell_bipart_sptime} or~\ref{fig: Bell_tripart2} where the outcomes are assumed to be childless, i.e., to not be causes of any other variables relevant to the experiment. However, these causal structures are acyclic, and when embedded in Minkowski space-time in the usual way, all the causal arrows in these structures flow from past to future in the space-time and do not allow for superluminal signals. Thus relativistic causality violations such as superluminal signalling and causal loops are ruled out by assumption. Here, we have defined Bell scenarios without fixing the causal structure to be of this form, but only requiring that the setting variables are parentless (to model that they are freely chosen). This definition of Bell scenarios, and their standard space-time embedding does not rule out superluminal signalling or causal loops by construction. Furthermore, the following simple example illustrates that imposing \hyperref[cond: NS2]{(NS2)} and \hyperref[cond: NS3p]{(NS3$'$)} in such Bell scenarios is insufficient for ruling out signalling.

\begin{example}[\texorpdfstring{\hyperref[cond: NS2]{(NS2)}/\hyperref[cond: NS3]{(NS3)}}{NS2/3} are insufficient for ruling out superluminal signalling in Bell scenarios]
\label{example: NS_insuff}
Consider a causal structure over variables $A$, $B$, $X$ and $Y$ where $A\longrsquigarrow Y$ and $X\longrsquigarrow Y$ are the only edges (Figure~\ref{fig: examples_maintext}).
Consider a causal model over this causal structure where all variables are binary, $A$ and $X$ are uniformly distributed and $Y=A\oplus X$, while $B$ can have any distribution (note that it is a dummy variable, which is both parentless and childless). This causal structure satisfies our minimal definition of a bipartite Bell scenario (Definition~\ref{def: Bell_scenario}), and the associated model leads to probabilities $P(XY|AB)$ obeying the non-signalling conditions \hyperref[cond: NS2]{(NS2)} (since $A$ does not affect $Y$ given do$(B)$ and $B$ does not affect $X$ given do$(A)$, cf.\ Theorem~\ref{theorem: ns2_affects}). However we have $AX$ affects $Y$ (since $A$ and $X$ are parentless and correlated with $Y$), which enables Alice (who has access to $A$ and $X$) to signal to Bob (who has access to $Y$). This signal would be superluminal in any space-time configuration where Alice and Bob's measurements take place at space-like separated regions. The same construction can be trivially embedded in a tripartite Bell scenario, by choosing $C$ and $Z$ to be dummy variables, such that \hyperref[cond: NS3]{(NS3)} (and hence \hyperref[cond: NS3p]{(NS3$'$)} is also trivially satisfied and we still have the possibility of signalling from Alice to Bob. Notice that $\Lambda$ is irrelevant to this example, so imposing that $\Lambda$ is unobserved (which is necessary for avoiding superluminal signalling as shown in Theorem~\ref{theorem: jamming}) is insufficient to avoid superluminal signalling in this example. 
\end{example}

The above example emphasizes the need to account for interventions in addition to correlation constraints such as \hyperref[cond: NS2]{(NS2)} and \hyperref[cond: NS3]{(NS3)} to rule out superluminal signalling in Bell type scenarios in space-time, which we do in the following theorems. In the following, when we say there are no affects relations emanating from a set $S_1$ we mean that there are no (possibly higher-order) affects relations $S_1$ affects $S_2$ given do$(S_3)$ (cf.\ Definition~\ref{definition:HOaffects}) with $S_1$ being the first argument. 

\begin{restatable}{theorem}{SuperlumBipart}
\label{theorem: superlum_bipart}
Consider a bipartite Bell scenario according to Definition~\ref{def: Bell_scenario} where $\Lambda$ is unobserved. Let the observed nodes $\{A,B,X,Y\}$ of the causal structure be embedded in Minkowski space-time according to the standard bipartite configuration of Definition~\ref{def: spacetime_config_bipart}. Then the following conditions are necessary and sufficient for the causal model of the scenario to be compatible with such a space-time embedding.
\begin{enumerate}
    \item \label{cond:bi1} The correlations $P(XY|AB)$ satisfy \hyperref[cond: NS2]{(NS2)}. 
    \item \label{cond:bi2} There are no affects relations emanating from any subset of $\{X,Y\}$.
\end{enumerate}
\end{restatable}

In the above, we have provided necessary and sufficient conditions formulated at the level of the observed correlations and affects relations in the Bell scenario, which are operationally accessible quantities. However, often one makes assumptions about the underlying causal structure (which may include unobserved systems and may not always be directly operationally accessible). Such assumptions at the level of the causal structure are often much stronger than their counterparts at the level of affects relations, thus the above theorem has the following corollary.

\begin{corollary}
\label{corollary: bipart}
    In any bipartite Bell scenario (Definition~\ref{def: Bell_scenario}) where $\Lambda$ is unobserved and we assume that the outcomes $X$ and $Y$ are not causes of any variables relevant to the experiment, the bipartite correlation constraints \hyperref[cond: NS2]{(NS2)} are necessary and sufficient for no-superluminal signalling in the standard space-time configuration of such a scenario (Definition~\ref{def: spacetime_config_bipart}). 
\end{corollary}
The corollary follows because in any such causal structure, Condition~\ref{cond:bi2} of Theorem~\ref{theorem: superlum_bipart} is automatically satisfied: if $X$ and $Y$ are not causes of any other variables, then they also cannot have any affects relations emanating from them.  The causal structure of Figure~\ref{fig:Bell_bipart_CS} is of this form, and is the one that is typically used to describe physical bipartite Bell scenarios. Note however that Condition~\ref{cond:bi2} of Theorem~\ref{theorem: superlum_bipart} is weaker than assuming such a causal structure because is possible to have fine-tuned causal models defined on causal structures where the outcomes are a cause of other variables, without having any affects relations emanating from them i.e., while satisfying Condition~\ref{cond:bi2} of the theorem.

Moreover, $\Lambda$ being unobserved is also necessary in the bipartite Bell scenario when we wish to argue for no superluminal signalling using \hyperref[cond: NS2]{(NS2)} alone. This is because non-local hidden variable models explain Bell non-local correlations satisfying \hyperref[cond: NS2]{(NS2)} by allowing for superluminal influences between Alice and Bob, which would lead to superluminal signalling if $\Lambda$ was fully accessible. For example, in a Bohmian explanation of Bell inequality violating correlations, where $\Lambda$ includes the (usually hidden) particle positions, the non-signalling conditions \hyperref[cond: NS2]{(NS2)} are satisfied by construction, in particular we have $P(X|AB)=P(X|A)$. However this would generally also have $P(X|AB\Lambda)\neq P(X|A\Lambda)$, which would enable superluminal signalling between Alice and Bob, if they were able to access $\Lambda$.

\begin{restatable}{theorem}{SuperlumTripart}
\label{theorem: superlum_tripart}
Consider a tripartite Bell scenario according to Definition~\ref{def: Bell_scenario} where $\Lambda$ is unobserved. Let the observed nodes $\{A,B,C,X,Y,Z\}$ of the causal structure be embedded in Minkowski space-time according to a jamming configuration (Definition~\ref{def: spacetime_config_jamming}). Further, if $\overline{\mathcal{F}}(\cX)\cap \overline{\mathcal{F}}(\cZ)\not\subseteq \overline{\mathcal{F}}(\cY)$, the following conditions are necessary and sufficient for the causal model of the scenario to be compatible with such a space-time embedding.

\begin{enumerate}
    \item \label{cond:tri1} The correlations $P(XYZ|ABC)$ satisfy \hyperref[cond: NS3p]{(NS3$'$)}  
    \item \label{cond:tri2} There are no affects relations emanating from any subset of $\{X,Y,Z\}$.
\end{enumerate}

On the other hand, if $\overline{\mathcal{F}}(\cX)\cap \overline{\mathcal{F}}(\cZ)\subseteq \overline{\mathcal{F}}(\cY)$ also holds, then condition~\ref{cond:tri1} from above along with the following condition are necessary and sufficient for compatibility. 

\begin{enumerate}[1$'$]\setcounter{enumi}{1}
    \item \label{cond:tri2p} There are no affects relations emanating from any subset of $\{X,Y,Z\}$, except affects relations of the form $Y$ affects $S_1$ given do$(S_2)$ where $\{X,Z\}\subseteq S_1$. 
\end{enumerate}

\end{restatable}

\begin{corollary}
\label{corollary: tripart}
    In any tripartite Bell scenario (Definition~\ref{def: Bell_scenario}) where $\Lambda$ is unobserved and we assume that the outcomes $X$, $Y$ and $Z$ are not causes of any variables relevant to the experiment, the relaxed tripartite correlation constraints \hyperref[cond: NS3p]{(NS3$'$)} are necessary and sufficient for no-superluminal signalling in the jamming space-time configuration of such a scenario (Definition~\ref{def: spacetime_config_jamming}). 
\end{corollary}
The corollary follows because in any such causal structure, condition~\ref{cond:tri2} and condition~\ref{cond:tri2p} of Theorem~\ref{theorem: superlum_tripart} are automatically satisfied. Note that the causal structure of Figure~\ref{fig: tripart_jamming} is of this form, and is the one that is typically used to describe physical experiments implementing tripartite Bell scenarios. Therefore, Corollaries~\ref{corollary: bipart} and~\ref{corollary: tripart} illustrate how we recover the standard understanding of the non-signalling conditions as fully capturing non-violation of relativistic causality in usual treatment of bi/tripartite Bell scenarios. 

\begin{remark}
    For the above theorems, we have taken into account higher order affects relations of the form $S_1$ affects $S_2$ given do$(S_3)$ which capture signalling from $S_1$ to $S_2$ given information about interventions performed on $S_3$ (where these are disjoint subsets of observed nodes in some causal model). More generally, the concept of conditional higher-order affects relations was introduced in~\cite{VilasiniColbeckPRA}, these are relations of the form $S_1$ affects $S_2$ given (do$(S_3)$, $S_4$) which capture signalling from $S_1$ to $S_2$ given information about interventions performed on $S_3$ and when post-selecting on $S_4$ (without performing interventions). The definition of such affects relations is reviewed in Appendix~\ref{appendix: framework}. However, such conditional relations are not relevant for our main results because we are already able to find necessary and sufficient conditions for no superluminal signalling in terms of the simpler, non-conditional relations. Intuitively, this is because a conditional relation such as $S_1$ affects $S_2$ given (do$(S_3)$, $S_4$) always implies an unconditional relation $S_1$ affects $S_2 S_4$ given do$(S_3)$~\cite{VilasiniColbeckPRA}. 
\end{remark}

\begin{figure}
    \centering
  
 \includegraphics[]{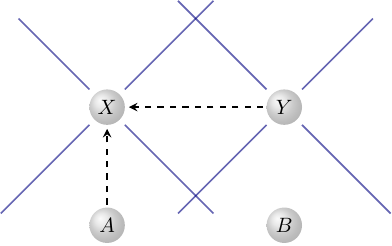}
    \caption{ Causal structure of Example~\ref{example: NS_insuff} along with the associated space-time embedding of the observed variables (circled) of this causal structure. We recall that in causal structures, a dashed causal arrow $N_1\xdashrightarrow{} N_2$ indicates that $N_1$ is a direct cause of $N_2$, but $N_1$ does not affect $N_2$ in the given causal model. This example involves superluminal causal influence from $Y$ to $X$, even though $Y$ does not affect $X$ (as captured by the dashed arrow) and the associated correlations satisfy the non-signalling constraints \hyperref[cond: NS2]{(NS2)}, the higher-order affects relation of this example enable superluminal signalling from Bob to Alice.}
    \label{fig: examples_maintext}
\end{figure}

\subsubsection{Conditions for ruling out causal loops}

We now turn our attention to the principle of no causal loops. This is another principle that is often associated with relativistic causality, and comprises the relativistic causality condition proposed in~\cite{Horodecki2019} (Definition~\ref{definition: horodecki}). However, in the absence of a rigorous definition of ``cause'' and ``events'' the meaning and scope of this principle, and claims based on them~\cite{Horodecki2019} remained unclear. Distinguishing between information-theoretic and spatio-temporal causality notions we see that in Minkowski space-time, we have no causal loops according to the spatio-temporal notion of causality.\footnote{Since the light-cone structure of the space-time forms a partially ordered set which defines a directed acyclic graph.} We can nevertheless consider cyclic information-theoretic causal structures to model scenarios where variables, possibly under the control of different parties, may causally influence each other, which are the kind of causal loops we wish to avoid in information processing protocols in Minkowski space-time. In Section~\ref{sec: framework}, we formalised and distinguished between two classes of information-theoretic causal loops (formulated in a theory-independent manner), namely affects loops and hidden causal loops which correspond to operationally detectable and operationally undetectable causal loops respectively. 

The previous section focused on the relativistic principle of no superluminal signalling, which is generally neither implied by nor implies the principle of no causal loops~\cite{VilasiniColbeckPRA, VilasiniColbeckPRL}. The following results address the sufficiency and necessity of the conditions of Theorems~\ref{theorem: superlum_bipart} and~\ref{theorem: superlum_tripart} (which were shown to be necessary and sufficient for no superluminal signalling in Bell scenarios) for the requirement of having no causal loops, while clearly distinguishing between affects and hidden causal loops. 

\begin{restatable}{theorem}{CausalLoopsBell}
\label{theorem: loops1}
    Condition~\ref{cond:bi2} of Theorem~\ref{theorem: superlum_bipart} and condition~\ref{cond:tri2p} of Theorem~\ref{theorem: superlum_tripart} (and therefore the stronger condition~\ref{cond:tri2} of the same theorem) are sufficient for having no affects causal loops in bi- and tripartite Bell scenarios respectively. They are however insufficient for ruling out hidden causal loops. 
\end{restatable}

\begin{restatable}{theorem}{CausalLoopsBell2}
\label{THEOREM: LOOPSBELL2}
    Condition 1 of Theorems~\ref{theorem: superlum_bipart} and~\ref{theorem: superlum_tripart} are insufficient for having no causal loops in bi- and tripartite Bell scenarios respectively. Neither of the conditions~\ref{cond:bi1} nor~\ref{cond:bi2} of Theorem~\ref{theorem: superlum_bipart}, and neither of the conditions~\ref{cond:tri1},~\ref{cond:tri2} nor~\ref{cond:tri2p} of Theorem~\ref{theorem: superlum_tripart} 
    are necessary for having no causal loops in the bi- and tripartite scenarios respectively. These statements apply to both affects and hidden causal loops.
\end{restatable}
The first of the above two theorems in proven in Appendix~\ref{appendix: proof_loop} and the statements of the second are proven through counter-examples which are detailed in Appendix~\ref{appendix: examples}. 
In particular, these results highlight that condition 1 of the two theorems, namely the bipartite and relaxed tripartite non-signalling conditions \hyperref[cond: NS2]{(NS2)} and \hyperref[cond: NS3p]{(NS3$'$)} are neither sufficient nor necessary for having no affects causal loops, contrary to the claims reviewed in Section~\ref{sec:jam_review}. 

The conditions featuring in the above theorems are at the level of affects relations and correlations, the stronger condition that the outcomes in the Bell scenario have no outgoing arrows in the causal structure is in fact sufficient for ruling out affects and hidden causal loops because this implies that the causal structure is itself acyclic. However this assumption is too strong as it rules out all causal loops by construction, and independently of whether \hyperref[cond: NS2]{(NS2)} or \hyperref[cond: NS3p]{(NS3$'$)} are satisfied in the Bell scenario. In general, to rule out hidden causal loops one would have to make strong enough assumptions on the causal structure to ensure that it is acyclic, and once we have these assumptions, the fact that there are no causal loops will be independent of the causal model we construct or the correlations it generates. Therefore it is more interesting to consider the conditions for ruling out affects causal loops, as these can be characterised purely at the level of observable data on interventions and correlations. And once these are ruled out, we know that there will exist an acyclic causal structure that generates the affects relations of the scenario. We refer the interested reader to~\cite{VilasiniColbeckPRA}  for a discussion on the different necessary and sufficient conditions that one may consider for ruling out such loops in general (and not necessarily Bell-type) scenarios.

\bigskip

{\bf Bell scenarios with pre or post-processing.} The above results apply to superluminal signalling and causal loops involving the variables of the original Bell scenario. More generally, one may consider that the scenario is part of a larger protocol which involves information processing steps before or after the Bell experiment under consideration. In that case, additional conditions involving observed variables appearing in these pre and post processing steps will also be needed for ensuring that we have no superluminal signalling or no causal loops in the extended protocol. For example, we may have $A$ does not affect $Y$ in a bipartite Bell scenario but there may exist an additional variable $W$ such that $AW$ affects $Y$ which may allow superluminal signalling from Alice to Bob if Alice can access the ORVs $\cA$ and $\cW$ at a space-time location that is space-like to $\cY$. 

The framework we have developed in~\cite{VilasiniColbeckPRA} is general and the compatibility condition proposed there captures necessary and sufficient conditions for avoiding superluminal signalling in arbitrary information processing protocols (even those involving causal loops in the causal model). These definitions can be applied to a given scenario to derive relativistic causality conditions specific to that scenario (as we have done for Bell scenarios in the above theorems). Sufficient conditions for ruling out superluminal signalling will depend on several factors such as which variables are accessible to agents, which correlations can be observed, which interventions are physically allowed, which are the variables and space-time locations between which we wish to avoid superluminal signalling/causal loops. While the set of all conditions in an arbitrary protocol may end up looking rather complicated, our framework nevertheless provides a systematic way to define and check these conditions while accounting for all these factors.

 \begin{figure}[t!]
    \centering
 \includegraphics[]{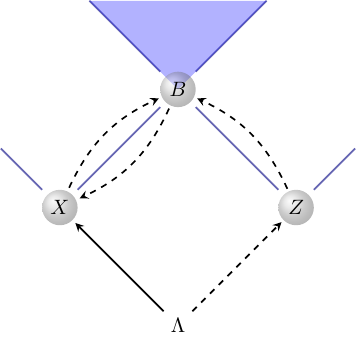}
    \caption{{\bf A causal loop that does not lead to superluminal signalling in Minkowski space-time~\cite{VilasiniColbeckPRL}.} The operational causal structure associated with the model is given in black, circled variables are observed nodes, while uncircled ones are not and the black arrows denote causation. Space-time information is given in blue with time along the vertical and space along the horizontal axis. The solid lines represent light-like surfaces and the shaded region corresponds to the joint future of the space-time locations of $X$ and $Z$ in all cases. In~\cite{VilasiniColbeckPRL}, we have constructed a causal model associated with this cyclic causal structure which satisfies the following properties 1. the model leads to correlations satisfying the conditions \hyperref[cond: NS3p]{(NS3$'$)} 2. it leads to affects relations that allow the cyclic causal influences to be operationally certified i.e., it is an affects causal loop and 3. it does not lead to superluminal signalling in the given space-time embedding. Notice that even though the correlations satisfy the non-signalling constraints \hyperref[cond: NS3p]{(NS3$'$)}, this example does not correspond to a Bell scenario according to our Definition~\ref{def: Bell_scenario} since $B$ is not freely chosen (it is not a parentless node).}
   \label{fig: loop_PRL}
\end{figure}

\begin{remark}[Different definitions of free choice]
In this paper, we have modelled the freely chosen variables (such as measurement settings) as parentless nodes in our information-theoretic causal structure, which is a common way to model free choice within causal modelling approaches to Bell scenarios. This ensures (by the d-separation property) that if $B$ is freely chosen, then it can only be correlated with variables $Y$ that are descendants of $B$ in the causal structure. When a causal model is embedded in space-time such that every directed edge in its causal structure flows from past to future in the space-time (i.e., there is no superluminal causation), then this model of free choice implies the definition of free choice proposed by Colbeck and Renner (CR)~\cite{CR_ext,CR2013} which states that a freely chosen variable $B$ can only be correlated with variables $Y$ that are in its future lightcone. In Bell scenarios with jamming correlations however, the principle of no superluminal causation is violated (cf.\ Lemma~\ref{lemma: jamming_causation}), and in general, the settings of the scenario would not be free according to the CR notion of free choice (although they are, according to the causal modelling definition). In~\cite{Horodecki2019}, a relaxation of the CR free choice definition was proposed, in order to be applicable to jamming correlations. This states that a freely chosen variable $B$ can be correlated with sets of variables whose joint future is contained in the future of $B$. Using this definition of free choice, as well, the claim of~\cite{Horodecki2019} that \hyperref[cond: NS3p]{(NS3$'$)} is necessary and sufficient for ruling out causal loops does not hold. This is because, the causal loop of~\cite{VilasiniColbeckPRL} (see also Figure~\ref{fig: loop_PRL}) leads to correlations satisfying \hyperref[cond: NS3p]{(NS3$'$)} and moreover, according to the definition of free choice suggested in~\cite{Horodecki2019}, the variable $B$ in this example would be regarded as a freely chosen variable even though it is involved in a causal loop.  In contrast, according to the definition used in the present work, that free choices are parentless, the variable $B$ of this example is not a free choice. 
\end{remark}

\section{Conclusions and outlook}
\label{sec: conclusions}

This work formalises the connections between non-signalling constrains in Bell scenarios and relativistic principles of causality through an approach based on causal modelling and causal inference~\cite{VilasiniColbeckPRA, VilasiniColbeckPRL}. This enabled us to go beyond previous analyses that typically only consider correlations, to consider both correlations and the effects of arbitrary interventions which are captured by the concept of \emph{higher-order affects relations}. In particular, this highlights the precise assumptions under which our usual understanding of relativistic causality in Bell scenarios holds true as well as how we can formulate such fundamental principles without these assumptions. Moreover, we have shown that post-quantum theories admitting jamming non-local correlations in Bell scenarios necessarily violate certain fundamental causality principles (both for information-theoretic and spatio-temporal causality notions), which has important consequences for the physicality of such theories. We note here another recent work pointing to problems with the physicality of jamming correlations, through arguments based on monogamy relations~\cite{Weilenmann2023}. Together with~\cite{VilasiniColbeckPRA, VilasiniColbeckPRL}, this work lays the basis for a rigorous causal modelling approach to studying causation, correlations and interventions in quantum and post-quantum theories, as well as their connection to relativistic causality principles in a space-time. We now outline some interesting future directions stemming from our work. 

\bigskip

{\bf Jamming theories and physical principles.} An important motivation for analysing post-quantum theories is to gain deeper insights into the physical principles that single out quantum theory. Our results of Section~\ref{sec: jamming_sig} do so by bringing to light further properties of post-quantum theories that admit jamming non-local correlations, such as the violation of certain causality principles (and the potential preservation of others). An interesting future direction would be to formalise post-quantum jamming theories in a similar manner to generalised probabilistic theories which will enable a further systematic study of operational principles that are satisfied/not satisfied in these theories. One immediate challenge in doing this is that the post-quantum jamming scenario requires a particular space-time configuration, while GPTs can be formalised independently of a space-time structure. Disentangling the notions of causality and space-time through operational considerations, as here is a first step towards such a more general characterisation of post-quantum jamming theories.

\bigskip

{\bf Other principles of relativistic causality.}
In this work, we have considered three principles associated with relativistic causality: no superluminal causation, no superluminal signalling and no causal loops. We have discussed how these are distinct, for instance jamming theories violate the first but can be made consistent with the second, and there are scenarios that satisfy the second and violate the third~\cite{VilasiniColbeckPRL} as well as those satisfying the third and violating the second (Example~\ref{example: non-necc1}). 
Our framework~\cite{VilasiniColbeckPRA, VilasiniColbeckPRL} also enables other relativistic causality principles to be formalised and their relationships to be studied. Other principles include (a) no signalling to the past, (b) no causation to the past, (c) a space-time variable can be accessed only within its future lightcone, or (d) a space-time variable can be accessed  everywhere within its future lightcone and nowhere else. Principles (a)--(d) are distinct from each other; (c) and (d) are distinct because a theory where a space-time variable (more generally, ordered random variable in our framework) can only be accessed within a strict subset of the future lightcone would impose different necessary and sufficient conditions for ensuring no superluminal signalling that a theory where it can be accessed everywhere in the future light-cone (see~\cite{VilasiniColbeckPRA} for details). Moreover the relationships between these principles that hold within special relativity theory need not hold within this more general framework, as witnessed by the example of causal loops without superluminal signalling~\cite{VilasiniColbeckPRL}. An interesting avenue for future work, which also relates to the questions in the previous paragraph, is to apply this approach to study the relationship between these principles in different physical theories. 

\bigskip

{\bf Relativistic information processing tasks in post-quantum theories.}
A closely related avenue for future research, relating to the previous paragraphs is to study relativistic information processing tasks in different post-quantum theories. In the context of quantum theory, it has been shown that when considering information processing tasks (such as cryptographic tasks) distributed in space-time and imposing relativistic principles of causality, new possibility and impossibility results emerge due to the trade-offs between the additional constraints imposed by relativistic principles both on honest players and on adversaries~\cite{Kent_RBC, CK1, Colbeck2, Vilasini_crypto, sandfuchs2023}. In a similar manner, one could consider tasks within post-quantum theories constrained by different relativistic principles, which can also provide an operational way to distinguish between different relativistic causality principles in a theory-independent manner. Ref.~\cite{Horodecki2019} suggests that theories with jamming correlations can violate the monogamy of non-local correlations and thus offer a different information-processing potential compared to quantum theory and other post-quantum theories such as box-world. Our work suggests that a distinguishing feature between jamming and these other theories is that the former is constrained only by no superluminal signalling but the latter also respect the stronger principle no superluminal causation. In an upcoming work~\cite{PRLinprep}, we further analyse the distinction between these two principles in terms of information-processing. 

\bigskip

{\bf Causal inference using higher-order affects relations.}
Higher-order affects relations introduced in~\cite{VilasiniColbeckPRA, VilasiniColbeckPRL} capture general interventions in a causal model and can allow more causal relations to be inferred from observed interventional data. In particular, already in a causal model over just 4 observed nodes (such as the bipartite Bell scenario), we can have almost 200 distinct higher-order affects relations of the form $S$ affects $S_1$ given do$(S_2)$ where $S$, $S_1$, $S_2$ are arbitrary disjoint subsets of the observed nodes with $S$ and $S_1$ being non-empty. In general, given two observed nodes $N_1$ and $N_2$ the following three conditions are distinct and do not imply one another~\cite{VilasiniColbeckPRA}, (a) $N_1$ and $N_2$ are correlated (i.e., $P_{\mathcal{G}}(N_1N_2)\neq P_{\mathcal{G}}(N_1)P_{\mathcal{G}}(N_2)$), (b) $N_1$ affects $N_2$ and (c) there exists a set of observed nodes $S$ such that $N_1$ affects $N_2$ given do$(S)$. This means that to perform causal inference in a general setting, one would typically have to consider, correlations, zeroth and higher-order affects relations. In the special case of Bell scenarios, Theorems~\ref{theorem: superlum_bipart} and~\ref{theorem: superlum_tripart} show that the usual non-signalling constraints on the correlations along with knowing the absence of some affects relations is necessary and sufficient to infer the absence of several other (possibly higher-order) affects relations which result from arbitrary interventions on the model, and to rule out all affects relations that could lead to superluminal signalling in the given space-time embedding. 

Thus these theorems which are about relativistic causality principles may be of independent interest for causal inference in non-classical theories, and where we can have causal fine-tuning.  In future work, it would be promising to study this problem beyond Bell scenarios, and to identify different sets of conditions/assumptions under which we can infer the presence/absence of higher-order affects relations from correlations. This is closely related to the problem of identifiability of causal structures which is widely studied within classical causal models. It is also of practical interest because it may be physically or ethically infeasible to perform certain interventions in a real-world setting, and it is therefore desirable to infer as many causal influences as possible while performing as few interventions as possible. The techniques of~\cite{VilasiniColbeckPRA, VilasiniColbeckPRL} and the present work may bear relevance beyond physics, and to other scientific disciplines where causal inference is widely used.

\bigskip
\noindent{\it Acknowledgements.}---Part of this research was carried out when VV was a PhD candidate supported by a Scholarship from the Department of Mathematics at the University of York. VV also acknowledges support from an ETH Postdoctoral Fellowship, the ETH Zurich Quantum Center, the Swiss National Science Foundation via project No.\ 200021\_188541 and the QuantERA programme via project No.\ 20QT21\_187724. 

\newpage
\appendix 

\section{Further technical details relating to our framework}
\label{appendix: framework}

\subsection{D-separation and causal models}
Here we provide formal definitions of some of the technical concepts that are relevant for our proofs which were only intuitively introduced in the main text. The first of these is d-separation, originally introduced in~\cite{Pearl2009, Spirtes2001}, which is based on the idea of blocked paths in a directed graph.

\begin{definition}[Blocked paths]
Let $\mathcal{G}$ be a directed graph in which $X$ and $Y\neq X$ are nodes and $Z$ be a
set of nodes not containing $X$ or $Y$.  A path from $X$ to $Y$ is
said to be \emph{blocked} by $Z$ if it contains either $A\longrsquigarrow W\longrsquigarrow B$
with $W\in Z$, $A\longlsquigarrow W\longrsquigarrow B$
with $W\in Z$ or $A\longrsquigarrow W\longlsquigarrow B$ such that neither $W$ nor any descendant of $W$ belongs to $Z$.
\end{definition}

\begin{definition}[d-separation]
\label{def: d-sep}
Let $\mathcal{G}$ be a directed graph in which $X$, $Y$ and $Z$ are disjoint
sets of nodes.  $X$ and $Y$ are \emph{d-separated} by $Z$ in
$\mathcal{G}$, denoted as $(X\perp^d Y|Z)_{\cG}$ (or simply $X\perp^d Y|Z$ if $\cG$ is obvious from context) if for every variable in $X$ and variable in $Y$ there is no path between them, or if every path from a variable in $X$ to a variable in
$Y$ is \emph{blocked} by $Z$. Otherwise, $X$ is said to be \emph{d-connected} with $Y$ given $Z$.
\end{definition}

For example, if $X$ and $Y$ are two disjoint subsets of nodes in $\cG$, $(X\perp^d Y)_{\cG}$ (here $Z=\emptyset$) is equivalent to saying that there are no directed paths between $X$ and $Y$, and also no common ancestors of $X$ and $Y$. In the graphs $X\longlsquigarrow Z \longrsquigarrow Y$ and $X\longrsquigarrow Z \longrsquigarrow Y$, we have $X\not \perp^d Y$ but $X\perp^d Y|Z$ and in the graph $X\longrsquigarrow Z \longlsquigarrow Y$, we have $X \perp^d Y$ but $X\not \perp^d Y|Z$. 

For completeness of the discussion here, we recall the d-separation property outlined in Equation~\eqref{eq: dsep_prop} in the following definition. Note that this property only considers d-separation relations between observed nodes which correspond to classical random variables.
\begin{definition}[d-separation property for an observed distribution relative to a causal structure] 
\label{definition: compatdist}
Let $\mathcal{G}$ be a causal structure associated with a set $N_{\mathrm{obs}}$ of observed nodes. A distribution $P$ on $N_{obs}$ is said to satisfy the \emph{d-separation property} with respect to $\mathcal{G}$ iff for all disjoint subsets $X$, $Y$ and $Z$ of $\{X_1,...,X_n\}$,
\begin{equation*}
    (X\perp^d Y|Z)_{\cG} \quad\Rightarrow\quad (X\indep Y|Z)_P, \quad \text{where the latter is shorthand for}\quad P(XY|Z)=P(X|Z)P(Y|Z).
\end{equation*}

\end{definition}

Our framework of~\cite{VilasiniColbeckPRA} then employs a rather minimal and theory-independent definition of a causal model (Definition~\ref{def: causalmodel}) as a directed graph $\mathcal{G}$ together with an observed distribution $P_{\mathcal{G}}$ over the observed nodes of $\mathcal{G}$ which respects the d-separation property. Definition~\ref{def: causalmodel} of a causal model is general enough that most previous approaches to classical and non-classical causal models also satisfy its requirements. Specifically, previous frameworks have defined classical and non-classical (quantum or post-quantum) causal models on directed acyclic graphs (in terms of the causal mechanisms) and such models have been shown to satisfy the d-separation property~\cite{Pearl2009, Spirtes2001, Henson2014, Barrett2020A}. In particular, for classical causal models defined on acyclic graphs, one can associate a conditional probability $P(X|\mathrm{par}(X))$ with every node $X$ in the causal graph where par$(X)$ is the set of all parents of $X$ in the graph (in particular, this can be readily computed given a distribution $P(X)$ over every parentless node $X$ and a function $f_X: \mathrm{par}(X)\mapsto X$ for every non-parentless node in the graph). The overall distribution over all nodes (observed and unobserved), is given as follows where we take $\{X_1,...,X_n\}$ to be the set of all nodes in the graph $\mathcal{G}$

\begin{equation}
\label{eq: markov}
    P_{\mathcal{G}}(X_1...X_n)=\prod_{i=1}^n P_{\mathcal{G}}(X_i|\mathrm{par}(X_i))
\end{equation}

This is often known as the causal Markov, global directed Markov or simply the Markov property in the literature. The observed distribution is then obtained by marginalising over the unobserved nodes and any observed distribution obtained in this manner from an acyclic graph will satisfy the d-separation property. Conversely, for any causal model (according to our Definition~\ref{def: causalmodel}), imposing that the graph is acyclic and that there are no latent nodes, we can recover the Markov property from the d-separation property~\cite{Pearl2009}. However, the Markov property generally fails when considering cyclic graphs, even in the case of classical causal models. Nevertheless, there are several examples of classical~\cite{Forre2017} and non-classical models~\cite{VilasiniColbeckPRA} defined on cyclic graphs that satisfy the more general, d-separation property. 

\subsection{Higher-order affects relations and causal inference}

We review some useful lemmas about (higher-order) affects relations derived in~\cite{VilasiniColbeckPRA}, which are relevant for causal inference in cyclic, fine-tuned and non-classical causal models. These results apply to all causal models that satisfy our minimal definition, i.e., they are independent of the theory describing the unobserved nodes and follow entirely from the validity of the d-separation property for the observed nodes. We then outline the notion of irreducibility for higher-order affects relations which is also central to our compatibility condition for space-time embeddings of a causal model, which captures relativistic causality principles.

\begin{lemma}
\label{lemma:exogenous}
If $X$ is a subset of observed parentless nodes of a causal graph $\cG$, then for any subset $Y$ of nodes disjoint to $X$ the do-conditional and the regular conditional with respect to $X$ coincide i.e.,
$$P_{\cG_{do(X)}}(Y|X)=P_{\cG}(Y|X).$$
In other words, for any subset $X$ of the observed parentless nodes, correlation between $X$ and a disjoint set of observed nodes $Y$ in $\cG$ guarantees that $X$ affects $Y$. 
\end{lemma}

\begin{lemma}
\label{lemma: affects}
For any two subsets $X$ and $Y$ of the observed nodes in a causal model associated with a causal graph $\cG$,
\begin{enumerate}
    \item $X$ affects $Y$ $\Rightarrow$ there exists a directed path from $X$ to $Y$ in $\cG$,
    \item $(X\perp^d Y)_{\cG_{do(X)}}$ $\Rightarrow$ $X$ \emph{does not affect} $Y$,
    \item $(X\nindep Y)_{\cG_{do(X)}}\Rightarrow$ $X$ affects $Y$.
\end{enumerate}
\end{lemma}

In the main text, we have defined higher-order affects relations which take the form $X$ affects $Y$ given do$(Z)$ and capture signalling from $X$ to $Y$ given an intervention on $Z$. More generally,~\cite{VilasiniColbeckPRA} defined conditional higher-order affects relations of the form $X$ affects $Y$ given $\{\mathrm{do}(Z),W\}$ which captures signalling from $X$ to $Y$ given an intervention on $Z$ and a post-selection on $W$ (without intervention). 

\begin{definition}[Conditional higher-order affects relation]

Consider a causal model associated with a causal graph $\cG$ over a set $N_{ons}$ of observed nodes and an observed distribution $P$. Let $X$ , $Y$, $Z$ and $W$ be four pairwise disjoint subsets of observed nodes in $N_{\mathrm{obs}}$. We say that $X$ affects $Y$ given $\{\mathrm{do}(Z),W\}$ if there exists values $x$ of $X$, $z$ of $Z$ and $w$ of $W$ such that
\begin{equation}
   P_{\cG_{\mathrm{do}(X Z)}}(Y|X=x,Z=z, W=w)\neq P_{\cG_{\mathrm{do}(Z)}}(Y|Z=z, W=w),
\end{equation}
which we simply denote as $P_{\cG_{\mathrm{do}(X Z)}}(Y|XZW)\neq P_{\cG_{\mathrm{do}(Z)}}(Y|ZW)$. 
\end{definition}

\begin{restatable}{lemma}{HOcauseB}
\label{lemma:HOaffects}
For a causal model over a set $S$ of RVs where $X$ is an RV, $Y$, $Z$ and $W$ are any pairwise disjoint subsets of $S$ that do not contain $X$,
  \begin{enumerate}
      \item $X$ affects $Y$ given do$(Z)$
 $\Rightarrow$ there exists an element $e_Y\in Y$ such that $X$ is a cause of $e_Y$.
 \item $X$ affects $Y$ given $\{\mathrm{do}(Z),W\}$ 
 $\Rightarrow$  there exists an element $e_{YW}\in Y W$ such that $X$ is a cause of $e_{YW}$.
 \end{enumerate}
\end{restatable}

\begin{definition}[Reducible and irreducible affects relations]
\label{definition: ReduceAffects}
If for all proper subsets $s_X$ of $X$, $s_X$ affects $Y$ given $\{\mathrm{do}(Z \tilde{s}_X),W\}$, the affects relation $X$ affects $Y$ given $\{\mathrm{do}(Z),W\}$ is said to be \emph{irreducible}, where $\tilde{s}_X:=X\backslash s_X$. Otherwise, it is said to be \emph{reducible}.
\end{definition}
\begin{restatable}{lemma}{Reduce}
\label{lemma:reduce}
For every reducible conditional higher-order affects relation $X$ affects $Y$ given $\{\mathrm{do}(Z),W\}$, there exists a proper subset $\tilde{s}_X$ of $X$ such that $\tilde{s}_X$ affects $Y$ given $\{\mathrm{do}(Z),W\}$ i.e., $X$ affects $Y$ given $\{\mathrm{do}(Z),W\}$ can be reduced to the affects relation $\tilde{s}_X$ affects $Y$ given $\{\mathrm{do}(Z),W\}$. 
\end{restatable}

\subsection{Compatibility with space-time}

For completeness, we state the original definition of compatibility that was formulated for conditional higher-order affects relations, which capture the most general way of signalling in our framework.

\begin{definition}[Compatibility of a set of affects relations with an embedding in a partial order ($\mathbf{compat}$)]

Let $\mathcal{S}$ be a set of ORVs formed by embedding a set of RVs $S$ in a space-time $\mathcal{T}$ with embedding $\mathscr{E}$. Then a set of affects relations $\mathscr{A}$ is said to be \emph{compatible} with the embedding $\mathscr{E}$ if whenever $X$ affects $Y$ given $\{\mathrm{do}(Z),W\}$ belongs to $\mathscr{A}$, and is irreducible with respect to the affects relations in $\mathscr{A}$, then $\overline{\mathcal{F}}(\cY)\bigcap \overline{\mathcal{F}}(\cW)\bigcap \overline{\mathcal{F}}(\cZ)\subseteq \overline{\mathcal{F}}(\cX) $ with respect to $\mathscr{E}$. 
\end{definition}

The condition captures the idea that in order for Alice to signal to Bob using $X$ affects $Y$ given $\{\mathrm{do}(Z),W\}$, Alice must have access to $X$ and Bob to $Y$, $Z$ and $W$ (since this relation that captures signalling when given additional information about do$(Z)$ and $W$, does not in general imply the possibility of signalling without this information, i.e., that $X$ affects $Y$). Then in order for this affects relation to not lead to signalling outside the space-time future when the RVs are embedded in a space-time, we would require that the joint future of $Y$, $Z$ and $W$ (which is where they can be jointly accessed) in contained in the future of $X$. We note that in the original framework of~\cite{VilasiniColbeckPRA}, we do not a-priory assume that every variable embedded in a space-time can only be accessed in its future lightcone, but instead introduce the concept of an accessible space-time region of an RV, which leads to an additional compatibility condition which sets the accessible regions to be the inclusive future. In this paper, we have set the accessible region to be equal to the inclusive future by default (for clarity) such that violations to compatibility necessarily indicate signalling outside the future light cone, and compatibility therefore captures the principle of no superluminal signalling.

\section{Proofs of all results}
\label{appendix: proofs}

\subsection{Proofs of Lemmas~\ref{lemma: jamming_causation}, \ref{lemma: jamming_finetune} and Theorem~\ref{theorem: jamming}}
\label{appendix: proof_jamming}

\JamCause*
\begin{proof}
The jamming condition $P(XZ|ABC)\neq P(XZ|AC)$ is equivalent to the conditional affects relation $B$ affects $XZ$ given do($AC$), which implies by Lemma~\ref{lemma:HOaffects} and the exogeneity of $A$, $B$ and $C$ that there exists a directed path from $B$ to $X$ or from $B$ to $Z$ i.e., that $B$ is either a cause of $X$ or of $Z$. 
\end{proof}

\JamFinetune*
\begin{proof}
$P(XZ|ABC)\neq P(XZ|AC)$ is equivalent to $(B\not\indep XZ|AC)$, by the d-separation property, we must have $(B\not\perp^d XZ|AC)$ in any underlying causal structure leading to these correlations. This is equivalent to saying that at least one of $(B\not\perp^d X|AC)$ and $(B\not\perp^d Z|AC)$ hold. However $P(X|ABC)=P(X|A)$ and $P(Z|ABC)=P(Z|C)$ imply that $(B\indep X|AC)$ and $(B\indep Z|AC)$ both hold. This means that we must have at least one d-connection relation in the underlying causal structure without the corresponding conditional independence in the observed distribution which implies that the causal model must be fine-tuned.
\end{proof}

\JammingLambda*
\begin{proof}
The most general causal structure over $A$, $B$, $C$, $X$, $Y$, $Z$ and $\Lambda$ where $A$, $B$, $C$ and $\Lambda$ are parentless while $X$, $Y$ and $Z$ are childless is one where $par(X)=par(Z)=par(Y)=\{A,B,C,\Lambda\}$. If $\Lambda$ is observed, then all nodes are classical. Notice that any causal structure over these variables with the above mentioned properties is necessarily acyclic. We can therefore apply the causal Markov condition\footnote{Note that this is not assuming any more than our Definition~\ref{def: causalmodel} together with the given conditions of this Theorem. This is because in any causal model where all nodes are observed, the graph is acyclic and the observed distribution satisfies the d-separation property, the Markov condition follows as a consequence~\cite{Pearl2009}.} to obtain
\begin{equation}
    P(XYZ\Lambda|ABC)=P(\Lambda)P(X|ABC\Lambda )P(Y|ABC\Lambda)P(Z|ABC\Lambda) 
\end{equation}

Now, notice that if we had the affects relation $BC$ affects $X$ given do$(A,\Lambda)$, which in this case is equivalent to $P(X|ABC\Lambda) \neq P(X|A\Lambda)$, then we would have superluminal signalling in the given space-time confirguration since an agent with access to $B$ could signal to an agent with joint access to $X$, $A$ and $\Lambda$ (that can be accessed outside the future light cone of $B$). Therefore, we must have $P(X|ABC\Lambda) =P(X|A\Lambda)$ in order to avoid superluminal signalling in the given space-time embedding. Similarly, we must also have $P(Z|ABC\Lambda Y) =P(Z|C\Lambda)$. Plugging these back in the above equation and marginalising over $\Lambda$ and $Y$, we have.

\begin{equation}
    P(XZ|ABC)=\sum_{\Lambda} P(\Lambda)P(X|A\Lambda)P(Z|C\Lambda)=P(XZ|AC). 
\end{equation}

This contradicts the fact that we have jamming correlations between $A$, $B$, $C$, $X$, $Y$ and $Z$ which requires $P(XZ|ABC)\neq P(XZ|AC)$. Therefore if we do indeed have jamming correlations in a Bell scenario taking place in the jamming space-time configuration, then the involved parties can superluminally signal to each other whenever $\Lambda$ is observed.
\end{proof}

\subsection{Proofs of Theorems~\ref{theorem: ns2_affects} and~\ref{theorem: ns3p_affects}}
\label{appendix: proof_ns_affects}

\NSAffectsBipart*
\begin{proof}
Using the fact that $A$ and $B$ are parentless nodes in a Bell scenario, we have $\mathcal{G}_{\mathrm{do}(S)}\equiv\mathcal{G}$ for all subsets $S$ of $\{A,B\}$. Then denoting $P_{\mathcal{G}}$ as simply $P$ for short, the first non-affects relation $B$ does not affect $X$ given do$(A)$ (which when written explicitly is: $P_{\mathcal{G}_{\mathrm{do}(AB)}}(X|AB)=P_{\mathcal{G}_{\mathrm{do}(A)}}(X|A)$) is equivalent to $P(X|AB)=P(X|A)$ which is one of the two conditions in \hyperref[cond: NS2]{(NS2)}. Similarly, in any scenario where $A$ and $B$ are parentless,  $B$ does not affect $X$ given do$(A)$ is equivalent to the second condition of \hyperref[cond: NS2]{(NS2)}, namely $P(Y|AB)=P(Y|B)$.
\end{proof}

\NSAffectsTripart*
\begin{proof}
Using the fact that $A$, $B$ and $C$ are parentless nodes in a tripartite Bell scenario, we have $\mathcal{G}_{\mathrm{do}(S)}\equiv\mathcal{G}$ for all subsets $S$ of $\{A,B, C\}$. Then denoting $P_{\mathcal{G}}$ as simply $P$ for short, the first non-affects relation $C$ does not affect $XY$ given do$(AB)$ is equivalent to $P(XY|ABC)=P(XY|AB)$ which is the first condition of \hyperref[cond: NS3p]{(NS3$'$)}. Similarly, each of the remaining non-affects relations translate to the corresponding condition of \hyperref[cond: NS3p]{(NS3$'$)} as long as $A$, $B$ and $C$ are parentless.
\end{proof}

\subsection{Proofs of Theorems~\ref{theorem: superlum_bipart} and~\ref{theorem: superlum_tripart}}
\label{appendix: proof_superlum}

\SuperlumBipart*
\begin{proof}
  Given the space-time embedding, compatibility implies that $A$ does not affect $Y$ given $\mathrm{do}(B)$ and $B$ does not affect $X$ given $\mathrm{do}(A)$. Since $A$ and $B$ are parentless, these are equivalent to the no-signalling conditions \hyperref[cond: NS2]{(NS2)}, cf.\ Theorem~\ref{theorem: ns2_affects}. Hence condition~\ref{cond:bi1} is necessary.  Condition~\ref{cond:bi2} is necessary because there are no variables in the future of either $X$ or $Y$ in the space-time embedding (nor the joint futures of any variables), so the compatibility conditions precludes affects relations of the form $X$ affects $S_1$ given $\mathrm{do}(S_2)$ and similarly for $Y$.
  
To prove sufficiency, we will show that in any bipartite Bell scenario satisfying conditions~\ref{cond:bi1} and~\ref{cond:bi2}, only irreducible affects relations of the following form are allowed:
\begin{enumerate}[I.]
    \item $A$ affects $S_1$ given do$(S_2)$ where $X\in S_1$.
    \item $B$ affects $S_1$ given do$(S_2)$ where $Y\in S_1$.
    \item $AB$ affects $XY$.
\end{enumerate}

All affects relations of the above forms are compatible with the standard space-time embedding of a bipartite Bell scenario, since we have $\cA\prec\cX$ and $\cB\prec\cY$.  We proceed to show that these are the only irreducible affects relations possible.\medskip

{\bf Step 1.} We first note that if $S_3$ affects $S_1$ given $\mathrm{do}(S_2)$ is irreducible then $S_3$ cannot contain $X$. This follows because if $S_3$ contains $X$, then the condition of being irreducible implies $X$ affects $S_1$ given $\mathrm{do}(S_2\cup S_3\setminus\{X\})$, in violation of condition~\ref{cond:bi2}.  By symmetry it follows that there is also no irreducible affects relation of the form $S_3$ affects $S_1$ given $\mathrm{do}(S_2)$ where $S_3$ contains $Y$. Thus, the only possibilities are where $S_3$ is $A$, $B$ or $AB$.\medskip

{\bf Step 2.} We start by considering affects relations emanating from $A$ which are not of the form of I. i.e., $A$ affects $S_1$ given do$(S_2)$ where $X\not\in S_1$.  Since $S_1$ must necessarily be non-empty and there are only two other variables $B$ and $Y$ in the scenario, $S_1$ must be a non-empty subset of $BY$. Now, we write $S_2:=S^i S^o$ where $S^i$ and $S^o$ are the sets of all setting and outcome variables contained in $S_2$ respectively. Then consider two cases depending on whether or not $S^o=\emptyset$.

\begin{itemize}
   \item {\bf Step 2a.} (Ruling out affects relations of the form  $A$ affects $S_1$ given do$(S^iS^o)$ where $X\not\in S_1$ and $S^o=\emptyset$) The only such possibilities are $A$ affects $Y$ given do$(B)$, $A$ affects $Y$, $A$ affects $B$ and $A$ affects $BY$. We first note that since $A$ and $B$ are parentless, d-separation implies that $P(AB)=P(A)P(B)$.\footnote{Here we use $P$ to denote the distribution $P_{\mathcal{G}}$ associated with any pre-intervention causal structure $\mathcal{G}$ allowed by a bipartite Bell scenario.} This rules out $A$ affects $B$ (Lemma~\ref{lemma: affects}). $A$ affects $Y$ given do$(B)$ is already ruled out by applying \hyperref[cond: NS2]{(NS2)} (cf.\ Theorem~\ref{theorem: ns2_affects}). To show $A$ does not affect $BY$ (which implies $A$ does not affect $Y$), we start with \hyperref[cond: NS2]{(NS2)} which states that $P(Y|AB)=P(Y|B)$ and multiply both sides by $P(AB)$. The left hand side gives $P(Y|AB)P(AB)=P(Y|AB)P(B|A)P(A)=P(YB|A)P(A)$ and the right hand side gives $P(Y|B)P(AB)=P(Y|B)P(A)P(B)=P(YB)P(A)$. Putting this together gives $P(YB|A)=P(YB)$ (which is equivalent to $A$ does not affect $BY$ since $A$ is parentless). [Summing over $B$ establishes that $A$ does not affect $Y$.] 
   \item {\bf Step 2b.} (Ruling out affects relations of the form  $A$ affects $S_1$ given do$(S^i S^o)$ where $X\not\in S_1$ and $S^o\neq \emptyset$) Observe that for any $S_1$ and $S_2:=S^iS^o$ associated with such an affects relation, the following non-affects relations between elements of $S_1$ and $S_2$ are implied by conditions~\ref{cond:bi1} and~\ref{cond:bi2} in any Bell scenario: $A$ does not affect $S_1$ given do$(S^i)$ (follows from Step 2a), $S^o$ does not affect $S_1$ given do$(S^i)$ and $S^o$ does not affect $S_1$ given do$(S^iA)$ (these two follow from condition~\ref{cond:bi2}). Writing these out using the fact that interventions on exogeneous variables do not change the causal structure, we have
    \begin{align}
        \begin{split}
            P(S_1|AS^i)&=P(S_1|S^i)\\
            P_{\mathcal{G}_{\mathrm{do}(S^o)}}(S_1|S^oS^i)&=P(S_1|S^i)\\
            P_{\mathcal{G}_{\mathrm{do}(S^o)}}(S_1|S^oAS^i)&=P(S_1|AS^i)
        \end{split}
    \end{align}
    which imply $P_{\mathcal{G}_{\mathrm{do}(S^o)}}(S_1|AS^oS^i)=P_{\mathcal{G}_{\mathrm{do}(S^o)}}(S_1|S^oS^i)$, i.e., $A$ does not affect $S_1$ given do$(S_2)$ since $S_2=S^o S^i$.  Together with the above item, this completes the proof that the only affects relations of the form $A$ affects $S_1$ given $\mathrm{do}(S_2)$ are those with $X\in S_1$. By symmetry, the only affects relations of the form $B$ affects $S_1$ given $\mathrm{do}(S_2)$ must have $Y\in S_1$.
    \end{itemize}

{\bf Step 3.} Case III. allows $AB$ affects $XY$ which can be irreducible. Other possible affects relations emanating from $AB$, such as $AB$ affects $Y$ will be reducible (since $A$ does not affect $Y$ given do$(B)$, by the arguments in Step 2a.). \medskip

This completes the proof.
\end{proof}

\SuperlumTripart*
\begin{proof}
Throughout this proof, we use that $A$, $B$ and $C$ are parentless by definition, in any tripartite Bell scenario (Definition~\ref{def: Bell_scenario}). This implies by the d-separation property (Definition~\ref{definition: compatdist}) that $P(ABC)=P(A)P(B)P(C)$. Here, we use $P$ to denote the distribution $P_{\mathcal{G}}$ associated with any given causal structure $\mathcal{G}$ associated with a tripartite Bell scenario, and recall that interventions on parentless nodes do not alter the causal structure or the associated distribution (Definition~\ref{def:post_intervention}).\medskip

{\bf Necessity} The necessity of condition~\ref{cond:tri1} (namely, that \hyperref[cond: NS3p]{(NS3$'$)} holds) for compatibility of the tripartite Bell scenario with a jamming space-time configuration follows from applying the equivalence between \hyperref[cond: NS3p]{(NS3$'$)} and the non-affects relations of Equation~\eqref{eq: ns3p_affects} in Bell scenario: if any of these affects relations were present in a tri-partite Bell scenario, they would be incompatible with a jamming space-time embedding (this holds irrespective of whether or not $\overline{\mathcal{F}}(\cX)\cap \overline{\mathcal{F}}(\cZ)\not \subseteq \overline{\mathcal{F}}(\cY)$ is additionally satisfied in the space-time embedding).

Now, we turn our attention to the second condition. Consider the case where we are given that  $\overline{\mathcal{F}}(\cX)\cap \overline{\mathcal{F}}(\cZ)\not \subseteq \overline{\mathcal{F}}(\cY)$. To show the necessity of condition~\ref{cond:tri2} in this case, we first observe that since $\cA\prec X$ and $\cC\prec \cZ$, $\overline{\mathcal{F}}(\cX)\cap \overline{\mathcal{F}}(\cZ)\subseteq \overline{\mathcal{F}}(\cA)\cap \overline{\mathcal{F}}(\cC)$ which implies that none of the joint futures $\overline{\mathcal{F}}(\cA)\cap \overline{\mathcal{F}}(\cC)$, $\overline{\mathcal{F}}(\cA)\cap \overline{\mathcal{F}}(\cZ)$, $\overline{\mathcal{F}}(\cX)\cap \overline{\mathcal{F}}(\cC)$ can be contained in $\overline{\mathcal{F}}(\cY)$. 
Then it follows that in any space-time embedding satisfying the conditions of Definition~\ref{def: spacetime_config_jamming} along with  $\overline{\mathcal{F}}(\cX)\cap \overline{\mathcal{F}}(\cZ)\not \subseteq \overline{\mathcal{F}}(\cY)$, there are no observed variables or joint futures of observed variables that are guaranteed to be contained in the future of $\cX$, of $\cY$ or of $\cZ$, which means that in order to avoid superluminal signalling, in such a space-time embedding, we must not have any affects relations originating from the outcome variables as required by condition~\ref{cond:tri2}.

Next, consider the case where $\overline{\mathcal{F}}(\cX)\cap \overline{\mathcal{F}}(\cZ)\subseteq \overline{\mathcal{F}}(\cY)$ in a space-time embedding satisfying the conditions of Definition~\ref{def: spacetime_config_jamming}.  First, we note that all affects relations emanating from $Y$ which satisfy condition~\ref{cond:tri2p} (i.e., where $Y$ affects $S_1$ given do$(S_2)$, with $\{X,Z\}\subseteq S_1$) are compatible with the given space-time embedding since we have $\overline{\mathcal{F}}(\cX)\cap \overline{\mathcal{F}}(\cZ)\subseteq \overline{\mathcal{F}}(\cY)$ (which implies $\overline{\mathcal{F}}(\cX)\cap \overline{\mathcal{F}}(\cZ) \cap \overline{\mathcal{F}}(\cS_2) \subseteq \overline{\mathcal{F}}(\cY)$ for any $S_2$). Below, we show that compatibility of a Bell scenario with such a space-time embedding rules out all affects relations which are not of the form given in condition~\ref{cond:tri2p}.

\begin{itemize}
    \item {\bf Case N1: Affects relations emanating from $Y$ which do not satisfy condition~\ref{cond:tri2p}.} These are affects relations of the form $Y$ affects $S_1$ given do$(S_2)$ where $\{X,Z\}\not\subseteq S_1$ and can be partitioned into the cases where $\{X,Z\}\cap S_1=\emptyset$ and $\{X,Z\}\cap S_1\neq \emptyset$. In the former case, we would necessarily have $S_1\subseteq \{A,B,C\}$ and these affects relations would never arise in any Bell scenario (where $A$, $B$, $C$ are parentless by definition) as they are ruled out by Lemma~\ref{lemma:HOaffects} (a consequence of the d-separation property). In the latter case, we can further partition into the sub-cases where $\{X,Z\}\subseteq S_1\cup S_2$ (Case N1.1) or $\{X,Z\}\not\subseteq S_1\cup S_2$ (Case N1.2). 
    We first reduce Case N1.1 to Case N1.2 using compatibility and then rule out the affects relations of the latter also using compatibility of a Bell scenario with the given embedding.
   The only affects relations of the form of Case N1.1 (recalling that we have also assumed from before, that exactly one of $X$ or $Z$ is contained in $S_1$) are: $Y$ affects $X S^i$ given do$(Z \bar{S}^i)$, $Y$ affects $Z S^i$ given do$(X \bar{S}^i)$ for $S^i$ and $\bar{S}^i$ being some disjoint subsets of $\{A,B,C\}$. We only consider the first affects relations as the second one is then covered by symmetry between $X$ and $Z$. Writing out the affects relation, we have

  \begin{align}
  \label{eq: tripart_necc_proof}
        P_{\mathcal{G}_{\mathrm{do}(YZ)}}(XS^i|YZ\bar{S}^i)\neq  P_{\mathcal{G}_{\mathrm{do}(Z)}}(XS^i|Z\bar{S}^i)
    \end{align}

 We can simplify the expression using constraints from compatibility.  First notice that compatibility with the given embedding implies $Z$ does not affect $X S^i$ given do$(\bar{S}^i)$ (for any $S^i$ and $\bar{S}^i$), since the future of $Z$ does not contain the joint future of any set of variables. Using this non-affects relation, the right hand side becomes equal to $P(XS^i|\bar{S}^i)$. Plugging this back in, we have 
 \begin{equation}
             P_{\mathcal{G}_{\mathrm{do}(YZ)}}(XS^i|YZ\bar{S}^i)\neq  P(XS^i|\bar{S}^i).
 \end{equation}
 
 This is equivalent to $YZ$ affects $XS^i$ given do$(\bar{S}^i)$, which can either be irreducible or reducible. If it is irreducible, we must necessarily have $Z$ affects $XS^i$ given do$(Y\bar{S}^i)$ which is ruled out by compatibility with the given embedding. If it is reducible, it must either be reducible to $Z$ affects $XS^i$ given do$(\bar{S}^i)$ or to $Y$ affects $XS^i$ given do$(\bar{S}^i)$. The former is again ruled out by compatibility and we are left with the latter, which is equivalent to $P_{\mathcal{G}_{\mathrm{do}(Y)}}(XS^i|Y\bar{S}^i)\neq P(XS^i|\bar{S}^i)$. Notice that $Y$ affects $XS^i$ given do$(\bar{S}^i)$ are precisely the set of affects relations that comprise Case N1.2. Therefore the rest of the proof now covers both Cases N1.1 and N1.2.\footnote{Case N1.2 allows all affects relations $Y$ affects $S_1$ given do$(S_2)$ where $S_1$ contains exactly one of $X$ or $Z$ and $\{X,Z\}\not\subseteq S_1\cup S_2$. The most general affects relations of this form are $Y$ affects $XS^i$ given do$(\bar{S}^i)$ and $Y$ affects $ZS^i$ given do$(\bar{S}^i)$, where $S^i$ and $\bar{S}^i$ are arbitrary disjoint subsets of $\{A,B,C\}$. The argument for the two sets of affects relations are symmetric under exchange of $X$ and $Z$ to we focus on $Y$ affects $XS^i$ given do$(\bar{S}^i)$.} We had established that $P_{\mathcal{G}_{\mathrm{do}(Y)}}(XS^i|Y\bar{S}^i)\neq P(XS^i|\bar{S}^i)$, which is equivalent to
\begin{equation}\label{eq:dorel}
            P_{\mathcal{G}_{\mathrm{do}(Y)}}(X|S^iY\bar{S}^i)P_{\mathcal{G}_{\mathrm{do}(Y)}}(S^i|Y\bar{S}^i)\neq  P(X|S^i\bar{S}^i) P(S^i|\bar{S}^i).
 \end{equation}
Next we use Lemma~\ref{lemma:HOaffects}, which implies that (recall that $S^i$ and $\bar{S^i}$ consist of parentless nodes) $Y$ does not affect $S^i$ given do$(\bar{S}^i)$, which reduces~\eqref{eq:dorel} to $P_{\mathcal{G}_{\mathrm{do}(Y)}}(X|S^iY\bar{S}^i)\neq P(X|S^i\bar{S}^i)$. Then, denoting $S^i\bar{S}^i\setminus A:= \tilde{S}^i$ (and assuming for the moment that $A\in S^i\bar{S}^i$), we have: $ P_{\mathcal{G}_{\mathrm{do}(Y)}}(X|Y\tilde{S}^i A)\neq  P(X|\tilde{S}^i A)$. We again invoke compatibility which implies that $\tilde{S}^i$ does not affect $X$ given do$(YA)$ (i.e., $ P_{\mathcal{G}_{\mathrm{do}(Y)}}(X|Y\tilde{S}^i A)= P_{\mathcal{G}_{\mathrm{do}(Y)}}(X|YA)$) and $\tilde{S}^i$ does not affect $X$ given do$(A)$ (i.e., $P(X|\tilde{S}^i A)=P(X| A)$) for any subset of settings $\tilde{S}^i$ which does not include $A$. This further simplifies our expression to 
\begin{equation}
            P_{\mathcal{G}_{\mathrm{do}(Y)}}(X|Y A)\neq  P(X| A).
 \end{equation}
 However, this is equivalent to $Y$ affects $X$ given do$(A)$ which is ruled out, as it is incompatible with the given space-time embedding. In the last step, we assumed that $A\in S^i\bar{S}^i$. If $A\not\in S^i\bar{S}^i$, we would instead obtain $Y$ affects $X$ in this step, which is also incompatible with the embedding. This covers all relations from $Y$.

\item {\bf Case N2: Affects relations emanating from $X$, $Z$, $XY$, $YZ$, $XZ$ and $XYZ$} Affects relations from $XYZ$ would have to affect an input variable which is not possible due to the parentless nature of these variables in the Bell scenario. Affects relations from $X$, $Z$ and $XZ$ are ruled out by compatibility of the Bell scenario with the given embedding (there are no subsets whose joint future is contained in that of $X$, $Z$ or $XZ$). Affects relations from $XY$ (if they exist), must be of the form $XY$ affects $S_1$ given do$(S_2)$ and can be either irreducible or reducible. If irreducible, this means that $X$ affects $S_1$ given do$(S_2Y)$ which we have already ruled out by compatibility. If reducible, it must be reduced to either one of the irreducible relations $X$ affects $S_1$ given do$(S_2)$ or $Y$ affects $S_1$ given do$(S_2)$ (where $X\not\in S_1 \cup S_2$). The former, we have already ruled out and in the latter case, and the latter cannot possibly satisfy condition~\ref{cond:tri2p} and are therefore also ruled out by earlier arguments.
\end{itemize}
This completes the proof for the necessary part. \medskip

{\bf Sufficiency} We start with the case where $\overline{\mathcal{F}}(\cX)\cap \overline{\mathcal{F}}(\cZ)\not \subseteq \overline{\mathcal{F}}(\cY)$. To prove sufficiency in this case, we will show that in any tripartite Bell scenario satisfying conditions~\ref{cond:tri1} and ~\ref{cond:tri2} of the current theorem, the only irreducible affects relations that are allowed are the following.
\begin{itemize}
    \item {\bf Case S1:} \begin{itemize}
        \item[S1.1]  $A$ affects $S_1$ given do$(S_2)$ where $X\in S_1$ 
  
      \item[S1.2] $C$ affects $S_1$ given do$(S_2)$ where $Z\in S_1$ 
  
        \item[S1.3] $AC$ affects $S_1$ given do$(S_2)$ where $\{X,Z\}\subseteq S_1$
    \end{itemize}

     \item {\bf Case S2:} \begin{itemize}
     \item[S2.1]  $B$ affects $S_1$ given do$(S_2)$ where $Y\in S_1$ or $\{X,Z\}\subseteq S_1$ or both
  \item[S2.2] $AB$ affects $S_1$ given do$(S_2)$ where $\{X,Y\}\subseteq S_1$ or $\{X,Z\}\subseteq S_1$ or both
    \item[S2.3] $BC$ affects $S_1$ given do$(S_2)$ where $\{Y,Z\}\subseteq S_1$ or $\{X,Z\}\subseteq S_1$ or both
     \end{itemize}
     
        \item {\bf Case S3:} $ABC$ affects $S_1$ given do$(S_2)$ where $\{X,Z\}\subseteq S_1$

\end{itemize}

Notice that all affects relations of the above form are compatible with the jamming space-time embedding of a tripartite Bell scenario, since we have $\cA\prec \cX$, $\cB\prec Y$, $\cC\prec \cZ$ and $\overline{\mathcal{F}}(\cX)\cap \overline{\mathcal{F}}(\cZ)\subseteq \overline{\mathcal{F}}(\cB)$. As reducible affects relations (even if they exist) do not imply any additional compatibility constraints, any scenario where these are the only irreducible affects relations would be compatible with the given space-time embedding and consequently, not lead to superluminal signalling. Therefore showing that conditions~\ref{cond:tri1} and ~\ref{cond:tri2} imply that the affects relations of Cases S1 to S3. are the only allowed irreducible affects relations, would complete this proof for the case where $\overline{\mathcal{F}}(\cX)\cap \overline{\mathcal{F}}(\cZ)\not \subseteq \overline{\mathcal{F}}(\cY)$.

{\bf Step 1.} Using the same argument as in Step~1 in the proof of Theorem~\ref{theorem: superlum_bipart}, we conclude that if $S_3$ affects $S_1$ given do$(S_2)$
is irreducible, then $S_3$ cannot contain $X$, $Y$ or $Z$ (any affects relation of this form where $S_3$ contains one of these outcome variables would be rendered reducible to an affects relation without these variables in the first argument of the affects relation, due to condition~\ref{cond:tri2} of the theorem). Thus the only possibilities for irreducible affects relations are those where $S_3$ is a non-empty subset of $\{A,B,C\}$, and there are 7 such possibilities.   

{\bf Step 2.} The proof for the cases where $S_3=A$ (Case~S1.1) and $S_3=C$ (Case~S1.2) is identical to that of Step~2 in Theorem~\ref{theorem: superlum_bipart}. [Cases where $A$ affects $S_1$ given do$(S_2)$ where $X\notin S_1$ and $S_2$ only contains inputs are excluded using NS3$'$, and the extension to the case of general $S_2$ uses Step~2b of the previous proof.]

For the case where $S_3=AC$ (Case~S1.3), we have irreducible affects relations of the form $AC$ affects $S_1$ given do$(S_2)$. By the definition of irreducibility this implies that $A$ affects $S_1$ given do$(CS_2)$ and $C$ affects $S_1$ given do$(AS_2)$, which are affects relations covered by Cases~S1.1 and~S1.2 respectively (which we have already addressed). By the condition of Case~S1.1, we must have $X\in S_1$ and by the condition of Case~S1.1, we must have $X\in S_1$ for such an affects relation to hold. This implies that the only irreducible affects relations emanating from $S_3=AC$ are those of the form of Case~S1.3.\footnote{While affects relations such as $AC$ affects $X$ (with possible conditioning on other interventions) are possible, they would not be irreducible since $C$ does not affect $X$ (even when conditioning on interventions on other variables).} This covers all of Case~S1, which concerns affects relations emanating from sets of variables not containing $B$.

{\bf Step 3.} We now consider Cases~S2 and~S3 which involve the variable $B$, which has a special role in this scenario (due to the possibility of jamming). The proof in this case is also similar to that of Step 2, but we repeat it for completeness. Consider $S_3=B$, and consider affects relations emanating from $S_3$ which are not of the form of S2.1 i.e., where $B$ affects $S_1$ given do$(S_2)$ with $Y\not\in S_1$ and $\{X,Z\}\not\subseteq S_1$. As in the proof of Theorem~\ref{theorem: superlum_bipart} we partition $S_2$ as $S_2=S^i S^o$ where $S^i$ and $S^o$ consists only of setting and outcome variables respectively, and we treat the two cases $S^o=\emptyset$ and $S^o\neq\emptyset$ separately. 

\begin{itemize}
    \item {\bf Step 3a.} (Ruling out affects relations of the type $B$ affects $S_1$ given do$(S^i S^o)$ with $Y\not\in S_1$, $\{X,Z\}\not\subseteq S_1$ and $S^o=\emptyset$) The only such possibilities are: affects relations where $S_1$ only contains settings (these are ruled out using the free choice of settings, and the d-separation property), $B$ affects $X$, $B$ affects $Z$, and the last two affects relations with the conditioning on do$(A)$, do$(C)$ or do$(AC)$. Consider $B$ affects $X$ given do$(AC)$. Using Theorem~\ref{theorem: ns3p_affects} we know that \hyperref[cond: NS3p]{(NS3$'$)} (condition~\ref{cond:tri1} of the current theorem) implies that $BC$ does not affect $X$ given do$(A)$ and $C$ does not affect $XY$ given do$(AB)$. The latter immediately implies that $C$ does not affect $X$ given do$(AB)$ and together with $P(ABC)=P(A)P(B)P(C)$, we can easily verify that this gives us $C$ does not affect $X$ given do$(A)$, which is equivalent to $P(X|AC)=P(X|A)$ in a Bell scenario. Using this in $BC$ does not affect $X$ given do$(A)$, which is equivalent to $P(X|ABC)=P(X|A)$, we have $P(X|ABC)=P(X|AC)$, which is equivalent to $B$ does not affect $X$ given do$(AC)$. It is straightforward to verify that this also implies $B$ does not affect $X$ given do$(A)$, $B$ does not affect $X$ given do$(C)$ and $B$ does not affect $X$. By symmetry, we can rule out affects relations $B$ affects $Z$ (given do$(A)$, do$(C)$, do$(AC)$) through similar arguments based on the equivalent version of \hyperref[cond: NS3p]{(NS3$'$)} given in Theorem~\ref{theorem: ns3p_affects} and the free choice of the settings.

    \item {\bf Step 3b.} (Ruling out affects relations of the type $B$ affects $S_1$ given do$(S^i S^o)$ with $Y\not\in S_1$, $\{X,Z\}\not\subseteq S_1$ and $S^o\neq\emptyset$) Observe that for any $S_1$ and $S^i\cup S^o$ associated with such an affects relation, the following non-affects relations are implied by conditions~\ref{cond:tri1} and~\ref{cond:tri2} in any Bell scenario: $B$ does not affect $S_1$ given do$(S^i)$ (follows from Step 3a), $S^o$ does not affect $S_1$ given do$(S^i)$ and $S^o$ does not affect $S_1$ given do$(S^i B)$ (the last two non-affects relations follow from condition~\ref{cond:tri2} that forbids affects relations emanating from sets of outcome variables). Writing down these three non-affects relations explicitly, while using the fact that interventions on any set $S^i$ (which consists only of settings) does not change the causal structure, we have
    \begin{align}\label{eq:B8}
        \begin{split}
            P(S_1|BS^i)&=P(S_1|S^i)\\
            P_{\mathcal{G}_{\mathrm{do}(S^o)}}(S_1|S^oS^i)&=P(S_1|S^i)\\
            P_{\mathcal{G}_{\mathrm{do}(S^o)}}(S_1|S^oAS^i)&=P(S_1|BS^i)
        \end{split}
    \end{align}
    The above implies that $P_{\mathcal{G}_{\mathrm{do}(S^o)}}(S_1|S^oBS^i)=P_{\mathcal{G}_{\mathrm{do}(S^o)}}(S_1|S^oS^i)$ which, in any Bell scenario, is equivalent to $B$ does not affect $S_1$ given do$(S_2)$ since $S_2=S^o\cup S^i$. Together with Step 3a, this completes the proof for Case S2.1.
\end{itemize}

The proofs for Cases~S2.2 and~S2.3 follow from Cases~S1 and~S2.1 and the definition of irreducibility. For Case~S2.2, we consider irreducible affects relations of the form $AB$ affects $S_1$ given do$(S_2)$. By the definition of irreducibility, this implies that $A$ affects $S_1$ given do$(BS_2)$ and $B$ affects $S_1$ given do$(AS_2)$, which are affects relations of the form covered by previous Cases S1.1 and S2.1 respectively. The conditions of Case S1.1 tell us that $X\in S_1$ must hold and the conditions of Case S2.2 tell us that either $Y\in S_1$ or $\{X,Z\}\subseteq S_1$ must hold for these affects relation to hold. Combining this, we know that we must either have $\{X,Y\}\subseteq S_1$ or $\{X,Z\}\subseteq S_1$ for an irreducible affects relation $AB$ affects $S_1$ given do$(S_2)$ to hold, which is precisely the condition of Case~S2.2. Case~S2.3 is similarly covered by using the definition of irreducibility and referring to the previous Cases~S1.2 and~S2.1.

{\bf Step 4.} For Case S3, we see that any affects relation $ABC$ affects $S_1$ given do$(S_2)$ where $\{X,Z\}\not\subseteq S_1$ would be reducible. This is because, the only possibilities for $S_1$ are then $X$, $Y$, $Z$, $XY$, $YZ$ and in all these cases, there is at least one setting in $\{A,B,C\}$ that does not affect such as $S_1$ (even when conditioned on interventions on other outcomes). 

Therefore, using conditions 1 and 2, we have shown that the only irreducible affects relations are of the form of those listed in Cases S1 to S3. All of these are compatible with any space-time embedding that respects the jamming configuration (Definition~\ref{def: spacetime_config_jamming}) and  $\overline{\mathcal{F}}(\cX)\cap \overline{\mathcal{F}}(\cZ)\not \subseteq \overline{\mathcal{F}}(\cY)$.

Finally, we consider the case where the space-time embedding satisfies $\overline{\mathcal{F}}(\cX)\cap \overline{\mathcal{F}}(\cZ) \subseteq \overline{\mathcal{F}}(\cY)$ along with the conditions of the jamming configuration, and wish to show that conditions 1 and 2$'$ are sufficient for compatibility. This only affects the arguments of Step 1 and Step 3b above, as these are the only steps which use condition~\ref{cond:tri2}. In Step 1, using condition~\ref{cond:tri2p} instead of condition~\ref{cond:tri2}, we now get an additional class of irreducible affects relations in addition to Cases S1 to S3.

\begin{itemize}
    \item {\bf Case S4:} $S^iY$ affects $S_1$ given do$(S_2)$ where $\{X,Z\}\subseteq S_1$ and $S^i\subseteq \{A,B,C\}$.
\end{itemize}

However all these affects relations are compatible with the given embedding since we have $\overline{\mathcal{F}}(\cX)\cap \overline{\mathcal{F}}(\cZ) \subseteq \overline{\mathcal{F}}(\cY)$, $\overline{\mathcal{F}}(\cX)\cap \overline{\mathcal{F}}(\cZ) \subseteq \overline{\mathcal{F}}(\cB)$ as well as $\cA\prec \cX$ and $\cC\prec \cZ$. In Step 3b, we use condition~\ref{cond:tri2} to impose the non-affects relations  $S^o$ does not affect $S_1$ given do$(S^i)$ and $S^o$ does not affect $S_1$ given do$(S^i B)$. But not all such non-affects relations follow from the weaker condition~\ref{cond:tri2p}, as we can for instance have $Y$ affects $XZ$. However, this does not affect the proof because the conditions of Step 3b tell us that $\{X,Z\}\not\subseteq S_1$, and we only need to use the non-affects relations for such a set $S_1$. These non-affects relations are indeed implied by the weaker condition~\ref{cond:tri2p} as we have argued in Case S4 above. Therefore the proof for Step 3b remains unaffected by replacing condition~\ref{cond:tri2} with condition~\ref{cond:tri2p}. This shows that conditions~\ref{cond:tri1} and~\ref{cond:tri2p} are sufficient for compatibility in any jamming space-time embedding where  $\overline{\mathcal{F}}(\cX)\cap \overline{\mathcal{F}}(\cZ) \subseteq \overline{\mathcal{F}}(\cY)$ is also satisfied. This completes the proof.
\end{proof}

\subsection{Proof of Theorem~\ref{theorem: loops1}}
\label{appendix: proof_loop}

\CausalLoopsBell*
\begin{proof}
Consider the bipartite Bell scenario. Condition~\ref{cond:bi2} of Theorem~\ref{theorem: superlum_bipart} is sufficient for ruling out affects causal loops because it restricts all affects relations to start from either $A$ or $B$ or $AB$, but these variables are parentless so no affects relation can end on $A$, $B$ or $AB$, so there cannot be an affects causal loop. The proof of sufficiency of condition~\ref{cond:tri2} of Theorem~\ref{theorem: superlum_tripart} for the tripartite scenario is analogous. For condition~\ref{cond:tri2p} of Theorem~\ref{theorem: superlum_tripart}, there is are additional possible affects relations of the form $Y$ affects $S_1$ given do$(S_2)$ where $\{X,Z\}\subseteq S_1$. Here, $S_1$ could contain the setting variables, $A$, $B$ or $C$. However, even such affects relations cannot allow us to infer a causal influence ending on $A$, $B$ or $C$ as these are parentless nodes in a Bell scenario. The fact that no causal influence can end on $A$, $B$ or $C$, together with the fact that no affects relations can start from $X$ or $Z$ means that there cannot be an affects causal loop in any set of affects relations satisfying condition~\ref{cond:tri2p}.

To show the sufficiency of condition~\ref{cond:tri2p} of Theorem~\ref{theorem: superlum_tripart} for ruling out affects loops, notice that~\ref{cond:tri2p} also allows affects relations of the form $Y$ affects $S_1$ given do$(S_2)$ where $\{X,Z\}\subseteq S_1$. This allows us to have affects relations from $Y$ to $B$, for instance where $B\in S_1$, and one might expect that together with $B$ affects $Y$ (which is also allowed), this might lead to an affects loop. However, while $B$ affects $Y$ allows us to infer that $B$ is a cause of $Y$ (Lemma~\ref{lemma: affects}), no affects relation emanating from $Y$ could allow us to infer that $Y$ is a cause of $B$, since $B$ is a parentless node and has no causes. The only causal loops (if present) could be between outcome variables, however as there are no affects relations emanating from $X$ or $Z$, having the affects relations $Y$ affects $S_1$ given do$(S_2)$ where $\{X,Z\}\subseteq S_1$ alone are insufficient to infer the existence of causal loops i.e., any causal loops (if they exist) cannot be affects causal loops.

The insufficiency of these conditions for ruling out hidden causal loops is established by Example~\ref{example: hidden_loop} (originally introduced in \cite{VilasiniColbeckPRA}). This example presents a hidden causal loop which satisfies \hyperref[cond: NS2]{(NS2)}, \hyperref[cond: NS3]{(NS3)} and \hyperref[cond: NS3p]{(NS3$'$)}. Moreover, as the example only involves two non-trivial outcomes $X$ and $Y$ without any affects relations emanating from them, it satisfies the second conditions of Theorems~\ref{theorem: superlum_bipart} and~\ref{theorem: superlum_tripart} (both~\ref{cond:tri2} and~\ref{cond:tri2p} in the latter case).
\end{proof}

\subsection{Counterexamples used in the proof of Theorem~\ref{THEOREM: LOOPSBELL2}}
\label{appendix: examples}

\begin{figure}
    \centering

      \subfloat[\label{fig: affects_loop}]{ \includegraphics[]{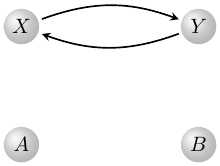}}\qquad\qquad \subfloat[\label{fig: hidden_loop}]{\includegraphics[]{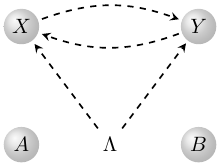}} \qquad\qquad \subfloat[\label{fig: non-necc1}]{\includegraphics[]{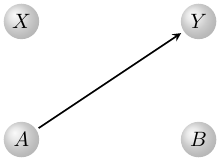}} \qquad\qquad \subfloat[\label{fig: non-necc2}]{\includegraphics[]{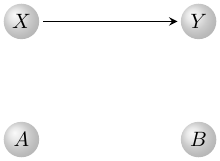}}
      \qquad\qquad  \subfloat[\label{fig: examples_spacetime2}]{\includegraphics[]{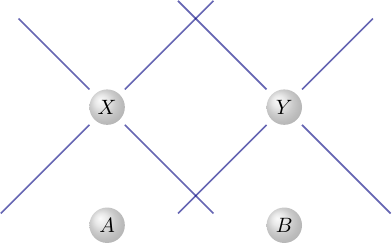}}
    \caption{ Causal structure of (a) Example~\ref{example: affects_loop} (which shows that \hyperref[cond: NS2]{(NS2)}/\hyperref[cond: NS3]{(NS3)} are insufficient for ruling out affects causal loops in Bell scenarios) (b) Example~\ref{example: hidden_loop} (which shows that \hyperref[cond: NS2]{(NS2)}/\hyperref[cond: NS3]{(NS3)} are insufficient for ruling out hidden causal loops in Bell scenarios) (c) Example~\ref{example: non-necc1} (which shows that \hyperref[cond: NS2]{(NS2)}/\hyperref[cond: NS3]{(NS3)} are not necessary for ruling out causal loops in Bell scenarios) (d) Example~\ref{example: non-necc2} (which shows that the second condition of Theorems~\ref{theorem: superlum_bipart} and~\ref{theorem: superlum_tripart} are not necessary for ruling out causal loops in Bell scenarios) (e) Space-time embedding of the observed variables (circled) in all of these examples. We recall that in causal structures, a dashed causal arrow $N_1\xdashrightarrow{} N_2$ as well as a solid causal arrow $N_1\longrightarrow N_2$ denote that $N_1$ is a direct cause of $N_2$, but the dashed arrow indicates that $N_1$ does not affect $N_2$ while the solid arrow indicates that $N_1$ affects $N_2$ in the given causal model.}
    \label{fig: examples_appendix}
\end{figure}

\begin{example}[\texorpdfstring{\hyperref[cond: NS2]{(NS2)}/\hyperref[cond: NS3]{(NS3)}}{NS2/3} are insufficient for ruling out affects causal loops in Bell scenarios] 
\label{example: affects_loop}
Consider a cyclic causal structure $\cG$ where $X\longrsquigarrow Y$ and $Y\longrsquigarrow X$ while $A$ and $B$ are parentless and childless nodes (Figure~\ref{fig: affects_loop}). Consider a causal model on this causal structure where $A$ and $B$ have arbitrary cardinality and are distributed according to some distributions $P(A)$ and $P(B)$, $X$ and $Y$ are trits ($\in \{0,1,2\}$) with $X=2Y$ mod $3$ and $Y=X$ being the causal mechanisms in the loop. Then the only consistent solution in the loop is $X=Y=0$. Moreover, since $A$ and $B$ are causally disconnected from each other and $X$ and $Y$, this means that the observed distribution of this example is as follows (due to the $d$-separation property and the above unique solution for the loop equations)
\begin{equation}
        P_{\cG}(XYAB)=P_{\cG}(XY)P(A)P(B),
    \end{equation}
    where $P_{\cG}(XY)$ is the deterministic distribution on $X=Y=0$. Then the non-signalling conditions \hyperref[cond: NS2]{(NS2)} are trivially satisfied. Moreover, the post-intervention graphs $\cG_{\mathrm{do}(X)}$ and $\cG_{\mathrm{do}(Y)}$ are obtained from $\cG$ by removing the $Y\longrsquigarrow X$ and $X\longrsquigarrow Y$ edges respectively. As this is a classical causal model defined in terms of causal mechanisms, we can also specify the post-intervention models in terms of the original causal mechanisms~\cite{Pearl2009, Forre2017}, this preserves the causal mechanisms of the edges that are not removed and fixes the variable being intervened on to some value. Doing so, we see that $P_{\cG_{\mathrm{do}(X)}}(Y=1|X=1)=1$ since we have the $Y=X$ causal mechanism here, while $P_{\cG}(Y=1)=0$, this means that $X$ affects $Y$. Similarly, we $P_{\cG_{\mathrm{do}(Y)}}(X=1|Y=2)=1$ since we have the $X=2Y$ mod $3$ causal mechanism here, while $P_{\cG}(X=1)=0$, which means that we also have $Y$ affects $X$. These affects relations form an affects causal loop as they allow us to infer a causal cycle between the RVs $X$ and $Y$. The same example can be extended to establish the insufficiency of \hyperref[cond: NS3]{(NS3)} (and hence the weaker condition \hyperref[cond: NS3p]{(NS3$'$)}) for ruling out affects causal loops in tripartite Bell scenarios.  
\end{example}

\begin{example}[\texorpdfstring{\hyperref[cond: NS2]{(NS2)}/\hyperref[cond: NS3]{(NS3)}}{NS2/3} are insufficient for ruling out hidden causal loops in Bell scenarios]
\label{example: hidden_loop}
This example was first given in~\cite{VilasiniColbeckPRA}. Consider a causal structure where $X\longrsquigarrow Y$ and $Y\longrsquigarrow X$, together with a common cause $\Lambda$, with $\Lambda\longrsquigarrow X$ and $\Lambda\longrsquigarrow Y$ (Figure~\ref{fig: hidden_loop}). Suppose all nodes correspond to classical and binary variables, $X$ and $Y$ are observed, $\Lambda$ is unobserved and the causal model is such that $\Lambda$ is uniformly distributed, $X=Y\oplus \Lambda$, $Y=X\oplus \Lambda$. One can verify that $X$ and $Y$ are uncorrelated, $X$ does not affect $Y$ and $Y$ does not affect $X$ (using the methods explained in the previous example). However we have a causal loop between $X$ and $Y$ by construction. This is a hidden causal loop and not an affects causal loops because the the observable affects relations of the model are trivially replicated in an acyclic model where $X$ and $Y$ have no causal relations. Moreover this model trivially satisfies the non-signalling conditions \hyperref[cond: NS2]{(NS2)}, \hyperref[cond: NS3]{(NS3)} and \hyperref[cond: NS3p]{(NS3$'$)} when we introduce dummy variables corresponding to the settings (which does not change any of the above properties).
\end{example}

The following counter-examples prove the latter statement of Theorem~\ref{THEOREM: LOOPSBELL2}.

\begin{example}[\texorpdfstring{\hyperref[cond: NS2]{(NS2)}/\hyperref[cond: NS3]{(NS3)}}{NS2/3} are not necessary for ruling out causal loops in Bell scenarios]
\label{example: non-necc1}
Consider a causal structure of a Bell scenario where $A\longrsquigarrow Y$ is the only causal influence, with $Y=A$ (Figure~\ref{fig: non-necc1}). This would violate \hyperref[cond: NS2]{(NS2)}, \hyperref[cond: NS3]{(NS3)} and \hyperref[cond: NS3p]{(NS3$'$)} since Alice's (freely chosen) setting is correlated with Bob's outcome enabling them to communicate. However, this causal structure is acyclic and therefore has no causal loops (neither affects nor hidden loops). 
\end{example}

\begin{example}[Second conditions of Theorems~\ref{theorem: superlum_bipart} and~\ref{theorem: superlum_tripart} are not necessary for ruling out causal loops in Bell scenarios]
\label{example: non-necc2}
Consider a causal structure of a Bell scenario where $X\longrsquigarrow Y$ is the only causal influence, with $Y=X$ (Figure~\ref{fig: non-necc2}). This would trivially satisfy \hyperref[cond: NS2]{(NS2)}, \hyperref[cond: NS3]{(NS3)} and \hyperref[cond: NS3p]{(NS3$'$)} as the settings are causally disconnected to all outcomes, and moreover is an acyclic causal structure and hence has no causal loops of any kind. However, we have $X$ affects $Y$ which violates conditions 2 of Theorems~\ref{theorem: superlum_bipart} and~\ref{theorem: superlum_tripart} as well as condition~\ref{cond:tri2p} of Theorem~\ref{theorem: superlum_tripart}. 
\end{example}

\section{Space-time positions compatible with jamming}\label{app:ST}
In this section we outline which space-time locations are compatible with jamming and how this depends on the number of spatial dimensions. To do so we discuss when the intersection of two lightcones is contained within a third for arbitrary space-time points.  This can be applied to jamming with the notation of the rest of the paper by taking the space-time location of $X$ to be $({\bf x}_A,t_A)$, that of $B$ to be $({\bf x}_B,t_B)$ and that of $Z$ to be $({\bf x}_C,t_C)$.  We begin with some notation.

\subsection{Lightcone notation}

Given an event $({\bf x}_A,t_A)$, its associated \emph{lightcone} $L({\bf x}_A,t_A)$ is the set of points satisfying $({\bf x}-{\bf x}_A).({\bf x}-{\bf x}_A)-(t-t_A)^2\leq0$, or $|{\bf x}-{\bf x}_A|^2\leq(t-t_A)^2$ (we take the speed of light to be $1$ for convenience). Its \emph{future lightcone}, denoted $L_>({\bf x}_A,t_A)$ is the subset of the lightcone with $t>t_A$.  We will be interested in the intersection of lightcones.  Firstly, a small observation: the future lightcone of any point in a lightcone remains in that lightcone.
\begin{lemma}\label{lem:1}
  Given an event $({\bf x}_A,t_A)$ and a second event in its future lightcone $({\bf x}_B,t_B)\in L_>({\bf x}_A,t_A)$, we have $L_>({\bf x}_B,t_B)\subseteq L_>({\bf x}_A,t_A)$.
\end{lemma}
\begin{proof}
Since $({\bf x}_B,t_B)\in L_>({\bf x}_A,t_A)$, we have $|{\bf x}_B-{\bf x}_A|^2\leq (t_B-t_A)^2$. Let $({\bf x},t)\in L_>({\bf x}_B,t_B)$. It follows that $|{\bf x}-{\bf x}_B|\leq(t-t_B)^2$ with $t_A<t_B<t$. Then, using the triangle inequality,
  \begin{eqnarray*}
    |{\bf x}-{\bf x}_A|&\leq&|{\bf x}-{\bf x}_B|+|{\bf x}_B-{\bf x}_A|\\
                       &\leq&(t-t_B)+(t_B-t_A)=t-t_A\,,
  \end{eqnarray*}
which implies $({\bf x},t)\in L_>({\bf x}_A,t_A)$. Since $({\bf x},t)$ was an arbitrary point in $L_>({\bf x}_B,t_B)$ it follows that $L_>({\bf x}_B,t_B)\subseteq L_>({\bf x}_A,t_A)$.
\end{proof}
We use the notation $L_>({\bf x}_A,t_A)\big|_t:=\{{\bf x}:({\bf x},t)\in L_>({\bf x}_A,t_A)\}$, for the set of possible spatial coordinates of points in the lightcone $L_>({\bf x}_A,t_A)$ with fixed time coordinate $t$.

\subsection{\texorpdfstring{$(1+1)$D}{1p1D}}
In $(1+1)$D the intersection of two future lightcones is itself a lightcone, eminating from a unique point of earliest interception.
\begin{lemma}\label{lem:2}
  Given two events $(x_A,t_A)$ and $(x_C,t_C)$ in $(1+1)$D, there is a unique point $(x',t')$ of earliest interception of their future lightcones, and $L_>(x_A,t_A)\cap L_>(x_C,t_C)=L_>(x',t')$.
\end{lemma}
\begin{proof}
  Suppose first that $(x_C,t_C)\in L_>(x_A,t_A)$. By Lemma~\ref{lem:1}, $L_>(x_C,t_C)\subseteq L_>(x_A,t_A)$, and hence $L_>(x_A,t_A)\cap L_>(x_C,t_C)=L_>(x_C,t_C)$, i.e., the statement holds with $(x',t')=(x_C,t_C)$. The case $(x_A,t_A)\in L_>(x_C,t_C)$ is analogous.
  
  In the alternative case, without loss of generality, assume $x_A<x_C$ and $t_A<t_C$.  To not be in the first case we require $x_C-x_A>t_C-t_A$. The edges of the lightcone $L_>(x_A,t_A)$ are given by $x=x_A\pm(t-t_A)$, and those of $L_>(x_C,t_C)$ are given by $x=x_C\pm(t-t_C)$. The only two that meet in the future are $x=x_A+t-t_A$ and $x=x_C-t+t_C$, which meet at $x'=(x_C+x_A+t_C-t_A)/2$ and $t'=(x_C-x_A+t_C+t_A)/2$. Any event with a time coordinate smaller than this cannot be in both lightcones, and since $(x',t')$ is in both lightcones, by Lemma~\ref{lem:1} it follows that $L_>(x',t')\subseteq L_>(x_A,t_A)\cap L_>(x_C,t_C)$.

  Now consider an arbitrary point $(x,t)$ in $L_>(x_A,t_A)\cap L_>(x_C,t_C)$. Such a point must satisfy $x\leq x_A+t-t_A$ and $x\geq x_C-t+t_C$. Thus,
  \begin{eqnarray*}
    x-x'&\leq& x_A+t-t_A-x'=t-(x_C-x_A+t_A+t_C)/2=t-t'\\
    x-x'&\geq& x_C-t+t_C-x'=-(t-(x_C-x_A+t_A+t_C)/2)=-(t-t')\,,
  \end{eqnarray*}
  which imply $|x-x'|\leq t-t'$, i.e., $(x,t)\in L_>(x',t')$.

  Putting all these together, we have $L_>(x_A,t_A)\cap L_>(x_C,t_C)=L_>(x',t')$, as required.
 \end{proof}

Next we show that consideration of a single time coordinate is enough to establish whether the intersection of two lightcones is contained in a third.

 \begin{lemma}\label{lem:3}
   Consider three arbitrary events $(x_A,t_A)$, $(x_B,t_B)$ and $(x_C,t_C)$.
   \begin{enumerate}
   \item If for some time $t>t_B$ we have $L_>(x_A,t_A)\big|_t\cap L_>(x_C,t_C)\big|_t\subseteq L_>(x_B,t_B)\big|_t$, then $L_>(x_A,t_A)\cap L_>(x_C,t_C)\subseteq L_>(x_B,t_B)$.
   \item If for some time $t>t_B$ we have $L_>(x_B,t_B)\big|_t\subseteq L_>(x_A,t_A)\big|_t\cap L_>(x_C,t_C)\big|_t$, then $L_>(x_B,t_B)\subseteq L_>(x_A,t_A)\cap L_>(x_C,t_C)$.
   \end{enumerate}
 \end{lemma}

 \begin{proof}
   Suppose $(x_A,t_A)$ and $(x_C,t_C)$ are spacelike separated and, without loss of generality, suppose $x_A<x_C$ and $t_A<t_C$ so that the condition of spacelike separation is $x_C-x_A>t_C-t_A$.

   The set of points in $L_>(x_A,t_A)\big|_t\cap L_>(x_C,t_C)\big|_t$ is given by $x\in[x_C-t+t_C,x_A+t-t_A]$, and the set of points in $L_>(x_B,t_B)\big|_t$ is given by $x\in[x_B-t+t_B,x_B+t-t_B]$ and we assume this set is not empty. For $L_>(x_A,t_A)\big|_t\cap L_>(x_C,t_C)\big|_t\subseteq L_>(x_B,t_B)\big|_t$ we require $x_C-t+t_C\leq x_B-t+t_B$ and $x_B+t-t_B\leq x_A+t-t_A$. However, if these hold for a particular value of $t$, they hold for all $t$ (since $t$ cancels from both equations).  Likewise for the other direction.  This proves the case where $(x_A,t_A)$ and $(x_C,t_C)$ are spacelike separated.  The case where they are not also follows analogously.
 \end{proof}

 Applying Lemma~\ref{lem:3} with $t=t'$ we have the corollary that the intersection of $L_>(x_A,t_A)$ and $L_>(x_C,t_C)$ is contained in $L_>(x_B,t_B)$ if and only if $(x',t')\in L_>(x_B,t_B)$.

\subsection{\texorpdfstring{$(\di+1)$D for $\di\geq2$}{2p1D}}
With more than one spatial dimension it is no longer the case that the intersection of two lightcones is a lightcone.  [For instance, in $(2+1)$D, for a fixed time a lightcone is a circle, but the intersection of two circles is not a circle.] The analogue of Lemma~\ref{lem:2} hence does not hold in higher dimensions.

Instead of Lemma~\ref{lem:3}, the following lemma shows that consideration of a single time coordinate is not sufficient to establish that the intersection of two lightcones is contained within a third.

Before this we introduce two pieces of notation: ${\bf 0}_\di$ is shorthand for $\di$ zeros, and ${\bf x}_{\di-n}$ is the subvector of ${\bf x}$ with the first $n$ components removed, e.g., in dimension $\di=3$, with ${\bf x}=(x,y,z)$, ${\bf x}_2$ means $(y,z)$.

\begin{lemma}\label{lem:4}
There exist three events $({\bf x}_A,t_A)$, $({\bf x}_B,t_B)$ and $({\bf x}_C,t_C)$, and two times $t$ and $t'>t$ such that $L_>({\bf x}_A,t_A)\big|_t\cap L_>({\bf x}_C,t_C)\big|_t\subseteq L_>({\bf x}_B,t_B)\big|_t$ but $L_>({\bf x}_A,t_A)\big|_{t'}\cap L_>({\bf x}_C,t_C)\big|_{t'}\nsubseteq L_>({\bf x}_B,t_B)\big|_{t'}$.
\end{lemma}
\begin{proof}
  We can show this with an example: take $({\bf x}_A,t_A)=((-a,{\bf 0}_{\di-1}),0)$, $({\bf x}_C,t_C)=((a,{\bf 0}_{\di-1}),0)$ and $({\bf x}_B,t_B)=((0,{\bf 0}_{\di-1}),a/2)$. The set $L_>({\bf x}_A,t_A)\big|_t\cap L_>({\bf x}_C,t_C)\big|_t$ is then nonempty if and only if $t>a$. For $t>a$, the points in $L_>({\bf x}_A,t_A)\big|_t\cap L_>({\bf x}_C,t_C)\big|_t$ furthest from ${\bf x}={\bf 0}_{\di}$ are the maxima of $|{\bf x}|^2$ subject to $|{\bf x-x}_A|^2\leq t^2$ and $|{\bf x-x}_C|^2\leq t^2$. The objective function has no global maximum, so the maximum must occur on the boundary. The first boundary, $(x+a)^2+|{\bf x}_{\di-1}|^2=t^2$, meets the second, $(x-a)^2+|{\bf x}_{\di-1}|^2=t^2$, when $x=0$ and $|{\bf x}_{\di-1}|^2=t^2-a^2$.  On the first boundary we have $x^2+|{\bf x}_{\di-1}|^2=t^2-2ax-a^2$, and $x\geq0$ to also be inside the second region. Thus, the maximum value of $x^2+|{\bf x}_{\di-1}|^2$ is $t^2-a^2$, which occurs at $x=0$ and $|{\bf x}_{\di-1}|^2=t^2-a^2$.

  The points with $x=0$ and $|{\bf x}_{\di-1}|^2=t^2-a^2$ are inside the region $L_>({\bf x}_B,t_B)\big|_t$ if and only if $t^2-a^2\leq (t-t_B)^2$, which for $t_B=a/2$ is equivalent to $t\leq5a/4$. Hence, for $1<t\leq5a/4$ we have $L_>({\bf x}_A,t_A)\big|_t\cap L_>({\bf x}_C,t_C)\big|_t\subseteq L_>({\bf x}_B,t_B)\big|_t$.  However, if we take $t'>5a/4$ the points with $x=0$ and  $|{\bf x}_{\di-1}|^2=(t')^2-a^2$ are inside $L_>({\bf x}_A,t_A)\big|_{t'}\cap L_>({\bf x}_C,t_C)\big|_{t'}$ but outside $L_>({\bf x}_B,t_B)\big|_{t'}$.
\end{proof}

We would like to know when the intersection of two lightcones is contained within a third. This breaks down into two cases treated separately in the following two lemmas (the first of which is most relevant for jamming).

\begin{lemma}\label{lem:8}
Let $({\bf x}_A,t_A)=((-a,{\bf 0}_{\di-1}),0)$, $({\bf x}_C,t_C)=((a,{\bf 0}_{\di-1}),0)$ and $({\bf x}_B,t_B)=((x_B,y_B,{\bf 0}_{\di-2}),t_B)$ for some $-a\leq x_B\leq a$. $L_>({\bf x}_A,t_A)\cap L_>({\bf x}_C,t_C)\subseteq L_>({\bf x}_B,t_B)$ if and only if $t_B\leq-|y_B|$.
\end{lemma}
Note that the form of the chosen points is not losing generality: for any two spacelike separated events we can find a frame where they have coordinates of the form given for $({\bf x}_A,t_A)$ and $({\bf x}_C,t_C)$. Similarly, because any three points lie on a 2D plane we can also choose $({\bf x}_B,t_B)$ to have the form given.

In addition, the condition $t_B\leq-|y_B|$ is equivalent to saying that an event with $t=0$ on the line between ${\bf x}_A$ and ${\bf x}_B$ lies in the future lightcone of $({\bf x}_B,t_B)$.

\begin{proof}
  For $t>a$, the part of the boundary of $L_>({\bf x}_A,t_A)\big|_t\cap L_>({\bf x}_C,t_C)\big|_t$ that is the same as the boundary of $L_>({\bf x}_A,t_A)\big|_t$, is given by $x=\sqrt{t^2-|{\bf x}_{\di-1}|^2}-a$, where $|{\bf x}_{\di-1}|^2\leq t^2-a^2$. [Note that the two lightcones meet when $x=0$ and $|{\bf x}_{\di-1}|^2=t^2-a^2$.]
  
We want to calculate the maximum distance a point on this boundary could be from ${\bf x}_B$.  The square of this distance is
  \begin{eqnarray}
    (x-x_B)^2+(y-y_B)^2+|{\bf x}_{\di-2}|^2&=&x^2+|{\bf x}_{\di-1}|^2+x_B^2+y_B^2-2xx_B-2yy_B\nonumber\\
                      &=&t^2-2ax-a^2+x_B^2+y_B^2-2xx_B-2yy_B\nonumber\\
                       &=&t^2+x_B^2+y_B^2-2yy_B-a^2-2(a+x_B)(\sqrt{t^2-|{\bf x}_{\di-1}|^2}-a)\nonumber\\
    &=&t^2+x_B^2+y_B^2-2yy_B+a^2+2ax_B-2(a+x_B)\sqrt{t^2-y^2-|{\bf x}_{\di-2}|^2}\,.\label{eq:lastline3}
  \end{eqnarray}

  Without loss of generality we can take $x_B\leq0$ and $y_B\geq0$.  For $-a\leq x_B\leq0$ the maximum distance occurs when $y=-\sqrt{t^2-a^2}$, $|{\bf x}_{\di-2}|^2=0$ and hence $x=0$.  This square distance is 
\begin{equation}\label{eq:simp3}
  x_B^2+\left(y_B+\sqrt{t^2-a^2}\right)^2.
\end{equation}
Hence $L_>({\bf x}_A,t_A)\big|_t\cap L_>({\bf x}_C,t_C)\big|_t\subseteq L_>({\bf x}_B,t_B)\big|_t$ if and only if $x_B^2+(y_B+\sqrt{t^2-a^2})^2\leq(t-t_B)^2$, or equivalently,
  \begin{equation}\label{eq:cond3}
    2tt_B+2y_B\sqrt{t^2-a^2}\leq t_B^2+a^2-x_B^2-y_B^2.
  \end{equation}
  We want to find values of $t_B$ such that this holds for all $t>a$. For large enough $t$ the left hand side approaches $2t(t_B+y_B)$. Hence, if $t_B>-y_B$ there exists a sufficiently large $t$ such that~\eqref{eq:cond3} ceases to hold. Noting that we are in the case $-a\leq x_B\leq0$, so $a^2-x_B^2\geq0$, if $t_B\leq-y_B$, the right hand side of~\eqref{eq:cond3} is always positive. In addition,
  \begin{align*}
    2tt_B+2y_B\sqrt{t^2-a^2}\leq 2t(t_B+y_B)\leq 0\,,
  \end{align*}
and hence~\eqref{eq:cond3} holds for all $t$.
\end{proof}

\begin{lemma}
Let $({\bf x}_A,t_A)=((-a,{\bf 0}_{\di-1}),0)$, $({\bf x}_C,t_C)=((a,{\bf 0}_{\di-1}),0)$ and $({\bf x}_B,t_B)=((x_B,y_B,{\bf 0}_{\di-2}),t_B)$ for some $|x_B|\geq a$. $L_>({\bf x}_A,t_A)\cap L_>({\bf x}_C,t_C)\subseteq L_>({\bf x}_B,t_B)$ if and only if at least one of  $({\bf x}_A,t_A)\in L_>({\bf x}_B,t_B)$ and $({\bf x}_C,t_C)\in L_>({\bf x}_B,t_B)$ holds.
\end{lemma}
\begin{proof}
  Consider the case $x_B\leq-a$ (the other case follows symmetrically) and take $y_B\geq0$ without loss of generality. We want to maximize~\eqref{eq:lastline3} with respect to ${\bf x}_{\di-1}$ subject to $|{\bf x}_{\di-1}|^2\leq t^2-a^2$.  Writing $y=r\cos(\phi)$, i.e., $r=|{\bf x}_{\di-1}|$ we can rewrite~\eqref{eq:lastline3} as
  \begin{equation}\label{eq:r}
    t^2+x_B^2+y_B^2-2y_Br\cos(\phi)+a^2+2ax_B-2(a+x_B)\sqrt{t^2-r^2}.
    \end{equation}
    The maximum of this over $\phi$ is when $\cos(\phi)=-1$.

The derivative of~\eqref{eq:r} with respect to $r$ is $2r(a+x_B)/\sqrt{t^2-r^2}-2y_B\cos(\phi)$, which, substituting $\cos(\phi)=-1$ is equal to zero for $r=y_Bt/\sqrt{(a+x_B)^2+y_B^2}$.  Substituting this into the distance squared gives
  \begin{align}\label{eq:max_dist}
(a+x_B)^2\left(1+\frac{t}{\sqrt{(a+x_B)^2+y_B^2}}\right)^2+y_B^2\left(1+\frac{t}{\sqrt{(a+x_B)^2+y_B^2}}\right)^2=\left(t+\sqrt{(a+x_B)^2+y_B^2}\right)^2.
  \end{align}
  This is the maximum provided the value of $r$ is within the range of possible $r$ values (on which the boundary of $L_>({\bf x}_A,t_A)\big|_t\cap L_>({\bf x}_C,t_C)\big|_t$ is the same as the boundary of $L_>({\bf x}_A,t_A)\big|_t$), i.e., if
  \begin{align*}
    \frac{y_Bt}{\sqrt{(a+x_B)^2+y_B^2}}\leq\sqrt{t^2-a^2},
  \end{align*}
  which rearranges to
  \begin{align*}
    t^2\geq\frac{a^2y_B^2}{(a+x_B)^2}+a^2.
  \end{align*}
  Hence, for large enough $t$, Eq.~\eqref{eq:max_dist} gives the maximum distance.

If $t^2<\frac{a^2y_B^2}{(a+x_B)^2}+a^2$, then the maximum occurs on the boundary, i.e., where $r^2=t^2-a^2$. Exactly as in the proof of Lemma~\ref{lem:8}, the maximum square distance in this case is $x_B^2+(y_B+\sqrt{t^2-a^2})^2$. Hence $L_>({\bf x}_A,t_A)\big|_t\cap L_>({\bf x}_C,t_C)\big|_t\subseteq L_>({\bf x}_B,t_B)\big|_t$
  requires
  \begin{align*}
    \begin{cases}
      x_B^2+(y_B+\sqrt{t^2-a^2})^2\leq(t-t_B)^2 & t\leq\sqrt{\frac{a^2y_B^2}{(a+x_B)^2}+a^2}\\
\left(t+\sqrt{(a+x_B)^2+y_B^2}\right)^2\leq(t-t_B)^2 & t\geq\sqrt{\frac{a^2y_B^2}{(a+x_B)^2}+a^2}
    \end{cases}.
  \end{align*}
For the second case to hold for all $t$ we require $t_B\leq-\sqrt{(a+x_B)^2+y_B^2}$.  This implies $t_B\leq-y_B$ which means that the first case also holds as shown in the proof of Lemma~\ref{lem:8}. Noting that $t_B\leq-\sqrt{(a+x_B)^2+y_B^2}$ is equivalent to $({\bf x}_A,t_A)\in L_>({\bf x}_B,t_B)$ completes the proof.
\end{proof}


\end{document}